\documentclass[11pt,english]{article}

\usepackage[utf8]{inputenc}
\usepackage[letterpaper,margin=1.00in]{geometry}
\usepackage{amsmath, amssymb, amsthm, amsfonts}
\usepackage{xspace}
\usepackage[linktocpage]{hyperref}
\usepackage[dvipsnames]{xcolor}

\usepackage[utf8]{inputenc}
\usepackage{amsmath,amsfonts,amssymb,amsthm}
\usepackage[sans]{dsfont}
\usepackage{mathtools}
\usepackage{url}
\usepackage{appendix}
\usepackage{algorithm}
\usepackage{fullpage}
\usepackage{bbm}

\definecolor{ForestGreen}{rgb}{0.1333,0.5451,0.1333}
\definecolor{DarkRed}{rgb}{0.8,0,0}
\definecolor{Red}{rgb}{1,0,0}

\hypersetup{
    unicode=false,          
    colorlinks=true,        
    linkcolor=DarkRed,          
    citecolor=ForestGreen,        
    filecolor=magenta,      
    urlcolor=cyan           
}
\usepackage{thm-restate}
\usepackage[capitalize, nameinlink]{cleveref}

\newtheorem{theorem}{Theorem}[section]
\newtheorem{thm}[theorem]{Theorem}
\newtheorem{lemma}[theorem]{Lemma}
\newtheorem{lem}[theorem]{Lemma}
\newtheorem{claim}[theorem]{Claim}\crefname{claim}{claim}{claims}
\newtheorem{remark}[theorem]{Remark}
\newtheorem{corollary}[theorem]{Corollary}
\newtheorem{cor}[theorem]{Corollary}

\newtheorem{prop}[theorem]{Proposition}

\newtheorem{definition}[theorem]{Definition}
\newtheorem{defn}[theorem]{Definition}
\newtheorem{assumption}[theorem]{Assumption}
\newtheorem{fact}[theorem]{Fact}

\newtheorem{comment}[theorem]{Comment}

\newcounter{note}[section]

\newcommand{\eps}{\varepsilon}

\newcommand{\congest}{$\mathsf{CONGEST}$\xspace}

\newcommand{\poly}{\operatorname{\text{{\rm poly}}}}

\newcommand{\ceil}[1]{\lceil #1 \rceil}

\newcommand{\N}{\mathbb{N}}

\newcommand{\diam}{\operatorname{diam}}

\newcommand{\sep}{\operatorname{sep}}
\newcommand{\load}{\operatorname{load}}

\newcommand{\cond}{\operatorname{cond}}
\newcommand{\dist}{\operatorname{dist}}

\newcommand{\ball}{\operatorname{ball}}

\global\long\def\Otil{\tilde{O}}%
\global\long\def\poly{\mathrm{poly}}%
\global\long\def\eps{\epsilon}%
\global\long\def\N{\mathcal{N}}%
\global\long\def\cN{{\cal N}}%
\global\long\def\cS{{\cal S}}%
\global\long\def\l{\ell}%
\global\long\def\congest{\mathrm{cong}}%
\global\long\def\step{\mathrm{step}}%
\global\long\def\cond{\mathrm{cond}}%
\global\long\def\spars{\mathrm{spars}}%
\global\long\def\sep{\mathrm{sep}}%
\global\long\def\val{\mathrm{val}}%
\global\long\def\dist{\mathrm{dist}}%
\global\long\def\deg{\mathrm{deg}}%
\global\long\def\diam{\mathrm{diam}}%
\global\long\def\cov{\mathrm{cov}}%
\global\long\def\ball{\mathrm{ball}}%
\global\long\def\cluster{\mathrm{cluster}}%
\global\long\def\router{\mathrm{router}}%
\global\long\def\supp{\mathrm{supp}}%
\global\long\def\Fhat{Cidehat{F}}%
\global\long\def\fhat{Cidehat{f}}%

\renewcommand{\l}{\ell}

\renewcommand{\poly}{\operatorname{poly}}

\newcommand{\myparskip}{3pt}
\parskip \myparskip
\usepackage{verbatim}
\usepackage{float}
\usepackage{amsmath}
\usepackage{amsthm}
\usepackage{amssymb}
\usepackage{tikz}
\usetikzlibrary{shapes.geometric, arrows, positioning, calc}

\makeatletter

\floatstyle{ruled}
\newfloat{algorithm}{tbp}{loa}
\providecommand{\algorithmname}{Algorithm}
\floatname{algorithm}{\protect\algorithmname}

\usepackage{xcolor}
\usepackage{nameref}

\definecolor{ForestGreen}{rgb}{0.1333,0.5451,0.1333}
\definecolor{DarkRed}{rgb}{0.8,0,0}
\definecolor{Red}{rgb}{1,0,0}

\makeatother

\global\long\def\Otil{\tilde{O}}

\global\long\def\poly{\mathrm{poly}}

\global\long\def\eps{\epsilon}

\global\long\def\N{\mathcal{N}}
\global\long\def\cN{\mathcal{N}}

\global\long\def\l{\ell}
\global\long\def\congest{\mathrm{cong}}
\global\long\def\leng{\mathrm{leng}}

\global\long\def\step{\mathrm{step}}
\global\long\def\cond{\mathrm{cond}}
\global\long\def\spars{\mathrm{spars}}
\global\long\def\sep{\mathrm{sep}}
\global\long\def\val{\mathrm{val}}
\global\long\def\dist{\mathrm{dist}}
\global\long\def\diam{\mathrm{diam}}
\global\long\def\cov{\mathrm{cov}}
\global\long\def\ball{\mathrm{ball}}
\global\long\def\load{\mathrm{load}}

\global\long\def\decomp{\mathrm{dcp}}
\global\long\def\cluster{\mathrm{cluster}}
\global\long\def\router{\mathrm{router}}
\global\long\def\rt{\mathrm{rt}}
\global\long\def\supp{\mathrm{supp}}
\global\long\def\Fhat{\widehat{F}}
\global\long\def\fhat{\widehat{f}}
\global\long\def\S{{\cal S}}
\global\long\def\cS{{\cal S}}
\global\long\def\ceil#1{\left\lceil #1\right\rceil }
\global\long\def\done{\mathrm{done}}
\global\long\def\forward{\mathrm{forward}}

\global\long\def\path{\mathrm{path}}
\global\long\def\flow{\mathrm{flow}}

\global\long\def\Fhat{\widehat{F}}

\global\long\def\aug{\mathrm{aug}}
\global\long\def\avglen{\mathrm{avglen}}
\global\long\def\totlen{\mathrm{totlen}}
\global\long\def\supportsize{(|E|+\supp(D)) N^{\poly\eps}}

\global\long\def\supportsizet{(|E|+\supp(D)) N^{\poly\eps}}

\newcommand{\thatchaphol}[1]{\textcolor{purple}{(T: #1)}}

\begin{document}
\date{}
\title{Low-Step Multi-Commodity Flow Emulators}
\author{Bernhard Haeupler\thanks{Partially funded by the European Union's Horizon 2020 ERC grant 949272.} \\
\texttt{bernhard.haeupler@inf.ethz.ch}\\
INSAIT \& ETH Z\"urich\\
\and D Ellis Hershkowitz\\ \texttt{delhersh@brown.edu} \\ Brown University\\\and Jason Li \\ \texttt{jmli@cs.cmu.edu} \\ CMU\\ \and Antti Roeyskoe\footnotemark[1] \\ \texttt{antti.roeyskoe@inf.ethz.ch} \\ ETH Z\"urich \\ \and Thatchaphol Saranurak\thanks{Supported by NSF grant CCF-2238138.}\\ \texttt{thsa@umich.edu} \\ University of Michigan\\}

\maketitle

\begin{abstract}

We introduce the concept of \emph{low-step multi-commodity flow emulators} for any undirected, capacitated graph. At a high level, these emulators contain approximate multi-commodity flows whose paths contain a small number of edges, shattering the infamous \emph{flow decomposition barrier} for multi-commodity flow.

We prove the existence of low-step multi-commodity flow emulators and develop efficient algorithms to compute them. We then apply them to solve constant-approximate $k$-commodity flow in $O((m+k)^{1+\epsilon})$ time. To bypass the $O(mk)$ flow decomposition barrier, we represent our output multi-commodity flow \emph{implicitly}; prior to our work, even the existence of implicit constant-approximate multi-commodity flows of size $o(mk)$ was unknown.

Our results generalize to the \emph{minimum cost} setting, where each edge has an associated cost and the multi-commodity flow must satisfy a cost budget. Our algorithms are also parallel.

%
%
%
%
\end{abstract}
\thispagestyle{empty}

\newpage
\thispagestyle{empty}
\tableofcontents
\thispagestyle{empty}
\newpage

\setcounter{page}{1}

\section{Introduction}
In the maximum flow problem, we are given an edge-capacitated graph $G=(V,E)$ and two vertices $s,t\in V$ with the aim to send as much flow as possible from $s$ to $t$. Maximum flow is a fundamental problem in combinatorial optimization with a long history, from the classic Ford-Fulkerson algorithm \cite{cormen2022introduction} to the recent breakthrough almost-linear time algorithm of Chen~et~al.~\cite{chen2022maximum}. The maximum flow problem also exhibits rich structural properties, most notably in the \emph{max-flow min-cut theorem}, which states that the value of the maximum flow between $s$ and $t$ is equal to the size of the minimum edge cut separating $s$ and $t$.

A well-studied generalization of the maximum flow problem is the maximum \emph{multi-commodity flow} problem. Here, we are given $k\ge1$ pairs $(s_i,t_i)_{i\in[k]}$ of vertices and wish to maximize the amount of flow sent between each pair. 
On the algorithmic side, this problem can be solved in polynomial time by a linear program, and there has been exciting recent progress towards obtaining faster algorithms, both exact~\cite{van2023faster} and $(1+\epsilon)$-approximate~\cite{sherman2017area}. However, these algorithms output the flow for each commodity explicitly and, hence, must take at least $\Omega(mk)$ time because there exists a graph such that the total size of flow representation overall $k$ commodities of any $\alpha$-approximate solution has size at least $\Omega(mk/\alpha)$.
Designing multi-commodity flow algorithms is further complicated by the loss of structure exhibited in the single-commodity setting. Specifically, the max-flow min-cut equality no longer holds in the multi-commodity case, and the flow-cut gap (i.e.\ the multiplicative difference between the max-flow and min-cut) is known to be $\Theta(\log n)$ for undirected graphs~\cite{leighton1999multicommodity} and $\Omega(n^{1/7})$ for directed graphs~\cite{chuzhoy2009polynomial}.

The algorithmic and structural issues above suggest that outputting a (near) optimal multi-commodity flow in time less than $O(mk)$ may be impossible. Even for the problem of efficiently approximating \emph{the value} of the optimal multi-commodity flow in undirected graphs, current techniques (e.g.~\cite{racke2002minimizing,racke2014computing}) fail to produce approximations below the multi-commodity flow-cut gap, i.e., no $o(\log n)$-approximations running in $o(mk)$ time are known.


\subsection{Our Results}

In this paper, we break the $O(mk)$-time barrier by giving $(m+k)^{1+\epsilon}$-time algorithms with $O(1)$-approximations to the value of the maximum $k$-commodity flow on undirected graphs for any constant $\eps > 0$.

We consider both concurrent and non-concurrent flow problems. 
In the concurrent flow problem, given a capacitated graph and a demand function, the goal is to find a maximum-value capacity-respecting flow routing a multiple of the demand. In the non-concurrent flow problem, given a capacitated graph and a set of vertex pairs, the goal is to find a maximum-value capacity-respecting flow routing flow only between vertex pairs in the given set. 

Our results can be informally summarized as follows: 

\begin{theorem}[Constant-Approximate Concurrent/Non-Concurrent Flow (informal)]\label{thm:informal-constapprox-cnc-flow}
    For every constant $\epsilon \in (0, 1)$, there exists a $(m + k)^{1 + \poly(\eps)}$-time $O(2^{-1 / \eps})$-approximate algorithm for the concurrent and non-concurrent multicommodity flow value problems, where $k$ is the number of demand pairs. The algorithms work in parallel with depth $(m + k)^{\poly(\eps)}$. 
\end{theorem}


The key to the above result is a powerful new tool we introduce called \emph{low-step (multi-commodity) flow emulators}. Informally, a low-step flow emulator of an undirected graph is another graph which contains approximate multi-commodity flows whose flow paths contain a small number of edges. Because they are not based on cuts, such emulators face no $\Omega(\log n)$ flow-cut barriers, unlike the above-mentioned cut sparsifiers. It is instructive to view low-step flow emulators as a generalization of \emph{hopsets}, which are graphs that contain approximate shortest paths of low step-length (but do not respect capacities). We additionally give efficient algorithms to construct these objects, and generalize them to achieve two additional important properties:
\begin{itemize}
    \item \textbf{Cost-Constrained / Length-Constrained:} Our emulators generalize to various min-cost multi-commodity flow problems, where each edge has an associated length or cost (independent of its capacity) and any flow path must not exceed a given bound on the length or cost.
    \item \textbf{Implicit Mappings:} Our flow emulators also support \emph{implicit} flow mappings from the emulator back to the original graph. In other words, we can even maintain an implicit solution to a constant-approximate multi-commodity flow, and \Cref{thm:informal-constapprox-cnc-flow} can be modified to output such an implicit representation for a flow of the approximated value. (Recall that implicit solutions are required to obtain any $o(mk)$ running time.) With this implicit solution, we can answer the following queries in $O((m+k)^{1+\epsilon})$ time: given any subset of the $k$ commodities, return the union of the flows of each of these commodities. As this is a single-commodity flow, it is representable explicitly within the allotted time.
\end{itemize}
The combination of the two generalizations will allow us to apply flow \emph{boosting} in the spirit of Garg-Konemann~\cite{garg2007faster} to obtain the above algorithms, and further allows us to obtain the above multi-commodity flow result subject to a cost constraint.

\subsection{Our Techniques}
At a high-level, our approach is as follows. First, we compute low-step emulators by building on recent advances in length-constrained expander decompositions. Next, we use our low-step emulators to compute (implicit) flows on the original graph with potentially large congestion. Lastly, we use ``boosting'' to reduce this congestion down to a constant to get our final flow approximation.

\paragraph*{Step 1: Low-Step Emulators via Length-Constrained Expander Decompositions.}
Our techniques build on recent developments in \emph{length-constrained expander decompositions} ~\cite{haeupler2022hop,haeupler2022cut,Haeupler2023lenboundflow, haeupler2023parallel,haeupler2023length}. At a high level, $h$-length expanders are graphs with edge lengths and capacities for which any ``reasonable'' multi-commodity demand can be routed along (about) $h$-length paths with low congestion (by capacity). Informally, a reasonable demand is one where each demand pair $(s_i,t_i)$ is within $h$ by distance (so that sending flow from $s_i$ to $t_i$ along about-$h$-length paths is actually possible), and each vertex does not belong to too many demand pairs (so that the degree of the vertex is not an immediate bottleneck for congestion). Recent work~\cite{haeupler2022hop,haeupler2023length} has studied how to, in almost-linear time, compute \emph{length-constrained expander decompositions} which are length-increases---a.k.a.\ moving cuts---to the graph that make it an $h$-length expander. One caveat, however, is that these algorithms run in time polynomial in the \emph{length} parameter $h$ of the length-constrained expander decomposition. 

In this paper, we remove this polynomial dependency of $h$ by way of the above-mentioned low-step emulators. Specifically, we ``stack'' low-step emulators on top of each other with geometrically increasing lengths, similar to how hopsets are stacked in parallel algorithms~\cite{cohen2000polylog}. Each low-step emulator is responsible for flow paths of its corresponding length, and the union of all low-step emulators obeys all distance scales simultaneously.

Our construction of low-step emulators comes with an \emph{embedding} of the emulator into the base graph: each edge of the emulator maps to a small-length path in the base graph such that the set of all embedded paths has low congestion. When emulators are stacked on top of each other, an edge at the top level expands to a path at the previous level, each of whose edges expands to a path at the previous level, and so on. This hierarchical structure allows us to provide the aforementioned implicit representation of very long paths while keeping the total representation size small.

\paragraph*{Step 2: Flows on Emulators.}
Our next contribution is a fast algorithm that computes a multi-commodity flow on a low-step emulator, where an approximate flow with small representation size is indeed possible (since flow paths now have low step-length). We explicitly compute such a flow, and then (implicitly) map each flow path down the hierarchical structure to form our final implicit flow on the input graph. 

By setting parameters appropriately, we can guarantee a bicriteria approximation with constant-approximate cost and $n^\epsilon$-approximate congestion for any constant $\epsilon>0$. The $n^\epsilon$-approximate congestion appears in both the multi-commodity flow algorithm on a low-step emulator, and the hierarchical embeddings of the emulators into the base graph.

\paragraph*{Step 3: Boosting Away Congestion.}
Finally, through Garg and Konemann's approach~\cite{garg2007faster} based on multiplicative weight updates, we can \emph{boost} the congestion approximation of $n^\epsilon$ down to the cost approximation, which is constant. While Garg and Konemann's algorithm requires a near-linear number of calls to (approximate) shortest path, we only require roughly $n^\epsilon$ calls to $n^\epsilon$-approximate congestion, constant-approximate cost multi-commodity flow. Our final result is a multi-commodity flow with constant-approximate cost and congestion which is implicitly represented by the hierarchical emulator embeddings. Given any subset of commodities, we can then output the union of their flows by collecting the flow down the hierarchy of embeddings.

\section{Overview}
The rest of the paper is organized as follows.
The first part of the paper, Sections \ref{sec:emu} to \ref{sec:low-step emu}, develops the theory of low-step emulators.
In \Cref{sec:emu}, we study the $h$-length-constrained version of low-step emulators. We provide an existential result and an algorithm with a running time that depends on $\poly(h)$. These results are obtained by reducing to $h$-length-constrained expander decomposition.
To remove the dependency on $\poly(h)$ in the running time, in \Cref{sec:bootstrap-emu} we demonstrate how to bootstrap the construction of $h$-length-constrained low-step emulators with small $h$ to the large ones without spending $\poly(h)$ in the running time. Finally, by combining $h$-length-constrained low-step emulators from different $h$ values, we obtain our low-step emulator in \Cref{sec:low-step emu}.

The next part of the paper, \Cref{sec:length-constrained-expander-routing,sec:low-step-flow}, develops a fast $k$-commodity flow algorithm designed for running on top of low-step emulators. The crucial property of this algorithm is that its running time is independent of $k$.
In \Cref{sec:length-constrained-expander-routing}, we first show a prerequisite subroutine called "routing on an expansion witness" for efficiently routing a length-constrained flow. This subroutine assumes that the input graph is a length-constrained expander and also requires the expansion witness of this expander.
We then use this subroutine in \Cref{sec:length-constrained-expander-routing} to develop an algorithm that computes a $t$-step $k$-commodity flow (each flow path contains at most $t$ edges) with a running time of $m^{1+\epsilon} \poly(t)$. The algorithms have an $O(1)$-approximation in step and total length, but an $n^\epsilon$-approximation in congestion. Note that the dependency on $\poly(t)$ in the running time is inconsequential, as the algorithm will be run on a $t$-step emulator for constant $t$. The issue of the large approximation factor in congestion will be addressed in the next part.

The final part of the paper, \Cref{sec:flow boost,sec:general flow}, shows how we can combine the techniques from the previous two parts to achieve an $O(1)$-approximate $k$-commodity flow in $(m+k)^{1+\epsilon}$ time.
To address the problem of the $n^\epsilon$-approximation in congestion, in \Cref{sec:flow boost}, we give a boosting algorithm that improves the congestion approximation to $O(1)$. Finally, in \Cref{sec:general flow}, we combine these tools together to obtain the final result.

The dependencies between sections are summarized in \Cref{fig:depend}.

\begin{figure}
	\centering
 \footnotesize{
	\begin{tikzpicture}[auto]
		\tikzset{
			mynode/.style={rectangle, draw=red, thick, fill=red!10, text width=3.5cm, text centered, rounded corners, minimum height=1cm},
			arrow/.style={->, >=latex', shorten >=2pt, thick},
			line/.style={-, thick}
		}
		
		\node[mynode] (Section 4) at (0,0) {\Cref{sec:emu}\\Existence of LC Low-Step Emulator};
		\node[mynode] (Section 5) at (4,0) {\Cref{sec:bootstrap-emu}\\Bootstrapping LC Low-Step Emulator};
		\node[mynode] (Section 6) at (4,-2) {\Cref{sec:low-step emu}\\Low-step Emulator};
		\node[mynode] (Section 7) at (8,0) {\Cref{sec:length-constrained-expander-routing}\\Routing on an\\ Expansion Witness};
		\node[mynode] (Section 8) at (8,-2) {\Cref{sec:low-step-flow}\\Low-Step Multi-Commodity Flow};
		\node[mynode] (Section 9) at (12,-2) {\Cref{sec:flow boost}\\Flow Boosting};
		\node[mynode] (Section 10) at (8,-4) {\Cref{sec:general flow}\\$O(1)$-Approximate Multi-Commodity Flow};
		
		\draw[arrow] (Section 4) -- (Section 6);
		\draw[arrow] (Section 5) -- (Section 6);
		\draw[arrow] (Section 7) -- (Section 8);
		\draw[arrow] (Section 6) -- (Section 10);
		\draw[arrow] (Section 8) -- (Section 10);
		\draw[arrow] (Section 9) -- (Section 10);
		
	\end{tikzpicture}
	}
	\caption{Dependencies between sections in this paper.\label{fig:depend}}
\end{figure}
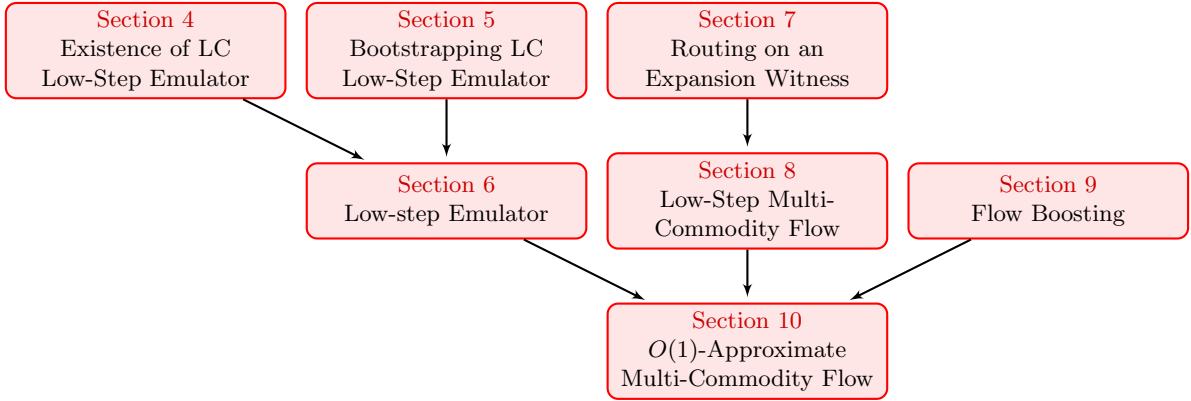


\section{Preliminaries}

Let $G=(V,E)$ be a graph. We denote the number of vertices of $G$ by $n := |V|$ and the number of edges by $m := |E|$, respectively. Let $u_{G}:E\rightarrow\mathbb{R}_{>0}$ denote the edge capacity function of $G$. The degree of a vertex $v$ is $\deg_{G}(v)=\sum_{(v,w)\in E}u_{G}(v,w)$. Note that a self-loop $(v,v)$ contributes $u_{G}(v,v)$ to the degree of $v$. Let $\l_{G}:E\rightarrow\mathbb{\mathbb{R}}_{>0}$ denote the edge length function of $G$. 

In this paper, we will use the word ``hop'' and ``length'' interchangeably\footnote{We sometimes use the term ``hop'' instead of ``length'' in this paper and even use $h$ for parameters related to length so that terminologies are consistent with previous literature \cite{haeupler2022hop,haeupler2023length,haeupler2022cut,haeupler2023parallel} on length-constrained expanders. In the previous papers, they used ``hop'' and $h$ because edges usually have unit length.}. A path from vertex $v$ to vertex $w$ is called a $(v,w)$-path. For any path $P$, the length $\l_{G}(P)=\sum_{e\in E}\l_{G}(e)$ of the path equals the total length of its edges, and the step-length $|P|$ of the path is simply the number of edges in $P$. The \emph{distance} between vertices $v$ and $w$ is $\dist_{G}(v,w)=\min_{P:(v,w)\text{-path}}\l_{G}(P)$. A \emph{ball} of radius $r$ around a vertex $v$ is $\ball_{G}(v,r)=\{w\mid\dist_{G}(v,w)\le r\}$. The diameter $\diam_{G}(S) := \max_{v, w \in S} \dist_{G}(v, w)$ of a vertex set $S$ is the maximum distance between two vertices in the set.

We assume all graphs to have polynomially bounded integral edge capacities and lengths. This assumption allows us to write $\log(\max_{e}\ell(e)),\log(\max_{e}u(e))=O(\log N)$, which simplifies notation. 


\begin{assumption}
The capacity and length of all edges in the input graphs are integral and at most $N=\poly(n)$. 
\end{assumption}
Although we will work with fractional capacities and lengths in the body of the paper, all these values will always be a multiple of some integer reciprocal $1/\poly(n)$ and are upper bounded by $\poly(n)$; by scaling, we can always obtain integral polynomially-bounded capacities and lengths.

\paragraph{Multicommodity Flows.}

A \emph{(multicommodity) flow} $F$ in $G$ is a function that assigns each simple path $P$ in $G$ a flow value $F(P)\ge0$. We say $P$ is a \emph{flow-path} of $F$ if $F(P)>0$, and write $\path(F) = \supp(F) := \{P : F(P) > 0\}$ for the set of flow-paths, equivalently the \textit{support} of $F$, and occasionally abuse notation to write $P \in F$ to mean $P \in \supp(F)$. $|\path(F)|$ is called the \emph{path count} of $F$, and the value of the flow $F$ is $\val(F) := \sum_{P} F(P)$. $P$ is a \emph{$(v,w)$-flow-path} of $F$ if $P$ is both a $(v,w)$-path and a flow-path of $F$. The $(v,w)$-flow $f_{(v,w)}$ of $F$ is the flow for which $f_{(v,w)}(P)=F(P)$ if $P$ is a $(v,w)$-path, otherwise $f_{(v,w)}(P)=0$.

The \emph{congestion of $F$ on an edge $e$} is $\congest_{F}(e)=F(e)/u_{G}(e)$ where $F(e)=\sum_{P:e\in P}F(P)$ denotes the total flow value of all paths going through $e$. The \emph{congestion} of $F$ is $\congest_{F}=\max_{e\in E(G)}\congest_{F}(e)$. The \emph{length} of $F$, denoted by $\leng_{F}=\max_{P\in\supp(F)}\ell(P)$, measures the maximum length of all flow-paths of $F$. The \emph{(maximum) step-length} of $F$ is $\step_{F}=\max_{P\in\supp(F)}|P|$, which measures the maximum number of edges in all flow-paths of $F$. Note that $\step_{F}$ is completely independent from $\leng_{F}$. We sometimes write $\congest_{G,F},\leng_{G,F},\step_{G,F}$ to emphasize that they are with respect to $G$.

\paragraph{Edge Representation of Flows.}

When we want to emphasize an edge $(v,w)$ is undirected, we use the notation $\{v,w\}$. Let $\overleftrightarrow{E}(G)$ denote the set of bidirectional edges of $G$. That is, for each edge $e=\{v,w\}\in E(G)$, we have $(v,w),(w,v)\in\overleftrightarrow{E}(G)$. The \emph{edge representation} of $F$ in $G$ is a function $\flow_{F}:\overleftrightarrow{E}(G)\rightarrow\mathbb{R}_{\ge0}$ where, for each edge $e=\{v,w\}\in E(G)$, let $\flow_{F}(v,w)$ and $\flow_{F}(w,v)$ denote the total flow-value of $F$ routed through $e$ from $v$ to $w$ and from $w$ to $v$, respectively. We sometimes use $F(v,w)=\flow_{F}(v,w)$ and $F(w,v)=\flow_{F}(w,v)$ for convenience. Note that $F(\{v,w\})=F((v,w))+F((w,v))$. 

\paragraph{Demands.}

A \emph{demand} $D:V\times V\rightarrow\mathbb{R}_{\ge0}$ assigns a value $D(v,w)\ge0$ to each ordered pair of vertices. 
Given a flow $F$, the\emph{ demand routed/satisfied by $F$ }is denoted by $D_{F}$ where, for each $u,v\in V$,\emph{ }$D_{F}(u,v)=\sum_{P\text{ is a }(u,v)\text{-path}}F(P)$ is the value of $(u,v)$-flow of $F$. We say that a \emph{demand $D$ is routable in $G$ with congestion $\eta$, length $h$, and $t$} \emph{steps }if there exists a flow $F$ in $G$ where $D_{F}=D$, $\congest_{F}\le\eta$, $\leng_{F}\le h$, and $\step_{F}\le t$. The total demand size is $|D|=\sum_{u,v}D(u,v)$. The support of $D$ is $\supp(D)=\{(v,w)\mid D(v,w)>0\}$.

A \emph{node-weighting} $A:V\rightarrow\mathbb{R}_{\ge0}$ of $G$ assigns a value $A(v)$ to each vertex $v$. The \emph{size} of $A$ is denoted by $|A|=\sum_{v}A(v)$. For any two node-weightings $A$ and $B$, we write $A\le B$ if $A(v)\le B(v)$ for all $v$. We say $D$ is \emph{$A$-respecting} if both $\sum_{w}D(v,w)\le A(v)$ and $\sum_{w}D(w,v)\le A(v)$. We say that $D$ is \emph{degree-respecting} if $D$ is $\deg_{G}$-respecting. Next, $D$ is \emph{$h$-length-bounded }(or simply \emph{$h$-length }for short) if it assigns positive values only to pairs that are within distance at most $h$, i.e., $D(v,w)>0$ implies that $\dist_{G}(v,w)\le h$. The \emph{load} of a demand $D$ is defined as the node weighting $\mathrm{load}(D)(v)=\sum_{w\in V}D(v,w)+D(w,v)$.

\subsection{Length-Constrained Expanders}

\begin{comment}
Recall the following almost equivalent characterizations of normal $\phi$-expanders: 
\begin{enumerate}
\item \textbf{(Cut characterization):} a graph with no \emph{$\phi$-sparse cut}\footnote{A cut $C\subseteq E$ is $\phi$-sparse if $S\subseteq V$ is a connected component in $G\setminus C$ and $\frac{|C|}{\min\{\deg_{G}(S),\deg_{G}(V\setminus S)\}}<\phi$. }, and 
\item \textbf{(Flow characterization):} a graph where any degree-respecting demand can be routed with congestion and length $O(\log(n)/\phi)$. 
\end{enumerate}
It is well-known that the two statements imply each other (within a logarithmic factor) CITE. The motivation behind the definition of \emph{hop-constrained expanders} developed in \cite{haeupler2022hop} is to extend the above powerful equivalence to the setting where we can give a much stronger bound on length $h\ll1/\phi$. To do this, it turns out that the ``right'' extension of cuts in the notion of moving cuts.

\textbf{NOTE: this discussion is not really good because we will work with $h\gg1/\phi$ too.} 
\end{comment}

\paragraph{Moving Cuts.}

An \emph{$h$-length moving cut }$C:E\rightarrow\{0,\frac{1}{h},\frac{2}{h},\dots,1\}$ assigns a fractional which is a multiple of $\frac{1}{h}$ between zero and one to each edge $e$. 
The \emph{size} of $C$ is defined as $|C|=\sum_{e}u(e)\cdot C(e)$. We denote with $G-C$ the graph $G$ with a new length function $\l_{G-C}(e)=\l_{G}(e)+h\cdot C(e)$ for all $e$. We refer to $G-C$ as the \emph{graph $G$ after cutting} $C$ or \emph{after applying the moving cut} $C$.

Given a $h$-length demand $D$, we usually work with a $(hs)$-length moving cut $C$ where $s>1$. The \emph{total demand of $D$ separated by $C$} is then denoted by 
\[
\sep_{hs}(C,D)=\sum_{(v,w):\dist_{G-C}(v,w)>hs}D(v,w),
\]
which measures the total amount of demand between vertex pairs whose distance is increased (from at most $h$, as $D$ is $h$-length) to strictly greater than $hs$ by applying the cut $C$. Note that an integral moving cut (i.e. one for which $C(e) \in \{0, 1\}$) functions somewhat like a classic cut: the size $|C|$ is the total capacity as a classic cut, and $\sep_{hs}(C,D) = \sum_{(v,w):\text{all $hs$-length $(v, w)$-paths in $G$ are cut}}D(v,w)$.

The \emph{sparsity of moving cut $C$ with respect to a demand $D$}, for a $(hs)$-length moving cut $C$ and a $h$-length demand $D$, denoted by 
\[
\spars_{(h,s)}(C,D)=\frac{|C|}{\sep_{hs}(C,D)}
\]
is the ratio between the size of the cut $C$ and the total demand of $D$ separated by $C$. We say that $C$ is \emph{$\phi$-sparse with respect to} $D$ is $\spars(C,D)<\phi$.

We are ready to define the key notion of length-constrained expanders. 
\begin{defn}
[Cut Characterization of Length-Constrainted Expanders]The \emph{$(h,s)$-length conductance }of a graph $G$ is 
\[
\cond_{(h,s)}(G)=\min_{D:h\text{\text{-length, }\text{degree-respecting}}}\min_{C:(hs)\text{-moving cut}}\spars_{(h,s)}(C,D).
\]
If $\cond_{(h,s)}(G)\ge\phi$, we say that $G$ is a \emph{$(h,s)$-length $\phi$-expander}. More generally, for any node-weighting $A$, the \emph{$(h,s)$-length conductance }of $A$ is 
\[
\cond_{(h,s)}(A)=\min_{D:h\text{-length }A\text{-respecting}}\min_{C:(hs)\text{-moving cut}}\spars_{(h,s)}(C,D).
\]
If $\cond_{(h,s)}(A)\ge\phi$, then we say that $A$ is \emph{$(h,s)$-length $\phi$-expanding in $G$ and that $G$ is a $(h,s)$-length $\phi$-expander for $A$.} 
\end{defn}

Note that $\cond_{(h,s)}(G)=\cond_{(h,s)}(\deg_{G})$. In words, if $G$ is an $(h,s)$-length $\phi$-expander, then there is no $\phi$-sparse $(hs)$-moving cut with respect to any $h$-length degree-respecting demand. The following observation draws a connection between length-constrained expanders and normal expanders. 

\begin{prop}
When $h\rightarrow\infty$, $G$ is an $(h,s)$-length $\phi$-expander if and only every connected component of $G$ is a $\phi$-expander.
\end{prop}

The theorem below shows that, similar to normal expanders, there is a flow characterization of length-constrained expanders that is almost equivalent to its cut characterization within a logarithmic factor. 
\begin{thm}
[Flow Characterization of Length-Constrainted Expanders (Lemma 3.16 of \cite{haeupler2022hop})]\label{thm:flow character}We have the following: 
\begin{enumerate}
\item If $A$ is $(h,s)$-length $\phi$-expanding in $G$, then every $h$-length $A$-respecting demand can be routed in $G$ with congestion at most $O(\log(N)/\phi)$ and length at most $s\cdot h$. 
\item If $A$ is not $(h,s)$-length $\phi$-expanding in $G$, then some $h$-length $A$-respecting demand cannot be routed in $G$ with congestion at most $1/2\phi$ and length at most $\frac{s}{2}\cdot h$. 
\end{enumerate}
\end{thm}

\subsection{Length-Constrained Expander Decomposition}

Expander decomposition are a powerful tool that allows algorithm designers to exploit the power of expanders in an arbitrary graph. The hop-constrained version of them is stated below. 
\begin{defn}
[Expander Decomposition] A \emph{$(h,s)$-length $(\phi,\kappa)$-expander decomposition} for a node-weighting $A$ in a graph $G$ is a $(hs)$-moving cut $C$ of size at most $\kappa\phi|A|$ such that $A$ is $(h,s)$-length $\phi$-expanding in $G-C$. $C$ is also called a $(h,s,\phi,\kappa)$-decomposition for $A$, for short. 
\end{defn}

We refer to the parameters $s$ and $\kappa$ as the \emph{length slack} and \emph{congestion slack}, respectively.

For any moving cut $C$, the \emph{degree with respect to $C$} of a vertex $v$ is defined as 
\[
\deg_{C}(v)=\sum_{e\text{ incident to }v}u(e)\cdot C(e).
\]
By the definition, observe that $\deg_{C}(v)\le\deg_{G}(v)$ for any moving cut $C$. As discovered in \cite{goranci2021expander}, expander decompositions become even more versatile when the vertices incident to the cut-edges of the decomposition are more ``well-linked''. The definition of a \textit{linked expander decomposition} is given below.
\begin{defn}
[Linked Expander Decomposition]\label{def:linked decomp} A \emph{$\beta$-linked $(h,s)$-length $(\phi,\kappa)$-expander decomposition} for a node-weighting $A$ in a graph $G$ is a $(hs)$-moving cut $C$ of size at most $\kappa\phi|A|$ such that $A+\beta\cdot\deg_{C}$ is $(h,s)$-length $\phi$-expanding in $G-C$. 
\end{defn}


Notice that an expander decomposition is simply a $\beta$-linked expander decomposition for $\beta = 0$. Requiring $A + \beta\cdot\deg_{C}$ instead of just $A$ to be expanding captures the intuition of requiring vertices incident to $C$ to be more ``well-linked''.
\begin{thm}
[implicit in \cite{haeupler2022hop}, explicit in Theorem 3 of \cite{haeupler2023parallel}]\label{thm:expdecomp exist}For any node-weighting $A$, length bound $h$, length slack $s\ge100$, conductance bound $\phi>0$, congestion slack $\kappa\ge N^{O(1/s)}\log N$, and linkedness $\beta=O(1/(\phi\kappa))$, there exists a $\beta$-linked $(h,s)$-length $(\phi,\kappa)$-expander decomposition for $A$. 
\end{thm}



\subsection{Routers}

Next, we define the notion of \emph{routers}. 
\begin{defn}
[Routers]\label{def:router}For any node-weighting $A$, we say that a \emph{unit-length} capacitated graph $G$ is a \emph{$t$-step $\kappa$-router for $A$} (or simply a \emph{router}) if every $A$-respecting demand can be routed in $G$ with congestion $\kappa$ and $t$ steps. If $G$ is a $t$-step $\kappa$-router for $\deg_{G}$, then we simply say that $G$ is $t$-step $\kappa$-router. 
\end{defn}

By the flow characterization of length-constrained expanders (\Cref{thm:flow character}), observe that routers are essentially the same object as length-constrained expanders when the graph has unit length and bounded diameter. 
\begin{prop}
\label{prop:expander-router} Let $G$ be a graph with unit edge-length. Let $A$ be a node-weighting where $\diam_{G}(\supp(A))\le h$. 
\begin{enumerate}
\item If $G$ is a $(h,s)$-length $\phi$-expander for $A$, then $G$ is a $(hs)$-step $O(\frac{\log N}{\phi})$-router for $A$. 
\item If $G$ is not a $(h,s)$-length $\phi$-expander for $A$, then $G$ is not a $(\frac{hs}{2})$-step $\frac{1}{2\phi}$-router for $A$. 
\end{enumerate}
\end{prop}

\begin{proof}
Since $\diam_{G}(\supp(A))\le h$, the set of $A$-respecting demands and the set of $h$-length $A$-respecting demands are identical. Also, as $G$ is unit-length, every path has length equal to the path's step-length. Therefore, by \Cref{thm:flow character}, if $G$ is a $(h,s)$-length $\phi$-expander for $A$, every $A$-respecting demand can be routed in $G$ with congestion $O(\log(N)/\phi)$ and $sh$ steps. Otherwise, some $A$-respecting demand cannot be routed in $G$ with congestion $1/2\phi$ and $(sh)/2$ steps. 
\end{proof}
Although routers and length-constrained expanders are very similar, they focus on different things. We bound the maximum step-length of flow on routers, and bound the length of flows on length-constrained expanders. The clear distinction between the two is made to avoid confusion.

A simple example of a router is a star. 
\begin{prop}
\label{lem:router star} For any node-weighting $A$, let $H$ be a star rooted at $r\notin\supp(A)$ with leaf set $\supp(A)$, each star-edge $(r,v)\in E(H)$ having capacity $A(v)$. Then, $H$ is a $2$-step $1$-router for $A$. 
\end{prop}

While the above router has great quality, it requires an additional vertex $r\notin\supp(A)$.

A strong router without steiner vertices (one for which $V(H)\subseteq\supp(A)$) can be constructed using constant-degree expanders. The parameter in the construction below can likely be improved, but we choose to present a simple construction. 
\begin{lem}
\label{lem:router expander} For any node-weighting $A$ and positive integer parameter $t$, there exists a $t$-step $1$-router $H=\router(A,t)$ for $A$ such that 
\begin{itemize}
\item $\deg_{H}\le\Delta\cdot A$ 
\item $|E(H)|\le\Delta\cdot|\supp(A)|$ 
\end{itemize}
where $\Delta=tn^{O(1/t)}\log^{2}N$. 

\end{lem}

\begin{proof}
First, suppose that the node-weighting $A$ is uniform, i.e., $A(v)=1$ for all $v$. Let $H_{0}$ be a $O(1)$-degree expander of constant conductance on the vertex set $|\supp(A)|$. It is well-known that such a $H_{0}$ is a $t' := O(\log n)$-step $O(\log n)$-router. Moreover, for $k=\ceil{\frac{t'}{t}}$, the power graph $H_{0}^{k}$ is a $t$-step $O(t)$-router for $A$. Each vertex in $H_{0}^{k}$ has degree $\Delta_{0}=O(1)^{k}=n^{O(1/t)}$. Define $H=H_{0}^{k}$. We have that $\deg_{H}\le\Delta_{0}\cdot A$ and $|E(H)|\le\Delta_{0}\cdot|\supp(A)|$ as desired. If $A(v)=c_{0}$ for all $v$ for some $c_{0}$, then $H$ can be constructed in the same way but we scale up the capacity by $c_{0}$.

Now, we handle the general node-weighting $A$. By paying at most a factor of $2$ in the congestion, we assume that $A(v)=2^{i}$ for some $i\in[\log N]$ for every $v$. Let $V_{i}=\{v\in\supp(A)\mid A(v)=2^{i}\}$ and $A_{i}=A\cap V_{i}$ be the restriction of $A$ to $V_{i}$. Let $H_{i}$ be the $t$-step $O(t)$-router for $A_{i}$. The final router $H$ is contained by connecting these $H_{i}$ together. For each $i,j$ where $i>j$, we construct a bipartite graph $H_{i,j}$ where $V(H_{i,j})=V_{i}\cup V_{j}$ such that every edge of $H_{i,j}$ has capacity $2^{j}$, $\deg_{H_{i,j}}(v)\le2^{i}$ for all $v\in V_{i}$, and $\deg_{H_{i,j}}(v)\le2^{j}$ for all $v\in V_{j}$. If $2^{i}|V_{i}|\le2^{j}|V_{j}|$, then $\deg_{H_{i,j}}(v)=2^{i}$ for all $v\in V_{i}$, otherwise $\deg_{H_{i,j}}(v)=2^{j}$ for all $v\in V_{j}$. This can be done be greedily adding edges of capacity $2^{j}$ between $V_{i}$ and $V_{j}$ in a natural way. We have $|E(H_{i,j})|\le\max\{|V_{i}|,|V_{j}|\}$. We define the final router as $H=(\bigcup_{i}H_{i})\cup(\bigcup_{i,j}H_{i,j}).$ Observe that 
\[
|E(H)|\le\sum_{i}|E(H_{i})|+\sum_{i,j}|E(H_{i,j})|\le\Delta_{0}|\supp(A)|+\log N|\supp(A)|.
\]
Similarly, $\deg_{H}\le(\Delta_{0}+\log N)\cdot A$. We claim that $H$ is a $2t$-step $O(t\log N)$-router. To see the claim, given any demand $D$ respecting $A$, we first route the demand from $V_{i}$ to $V_{j}$ through $H_{i,j}$ with congestion $1$ for all $i,j$. The residual demand will respects $\log N\cdot A$ and only need to route inside each $H_{i}$, which can be routed using $O(t\log N)$ congestion.

Finally, to reduce the congestion to $1$, we simply make $O(t \log N)$ parallel copies of each edge. Scaling the parameters by an appropriate constant, this implies the lemma.
\begin{comment}
By construction, if $A(v)\in\frac{1}{z}\mathbb{Z}_{\ge0}$, then $u_{H}(e)\in\frac{1}{z}\mathbb{Z}_{>0}$ for all edges $e\in E(H)$. 
\end{comment}
\end{proof}

\subsection{Length-Constrained Expansion Witness}

We first recall two more standard notions.

\paragraph{Neighborhood Covers.}

Given a graph $G$ with lengths, a \emph{clustering} $\cS$ in $G$ is a collection of mutually disjoint vertex sets $S_{1},\ldots,S_{|\cS|}$, called \emph{clusters}. A \emph{neighborhood cover} $\cN$ with width $\omega$ and covering radius $h$ is a collection of $\omega$ many clusterings $\cS_{1},\ldots,\cS_{\omega}$ such that for every node $v$ there exists a cluster $S\in\S_{i}$ where $\ball(v,h)\subseteq S$. We use $S\in\N$ to say that $S$ is a cluster in some clustering of $\N$. We say that $\cN$ has diameter $h_{\diam}$ if every cluster $S\in\N$ has (weak) diameter at most $h_{\diam}$, i.e., $\max_{u,v\in S}\dist_{G}(u,v)\le h_{\diam}$. The following is a classic result.
\begin{thm}[\cite{peleg2000distributed}]
\label{thm:cover basic}\label{thm:cover-separation-factor-existential}For any $h$, integer $k\ge1$, and graph $G$, there a deterministic parallel algorithm that computes a neighborhood cover $\N$ with covering radius $h$, diameter $h_{\diam}\le(2k-1)\cdot h$ and width $\omega=N^{O(1/k)}\log N$. The algorithm has $O(|E(G)|hk\omega)$ work and $O(hk\omega)$ depth.
\end{thm}

\paragraph{Embedding.}

Next, we recall the notion of \emph{graph embedding. }We  view it as a flow as follows. %
\begin{comment}
\begin{defn} {[}Embedding{]}Given a graph $G$ and a graph $H$ (possibly with parallel edges) where $V(H)\subseteq V(G)$, an embedding $\Pi_{H\rightarrow G}$ of $H$ into $G$ is a collection of paths $\{p_{(v,w)}\}_{(v,w)\in E(H)}$ where $p_{(v,w)}$ is a $(v,w)$-flow in $G$, which can be viewed as a $(v,w)$-flow with value $\val(f_{(v,w)})=u_{H}(v,w)$.

We say that $\Pi_{H\rightarrow G}$ has \emph{length $h$ }if $\l_{G}(p_{e})\le h$ for all $e\in E(H)$.

We say that $\Pi_{H\rightarrow G}$ has \emph{length slack} $s$ if $\l_{G}(p_{e})\le s\cdot\l_{H}(e)$ for all $e\in E(H)$.

We say that $\Pi_{H\rightarrow G}$ has \emph{congestion }$\kappa$ if\emph{ $\congest_{F}\le\kappa$ where $F=\sum_{e\in E(H)}p_{e}$.}

The representation size of $\Pi_{H\rightarrow G}$ is $|\Pi_{H\rightarrow G}|=\sum_{e\in E(H)}|E_{G}(p_{e})|$ where $E_{G}(p_{e})$ denote the set of edges of $p_{e}$ in $G$. \end{defn}
\end{comment}
\begin{defn}
[Edge Demand and Embedding]Given graphs $G$ and $H$ where $V(H)\subseteq V(G)$, the \emph{edge-demand} $D_{E(H)}$ of $H$ on $G$ is the demand where for all $(v, w) \in E(H)$, $D_{E(H)}(v,w) = u_{H}(v, w)$. The \emph{embedding $\Pi_{H \rightarrow G}$ of $H$ into $G$} is a multicommodity flow that routes $D_{E(H)}$ in $G$. We can write $\Pi_{H\rightarrow G}=\sum_{(v,w)\in E(H)}f_{(v,w)}$ where $f_{(v,w)}$ is a $(v,w)$-flow in $G$ of value $\val(f_{(v,w)})=u_{H}(v,w)$.

The embedding $\Pi_{H\rightarrow G}$ is said to have \emph{length slack} $s$ if $\leng_{f_{e}}\le s\cdot\l_{H}(e)$ for all $e\in E(H)$. 
\end{defn}

We use the same terminology as for flows for the embedding $\Pi_{H\rightarrow G}$. For example, $\Pi_{H\rightarrow G}$ is said to have length \emph{$h$ }and congestion $\kappa$ if $\leng_{\Pi_{H\rightarrow G}}\le h$ and \emph{$\congest_{\Pi_{H\rightarrow G}}\le\kappa$}.

\paragraph{Expansion Witness.}

Given the above definitions of neighborhood covers, routers, and embeddings, we can define a \textit{witness} of length-constrained expansion.

\begin{restatable}[Expansion Witness]{defn}{expansionwitness}\label{def:expansion witness}
Let $G$ be a graph and $A$ a node weighting. A \emph{$(h,t_{{\cal R}},h_{\Pi},\kappa_{\Pi})$-witness} of $A$ in $G$ is a tuple $({\cal N},\mathcal{R},\Pi_{{\cal R}\to G})$. 
\begin{itemize}
\item ${\cal N}$ is a neighborhood cover with covering radius $h$. 
\item $\mathcal{R}$ is a collection of routers. For each cluster $S\in{\cal N}$, there exists a $(t_{{\cal R}},1)$-router $R^{S}\in{\cal R}$ on vertex set $S$ for node-weighting $S\cap A$. 
\item $\Pi_{{\cal R}\to G}$ is an embedding of all routers in ${\cal R}$ to $G$. $\Pi_{{\cal R}\to G}$ has length $h_{\Pi}$ and congestion $\kappa_{\Pi}$. 
\end{itemize}
\end{restatable}


Below, we show that (1) an expansion witness indeed certifies the expansion, and moreover, gives an explicit routing structure, and (2) an expansion witness exists for every expander.
\begin{cor}
[Expansion Witness Certifies Expansion]Suppose that there exists a \emph{$(h,t_{{\cal R}},t_{\Pi},\kappa_{\Pi})$-witness} $({\cal N},\mathcal{R},\Pi_{{\cal R}\to G})$ of $A$ in $G$. Then, $G$ is a $(h,2t_{{\cal R}}t_{\Pi})$-length $(1/2\kappa_{\Pi})$-expander for $A$. Moreover, given any $A$-respecting $h$-length-constrained demand $D$ in $G$, there is a flow $F$ routing $D$ in $G$ where $\leng_{F}\le h_{\Pi}t_{{\cal R}}$ and $\congest_{F}\le\kappa_{\Pi}$. Moreover, $F$ is ``routed through'' the embedding $\Pi_{{\cal R}\to G}$, i.e., $F=\sum_{f\in\Pi_{{\cal R}\to G}}\val_{f}\cdot f$ where $\val_{f}\ge0$. 
\end{cor}

\begin{proof}
It suffices to prove the ``moreover'' part by \Cref{thm:flow character}. For each $(v,w)$ where $D(v,w)>0$, we have $\dist_{G-C}(v,w)\le h$ and so there exists $S\in\N$ where $v,w\in S$. Choose such cluster $S$ arbitrarily and assign the demand $D(v,w)$ to $S$. For each cluster, let $D_{S}$ denote the demand induced by this process. Note that $D_{S}\le D$ entry-wise and so $D_{S}$ respects $S\cap A$. So $D_{S}$ can be routed via a flow $F_{S}$ in $R^{S}$ with $t_{{\cal R}}$-step and $1$-congestion. Let $F_{{\cal R}}=\sum_{S\in\N}F_{S}$ be a flow on ${\cal R}=\cup_{S\in\N}R^{S}$.

We define $F$ from $F_{{\cal R}}$ as follows. From the embedding $\Pi_{{\cal R}\rightarrow G}=\sum_{e\in E({\cal R})}u_{{\cal R}}(e)\cdot f_{e}$ that embeds ${\cal R}$ into $G$, define the flow $F=\sum_{e\in E({\cal R})}v_{f_{e}}\cdot f_{e}$ where $f_{e}\in\Pi_{{\cal R}\rightarrow G}$ and $v_{f_{e}}=F_{{\cal R}}(e)$ denotes the total flow of $F_{{\cal R}}$ through $e$ in ${\cal R}$. Now, we bound the length and congestion of $F$. Since $F_{{\cal R}}$ has at most $t_{{\cal R}}$-step on ${\cal R}$, we have $\leng_{F}\le t_{{\cal R}}\cdot\leng_{\Pi_{{\cal R}\rightarrow G}}\le h_{\Pi}t_{{\cal R}}$. Since $F_{{\cal R}}$ has congestion $1$ on ${\cal R}$, we have $v_{f_{e}}\le u_{{\cal R}}(e)$ for all $e\in E({\cal R})$ and so $\congest_{F}\le\congest_{\Pi_{{\cal R}\rightarrow G}}=\kappa_{\Pi}$.
\end{proof}
\begin{cor}
[Expansion Witness Exists for Expanders]\label{cor:witness exists}Let $G$ be a $(h,s)$-length $\phi$-expander for a node-weighting $A$. There is an \emph{$(h',t_{{\cal R}},h_{\Pi},\kappa_{\Pi})$-}witness $({\cal N},{\cal R},\Pi_{{\cal R}\to G})$ for $A$ in $G$ where $h'=h/s$, $t_{{\cal R}}=s$, $h_{\Pi}=hs$, $\kappa_{\Pi}=sN^{O(1/s)}\poly(\log N)/\phi$.
\end{cor}

\begin{proof}
From \Cref{thm:cover basic}, let $\N$ be a neighborhood cover on $G$ with covering radius $h'=h/s$, diameter $h$, and width $\omega\le N^{O(1/s)}\log N$. For each cluster $S\in{\cal N}$, let $R^{S}\in{\cal R}$ be a $s$-step $1$-router for $B_{S}=S\cap A$ where $\deg_{R^{S}}\le\Delta\cdot B_{S}$ and $|E(R^{S})|\le\Delta\cdot|\supp(B_{S})|$ and $\Delta\le sN^{O(1/s)}\log^{2}N$. This follows from \Cref{lem:router expander}.

Let $D_{{\cal R}}$ be the edge demand of $\cup_{S\in\N}R^{S}$. First, observe that $D_{{\cal R}}$ is $h$-length-constrained in $G$. Indeed, for each edge $(v,w)\in R^{S}$ for any $S\in\N$, we have $\dist_{G}(v,w)\le h$ as $\N$ has diameter $h$. Second, observe that $D_{{\cal R}}$ respects $\Delta\omega A$. Indeed, for each vertex $v$, we have 
\[
D_{{\cal R}}(v)\le\sum_{S\in\N}\deg_{R^{S}}(v)\le\Delta\cdot\sum_{S\in\N}B_{S}(v)\le\Delta\omega\cdot A(v)
\]
 where the last inequality is because there are most $\omega$ many clusters $S\in\N$ containing $v$.

Since $A$ is $(h,s)$-length $\phi$-expanding in $G$, by \Cref{thm:flow character}, we have that $D_{{\cal R}}$ can be routed in $G$ with length $hs=h_{\Pi}$ and congestion $\kappa_{\Pi}=O(\Delta\omega\log N/\phi)$. Let $\Pi_{{\cal R}\rightarrow G}$ be such flow that routes $D_{{\cal R}}$ in $G$. 
\end{proof}

\subsection{Length-Constrained Witnessed Expander Decomposition}

By combining the existence of the expander decomposition and the expansion witness, we obtain the following. 
\begin{cor}
[Existential Witnessed Expander Decomposition]\label{cor:witness ED exist}Let $G$ be a graph with edge lengths with node-weighting $A$. Given parameters $(h,\phi,\beta,s)$ where $\beta\le1/(\phi\log N)$ and $s\ge100$, there exists a $h_{C}$-moving cut $C$ and a \emph{$(h,t_{{\cal R}},h_{\Pi},\kappa_{\Pi})$-witness} $({\cal N},\mathcal{R},\Pi_{{\cal R}\to G})$ of $A+\beta\deg_{C}$ in $G-C$ with the following guarantees: 
\begin{itemize}
\item $|C|\le\phi|A|$. 
\item The total number of edges in all routers is $|E({\cal R})|\le nN^{O(1/s)}\poly(\log N)$. 
\item $t_{{\cal R}}=s$, $h_{C},h_{\Pi}=hs^{2}$, $\kappa_{\Pi}=N^{O(1/s)}\poly(\log N)/\phi$.
\end{itemize}
\end{cor}

\begin{proof}
From \Cref{thm:expdecomp exist}, there exists a $hs^{2}$-moving cut $C$ of size at most $\phi|A|$ such that $A+\beta\deg_{C}$ is $(hs,s)$-length $(\phi/\kappa)$-expanding in $G-C$ where $\kappa=N^{O(1/s)}\log N$. By \Cref{cor:witness exists}, there is a \emph{$(h',t_{{\cal R}},h_{\Pi},\kappa_{\Pi})$-}witness $({\cal N},{\cal R},\Pi_{{\cal R}\to G})$ for $A+\beta\deg_{C}$ in $G-C$ where $h'=hs/s$, $t_{{\cal R}}=s$, $s_{\Pi}=hs^{2}$, $\kappa_{\Pi}=\frac{\kappa}{\phi}sN^{O(1/s)}\poly\log N$.
\end{proof}
The key subroutine that this paper relies on as a blackbox is an efficient parallel algorithm for computing a length-constrained expander decomposition $C$ for $A$, together with the expansion witness for $A$ in $G-C$.
\begin{thm}
[Algorithmic Witnessed Expander Decomposition from Theorem 1.1 of \cite{haeupler2023length}]
\label{thm:witness ED alg}Let $G$ be a graph with edge lengths with node-weighting $A$. Given parameters $(h,\phi,\beta,s)$ where $\beta\le1/(\phi\log N)$ and $s \leq \log^{c} N$ for some sufficiently small constant $c$, let $\eps = 1/s$. There exists an algorithm that computes an $h_{C}$-moving cut $C$ and a \emph{$(h,t_{{\cal R}},h_{\Pi},\kappa_{\Pi})$-witness} $({\cal N},\mathcal{R},\Pi_{{\cal R}\to G})$ of $A+\beta\deg_{C}$ in $G-C$ with the following guarantees: 
\begin{itemize}
\item $|C|\le\phi|A|$. 
\item The total number of edges in all routers is $|E({\cal R})|\le|E(G)|N^{\poly\eps}$. Moreover, $\Pi_{{\cal R}\to G}$ is an integral embedding with path count at most $|E(G)|N^{\poly\eps}$. 
\item $t_{{\cal R}}=s,$ $h_{C},h_{\Pi}\le h\cdot\exp(\poly(1/\eps))$, and $\kappa_{\Pi}=N^{\poly\eps}/\phi$. 
\end{itemize}
The algorithm takes $|E(G)|\poly(h)N^{\poly\eps}$ work and has depth $\poly(h)N^{\poly\eps}$.
\end{thm}



\section{Length-Constrained Low-Step Emulators: Existence}

\label{sec:emu}

The goal of this section is to show that, given any graph $G$, there is another graph $G'$ such that any length-constrained multi-commodity flow in $G$ can be routed in $G'$ with approximately the same congestion \emph{and} length, and moreover such flow in $G'$ can be routed via paths containing few edges. The definition below formalize this idea. 
\begin{defn}
[Length-Constrained Low-Step Emulators]\label{def:h-len emu}Given a graph $G$ and a node-weighting $A$ in $G$, we say that $G'$ is an \emph{$h$-length-constrained $t$-step emulator} of $A$ with length slack $s$ and congestion slack $\kappa$ if 
\begin{enumerate}
\item \label{enu:uniform}Edges in $G'$ have uniform length of $h'=sh$. 
\item Given a flow $F$ in $G$ routing an $A$-respecting demand and $\leng_{F}\le h$, there is a flow $F'$ in $G'$ routing the same demand where $\congest_{F'}\le\kappa\cdot\congest_{F}$ and $\step_{F'}\le t$ (equivalently, $\leng_{F'}\le t\cdot sh$ by \Cref{enu:uniform}). 
\item $G'$ can be embedded into $G$ with congestion $1$ and length slack $1$. 
\end{enumerate}
\end{defn}

\begin{comment}
\thatchaphol{move congestion to the embedding from $G'$ to $G$ instead? In this way, it is intuitive clear that $G'$ is \textquotedbl{}bigger and more well-connected\textquotedbl{} than $G$.}
\end{comment}

The condition requiring that edges in $G'$ have uniform length is crucial. This will allow us to bootstrap the construction for\emph{ }$h$-length-constrained $t$-step emulator for large $h$ using based on the ones for small $h$. Our main technical contribution of this section is showing the existence of length-constrained low-step emulators. 

\begin{thm}
[Existential]\label{thm:h leng emu}Given any graph $G$ with $n$ vertices, a node-weighting $A$ of $G$, and parameters $h$ and $t$, there exists an $h$-length-constrained $O(t^{2})$-step emulator $G'$ for $A$ in $G$ with length slack $O(t^{2})$ and congestion slack $\poly(t\log N)N^{O(1/t)}$. The emulator contains $nN^{O(1/t)}\poly(t\log N)$ edges.
\end{thm}

We also give a parallel algorithm for constructing an emulator with a worse trade-off. 
\begin{thm}
[Algorithmic]\label{thm:h leng emu alg}There exists a parallel algorithm that, given any graph $G$ with $m$ edges, a node-weighting $A$ of $G$, and parameters $h$ and $\eps \in (\log^{-c} N, 1)$ for some sufficiently small constant $c$, computes an $h$-length-constrained $O(t^{2})$-step emulator $G'$ for $A$ in $G$ where $t=\exp(\poly1/\eps)$ with length slack $O(t^{2})$ and congestion slack $N^{\poly\eps}$. The emulator contains $mN^{\poly\eps}$ edges and the embedding $\Pi_{G'\rightarrow G}$ has path count $mN^{\poly\eps}$.

The algorithm has $m \cdot \poly(h)N^{\poly\eps}$ work and $\poly(h)N^{\poly\eps}$ depth. 
\end{thm}

\subsection{Construction and Analysis}

In this section, we show a construction of $h$-length-constrained $t$-step emulator.%
\begin{comment}
For any node weighting $A$, let $\router(A,t_{{\cal R}},\kappa_{\rt})$ denote a $t_{{\cal R}}$-step $\kappa_{\rt}$-congestion router for $A$ from \Cref{lem:router expander}. We assign in this section. 
\end{comment}

\begin{algorithm}
\begin{comment}
, $h_{\cov}\gets4h$, $\phi\gets1/2N^{1/t}$, $\beta\gets h_{C}/h_{\cov}$, $s_{\decomp}\gets t$.

$h_{\diam}\gets8th$, $t_{{\cal R}}\gets t$, $\beta\gets h_{\diam}s_{\decomp}/h=8t^{2}$
\end{comment}
Initialize $A_{0}\gets A$, $h_{\cov}\gets4h,$ $s_{\decomp}\gets t$, $\beta\gets t^{2},$ and $\phi\gets1/2\beta N^{1/t}.$ (The choice of $\beta$ is so that $\beta=h_{C}/h$).
\begin{enumerate}
\item For $1\le i\le2t$: 
\begin{enumerate}
\item Given node-weighing $A_{i-1}$, compute a witnessed expander decomposition $C_{i}$ and $({\cal N}_{i},\mathcal{R}_{i},\Pi_{{\cal R}_{i}\to G})$ with parameters $(h_{\cov},\phi,\beta,s_{\decomp})$ using \Cref{cor:witness ED exist}.
\begin{enumerate}
\item $C_{i}$ is a $h_{C}$-moving cut, and 
\item $({\cal N}_{i},\mathcal{R}_{i},\Pi_{{\cal R}_{i}\to G})$ is a \emph{$(h_{\cov},t_{{\cal R}},h_{\Pi},\kappa_{\Pi})$}-witness of $A_{i-1}+\beta\deg_{C_{i}}$ in $G-C_{i}$.%
\begin{comment}
\begin{itemize}
\item $|C_{i}|\le|A_{i}|/2N^{1/t}$. 
\item The total number of edges in all routers is $|E({\cal R})|\le n\cdot N^{O(1/t)}\poly(\log N)$. 
\item $t_{{\cal R}}=s$, $h_{\Pi}=hs^{2}$ and $\kappa_{\Pi}=N^{O(1/t)}\poly(\log N)/\phi$. 
\end{itemize}
\begin{enumerate}
\item Let $C_{i}$ be a $\beta$-linked $(h_{\diam},s_{\decomp})$-length $(\phi,\kappa_{\decomp})$-expander decomposition for $A_{i-1}$ from \Cref{thm:expdecomp exist} where $\kappa_{\decomp}=N^{O(1/t)}\log N$ and $\phi=\frac{1}{2\beta\kappa_{\decomp}N^{1/t}}$. 
\item Let $B_{i}=A_{i-1}+\beta\deg_{C_{i}}$ and $A_{i}=\beta\deg_{C_{i}}$. 
\item Let $(\N_{i},{\cal R}_{i},\Pi_{{\cal R}_{i}\rightarrow G})$ be the ED-witness of $C_{i}$ from \Cref{prop:witness exists} where 
\begin{itemize}
\item $\N_{i}$ is a neighborhood cover for $B_{i}$ in $G-C_{i}$ with covering radius $h_{\cov}$ and diameter $h_{\diam}$, and width $\omega=N^{O(1/t)}\log N$. 
\item For each cluster $S\in{\cal N}_{i}$, $R^{S}\in{\cal R}_{i}$ is a $t_{{\cal R}}$-step $1$-router for $B_{i,S}=S\cap(A_{i}+\beta\deg_{C})$ where  $\deg_{R^{S}}\le\Delta\cdot B_{i,S}$ entry-wise, $|E(H)|\le\Delta\cdot|\supp(B_{i,S})|$ and $\Delta\le tN^{O(1/t)}\log^{2}N$. 
\item The embedding $\Pi_{{\cal R}_{i}\to G}$ has length $h_{\diam}s_{\decomp}$ and congestion $\kappa_{\Pi}=O(\Delta\omega\log N/\phi)=\Otil(N^{O(1/t)}/\phi)$. 
\end{itemize}
\end{enumerate}
\end{comment}
\end{enumerate}
\item Set $B_{i}\gets A_{i-1}+\beta\deg_{C_{i}}$ and $A_{i}\gets\beta\deg_{C_{i}}$.
\item Set $E'_{i}\gets E({\cal R}_{i}):=\bigcup_{S\in\N_{i}}R^{S}$ scaled the capacity down by $t\kappa_{\Pi}$. 
\end{enumerate}
\item Return $G'$ where $E(G')=\cup_{i}E'_{i}$ and $\ell_{G'}(e)=h_{\Pi}$ for all $e\in E(G')$. 
\end{enumerate}
\caption{\label{alg:emu}$\textsc{Emulator}(G,t,h)$}
\end{algorithm}

\begin{comment}
\textbf{Can we set $\beta$ very close to $1/\phi?$ No, because we want each level to shrink.} 
\end{comment}

Let us start with basic technical observations on \Cref{alg:emu}. 

\begin{prop}
\label{prop:basics}We have the following: 
\begin{enumerate}
\item For all $i$ and $v$, $A_{i}(v)$ and $B_{i}(v)$ are non-negative multiples of $\frac{1}{h}$.%
\begin{comment}
Let $H_{i,S}=\router(B_{i,S},t_{{\cal R}},\kappa_{\rt})$. We have $\deg_{H_{i,S}}\le\Delta\cdot B_{i,S}$ where $\Delta=N^{O(1/t)}$. 
\end{comment}
\item For all $i\ge0$, we have $|A_{i}|\le|A|/N^{i/t}$. In particular, $|A_{2t}|=0$. %
\begin{comment}
For all $i\ge1$, we have $|B_{i}|\le2|A|/N^{(i-1)/t}$. 
\end{comment}
\begin{comment}
\label{enu:total vol}$|\deg_{G'}|\le4\Delta\omega|A|=|A|\cdot N^{O(1/t)}\poly(\log N)$. 
\end{comment}
\item \label{enu:total edge}$|E(G')|\le2t\cdot\mathrm{size}_{{\cal R}}$ where $\mathrm{size}_{{\cal R}}$ upper bounds the total number of edges in the routers. %
\begin{comment}
\label{enu:edge len cap}For all $e\in E(G')$, $u_{G'}(e)\in\frac{1}{h}\mathbb{Z}_{>0}$, given that $u_{G}(e)\in\frac{1}{h}\mathbb{Z}_{>0}$ for all $e\in E(G)$. 
\end{comment}
\end{enumerate}
\end{prop}

\begin{proof}
(1): We assume that $A_{0}$ is integral. For $i\ge1$, $C_{i}(e)$ is a non-negative multiple of $\frac{1}{h_{C}}$ for all $e$ as $C_{i}$ is a $h_{C}$-moving cut. By induction, $A_{i}(v)$ and $B_{i}(v)$ are non-negative multiples of $\frac{\beta}{h_{C}}=\frac{1}{h}$. %
\begin{comment}
(2): By \Cref{lem:router expander} and (1), we have $\deg_{H_{i,S}}\le B_{i,S}\cdot(h|B_{i,S}|)^{O(1/t)}$ entry-wise since $B_{i}(v)\in\frac{1}{h}\mathbb{Z}_{>0}$. The bound on $\Delta$ follows using the trivial bound $|B_{i,S}|\le N^{O(1)}$. 
\end{comment}

(2): For $i=0$, this holds by the assumption. For $i\ge1$, we have that $|C_{i}|\le\phi|A_{i-1}|$ by \Cref{cor:witness ED exist} and so 
\[
|A_{i}|\le\beta|\deg_{C_{i}}|=2\beta|C_{i}|\le2\beta\phi|A_{i-1}|=|A_{i-1}|/N^{1/t}\le|A|/N^{i/t}.
\]
by the choice of $\phi$ and by induction hypothesis. Since we have $|A_{2t}|\le N^{1-2t/t}<\frac{1}{h}$, it follows that $|A_{2t}|=0$ by (1).

(3): This is because there are at most $2t$ levels.
\end{proof}

We show that any demand in $G$ can be routed in $G'$ with small congestion, length, and steps. This is the key technical lemma and we defer the proof to \Cref{sec:emu proof forward map}. 
\begin{lem}
[Forward Mapping]\label{lem:forward map}Let $F$ be a flow in $G$ where $D_{F}$ respects $A$ such that $\congest_{F}=1$ and $\leng_{F}\le h$. There is a flow $F'$ routing $D_{F}$ in $G'$ with $\congest_{F'}\le t\kappa_{\Pi}=N^{O(1/t)}\poly(t\log N)$ and $\step_{F'}\le O(t\cdot t_{{\cal R}})=O(t^{2})$. 
\end{lem}

Next, we show an embedding from $G'$ into $G$. 
\begin{lem}
[Backward Mapping]\label{lem:backward map}There exists an embedding $\Pi_{G'\rightarrow G}$ from $G'$ into $G$ with length slack $1$ and congestion $1$.%
\end{lem}

\begin{proof}
For each level $i$, there is an embedding $\Pi_{{\cal R}_{i}\rightarrow G}$ from ${\cal R}_{i}$ into $G$ with congestion $\kappa_{\Pi}$ and length $h_{\Pi}$. Recall that $E(G')=\cup_{i=1}^{2t}E({\cal R}_{i})$ where the capacity is scaled down by $t\kappa_{\Pi}$. Define the embedding $\Pi_{G'\rightarrow G}=\sum_{i}\Pi_{{\cal R}_{i}\rightarrow G}/(\kappa_{\Pi}t)$ scaled down by $t\kappa_{\Pi}$. 

The length slack of $\Pi_{G'\rightarrow G}$ is $1$ because the length of $\Pi_{{\cal R}_{i}\rightarrow G}$ is $h_{\Pi}$ for every $i$ but we have $\ell_{G'}(e)=h_{\Pi}$ for all $e\in E(G')$. The congestion of $\Pi_{G'\rightarrow G}$ is $1$ because the congestion of $\sum_{i}\Pi_{{\cal R}_{i}\rightarrow G}$ is at most $\kappa_{\Pi}t$ but we scaled down by the flow by $\kappa_{\Pi}t$. 
\end{proof}
Now, we are ready to prove \Cref{thm:h leng emu}.

\paragraph{Proof of \Cref{thm:h leng emu}.}
\begin{proof}
Let $G'=\textsc{Emulator}(G,t,h)$ be the output of \Cref{alg:emu}. By \Cref{lem:forward map} and \Cref{lem:backward map} immediately implies that $G'$ is an $h$-length-constrained $O(t^{2})$-step emulator for $G$ with congestion $t\kappa_{\Pi}=\poly(t\log N)N^{O(1/t)}/\phi=\poly(t\log N)N^{O(1/t)}$ because $\phi^{-1}=2\beta N^{1/t}=\Otil(t^{2}N^{O(1/t)})$. The length slack is $s=t^{2}$ because edges in $G'$ have length $h_{\Pi}=t^{2}h$ by construction. The bound on $|E(G')|$ follows \Cref{prop:basics} (\Cref{enu:total edge} and \Cref{cor:witness ED exist}.
\end{proof}

\subsection{Proof of \Cref{lem:forward map}: Forward Mapping}

\label{sec:emu proof forward map}

\paragraph{Strategy.}

Our strategy is to construct the flow $F'$ that routes $D_{F}$ in $G'$ incrementally. More concretely, let $D_{0}=D_{F}$. For each $i\ge1$, we will construct a flow $F'_{i}$ that partially routes $D_{i-1}$ in $G'$ using only edges from $E'_{i}$ so that the remaining demand is $D_{i}$. After $i>2t$, we have $D_{i}=0$, i.e., there is no remaining demand. By combining and concatenating these flows $F'_{i}$ for all $i\ge1$, we will obtain $F'$ routing $D_{F}$ in $G'$ with the desired properties.

We will maintain the following invariant, for all $i\ge0$, 
\begin{enumerate}
\item $D_{i}$ is $A_{i}$-respecting, and 
\item $D_{i}$ is routable in $G$ with congestion $1$ and length $h$. 
\end{enumerate}
Let us check that the invariant holds for $i=0$. First, $D_{0}$ respects $A_{0}$ because $D_{F}$ respects $A$ by assumption. Second, $D_{0}$ is routable in $G$ with congestion $1$ and length $h$ because $\congest_{F}=1$ and $\leng_{F}\le h$ by assumption. For $i\ge1$, assuming that the invariant holds for $i-1$, we will construct the flow $F'_{i}$ that partially routes $D_{i-1}$ so that the invariant holds for $i$. We will then argue why the invariant for all $i$ implies that our final flow $F'$ has the desired properties.

\paragraph{Construct $F'_{i}$.}

The high-level idea is that we try to send a packet from $v$ to $w$ for each demand pair $(v,w)$ of $D_{i-1}$. If $\dist_{G-C_{i}}(v,w)\le h_{\cov}$, then we will successfully route this packet via some router and be done with it. Otherwise, $\dist_{G-C_{i}}(v,w)>h_{\cov}$. In this case, for each $(v,w)$-flow path $P$, we will carefully identify a set of vertices $X_{v,P}$ and fractionally route the packet from $v$ to $X_{v,P}$. Similarly, from another end, the packet $w$ is routed to a set $X_{w,P}$ that we carefully define. We think of $v$ ``forward'' the packet to $X_{v,P}$ and $w$ ``forward'' the packet to $X_{w,P}$. The demand between $X_{v,P}$ and $X_{w,P}$ will induce the demand $D_{i}$ in the next level. Below, we explain this high-level idea in detail.

For each demand pair $(v,w)$ of $D_{i-1}$ where $\dist_{G-C_{i}}(v,w)\le h_{\cov}$, there must exist a cluster $S\in\cluster(\N_{i})$ where $v,w\in S$. We \emph{assign} the pair $(v,w)$ to such arbitrary cluster $S$. For each cluster $S\in\cluster(\N_{i})$, we route in $G'$ all demands $D_{i-1}(v,w)$ for all pairs $(v,w)$ assigned to $S$ using the edges of router $R^{S}$ for $B_{i,S}$. Let $F_{S}^{\done}$ denote the flow in $G'$ between vertices in $S$ induced by the above routing. Let $D_{S}^{\mathrm{done}}$ be the demand routed by $F_{S}^{\done}$.

Now, we take care of the remaining demand pair $(v,w)$ of $D_{i-1}$ where $\dist_{G-C_{i}}(v,w)>h_{\cov}=4h$. Let $F_{i-1}$ be a flow routing $D_{i-1}$ in $G$ with congestion $1$ and length $h$, whose existence is guaranteed by the invariant. Consider any $(v,w)$-flow-path $P$ of $F_{i-1}$. Since $\dist_{G-C_{i}}(v,w)>4h$ but $\l_{G}(P)\le\leng_{F_{i-1}}\le h$, $C_{i}$ must increase the length $P$ by least $3h$. %
\begin{comment}
Recall that the length increase on edge $e$ by $C_{i}$ is $C_{i}(e)h_{C}$ as $C_{i}$ is a $h_{C}$-length moving cut.
\end{comment}

Let $E_{v,P}$ the minimal edge set in $P$ closest to $v$ such that the total length increase by $C_{i}$ is at least $h$, i.e., $\sum_{e\in E_{v,P}}C_{i}(e)h_{C}\ge h$. Next, we define the vertex set $X_{v,P}\subseteq P$ as follows. For each $(x,y)\in E_{v,P}$ where $x$ is closer to $v$, if $C_{i}(x,y)>0$, we include $x$ into set $X_{v,P}$. Let us denote $e_{x}=(x,y)\in E_{v,P}$ as the edge corresponding to $x\in X_{v,P}$. Observe that, for any $x\in X_{v,P}$, 
\[
\dist_{G-C_{i}}(v,x)\le\l_{G}(P)+h\le2h\le h_{\cov}
\]
That is, $X_{v,P}\subseteq\ball_{G-C_{i}}(v,h_{\cov})$ and so there exists a cluster $S\in\cluster(\N_{i})$ containing both $v$ and $X_{v,P}$. For each $x\in X_{v,P}$, we route flow in $G'$ from $v$ to $x$ of value at most 
\[
F_{i-1}(P)\cdot\frac{C_{i}(e_{x})h_{C}}{h}
\]
via router $R^{S}$. Note that $v,x\in R^{S}$ which is a router for $S\cap(A_{i-1}+\beta\deg_{C_{i}})$ because $A_{i-1}(v)>0$ and $\deg_{C_{i}}(x)>0$. Since $\sum_{x\in X_{v,P}}C_{i}(e_{x})h_{C}\ge h$ by definition, we can route flow from $v$ to $X_{v,P}$ of total value $F_{i-1}(P)$. Symmetrically, we define $E_{w,P}$ and $X_{w,P}$ . Observe that $X_{v,P}$ and $X_{w,P}$ are disjoint because they are defined based on vertices closest to $v$ and $w$, respectively. We route flow in $G'$ from $w$ to $X_{w,P}$ of total value $F_{i-1}(P)$ via $R^{S}$ where $S\in\cluster(\N_{i})$ is a cluster containing both $w$ and $X_{w,P}$%
\begin{comment}
Let $x_{v}$ be the vertex in $P$ closest to $v$ where $\deg_{C_{i}}(x_{v})>0$. Let $P_{v}$ denote the subpath of $P$ from $v$ to $x_{v}$. By the choice of $x_{v}$, $C_{i}$ cannot increase the length of any edge in $P_{v}$ and so $\dist_{G-C_{i}}(v,x_{v})\le\l_{G}(P_{v})\le h=h_{\cov}$. So there exists a cluster $S\in\cluster(\N_{i})$ containing both $v$ and $x_{v}$. We route in $G'$ from $v$ to $x_{v}$ with value $F_{i-1}(P)$ via $\router(B_{i,S},t_{{\cal R}},\kappa_{\rt})$. Symmetrically, we define $x_{w}$ and $P_{w}$.

Then, we route in $G'$ from $w$ to $X_{w,P}$ in the symmetric way. 
\end{comment}

For each cluster $S\in\cluster(\N_{i})$, let $F_{S}^{\forward}$ be the flow in $G'$ induced by the routing described above, which routes from vertices positive demand in $D_{i-1}$ to the vertices incident to the cut $C_{i}$. Let $D_{S}^{\forward}$ be the demand routed by $F_{S}^{\forward}$.

Finally, we define the flow $F'_{i}=\sum_{S\in\cluster(\N_{i})}F_{S}^{\done}+F_{S}^{\forward}$ in $G'$ that routes $D_{S}^{\mathrm{done}}$ and $D_{S}^{\mathrm{forward}}$ for all $S\in\cluster(\N_{i})$ in the manner described above.

\paragraph{$D_{S}^{\mathrm{done}}$ and $D_{S}^{\mathrm{forward}}$ are $B_{i,S}$-respecting.}

Here, we argue that both $D_{S}^{\mathrm{done}}$ and $D_{S}^{\mathrm{forward}}$ are $B_{i,S}$-respecting. This is useful because, by \Cref{lem:router expander}, it means that both $D_{S}^{\mathrm{done}}$ and $D_{S}^{\mathrm{forward}}$ can be routed in $R^{S}\subseteq G'$ using $t_{{\cal R}}$ steps and congestion $1$.

It suffices to show that $D_{S}^{\mathrm{done}}$ and $D_{S}^{\mathrm{forward}}$ are $B_{i}$-respecting because their supports are only on the pairs of vertices inside $S$. This is easy to argue for $D_{S}^{\mathrm{done}}$. We have $D_{S}^{\mathrm{done}}\le D_{i-1}$ (entry-wise), $D_{i-1}$ is $A_{i-1}$-respecting, and $A_{i-1}\le B_{i}$ (entry-wise). So $D_{S}^{\mathrm{done}}$ is $B_{i}$-respecting. Next, we analyze $D_{S}^{\mathrm{forward}}$. On one hand, the total demand that each vertex $v$ may send out is $\sum_{x}D_{S}^{\mathrm{forward}}(v,x)\le\sum_{x}D_{i-1}(v,x)\le A_{i-1}(v)$ because $D_{i-1}$ is $A_{i-1}$-respecting by the invariant. On the other hand, we claim that the total demand that each vertex $x$ may receive is $\sum_{v}D_{S}^{\mathrm{forward}}(v,x)\le\beta\deg_{C_{i}}(x).$ Since $B_{i}=A_{i-1}+\beta\deg_{C_{i}}$, we also have $D_{S}^{\mathrm{forward}}$ is $B_{i}$-respecting.

Now, we prove the claim. The key observation is $\sum_{v}D_{S}^{\mathrm{forward}}(v,x)$ is at most 
\[
\sum_{e:\text{incident to }x}F_{i-1}(e)\frac{C_{i}(e)h_{C}}{h}.
\]
This is because, by construction of the flow $F_{S}^{\forward}$, $x$ may only receive the demand of quantity at most $F_{i-1}(P)\frac{C_{i}(e)h_{C}}{h}$ whenever $x\in P$ and $e\in P$ is incident to $x$. But we have

\[
\sum_{e:\text{incident to }x}F_{i-1}(e)C_{i}(e)\frac{h_{C}}{h}\le\sum_{e:\text{incident to }x}u_{G}(e)C_{i}(e)\beta=\beta\deg_{C_{i}}(x)
\]
because $F_{i-1}$ has congestion $1$ in $G$ and so $F_{i-1}(e)\le u_{G}(e)$. Also, $\beta=\frac{h_{C}}{h}$ by definition. This finishes the claim.

\paragraph{Define $D_{i}$ and Prove the Invariant.}

We define $D_{i}$ as the remaining demand after routing $F'_{i}$. Observe that the remaining demand is as follows. For each demand pair $(v,w)$ of $D_{i-1}$ where $\dist_{G-C_{i}}(v,w)>h_{\cov}$ and each $(v,w)$-flow-path $P$ of $F_{i-1}$, we need route flow of value from $X_{v,P}$ to $X_{w,P}$ of total value $F_{i-1}(P)$. Then, $D_{i}$ sums up these demand.

Now, we argue that $D_{i}$ satisfies the invariant. Consider the congestion and length required for routing $D_{i}$ in $G$. Observe that $D_{i}$ can be routed through subpaths of the flow-paths of $F_{i-1}$ and $F_{i-1}$ is routable with congestion 1 and length $h$ in $G$. Therefore, $D_{i}$ is routable in $G$ with the same congestion and length bound.

Next, we show that $D_{i}$ is $A_{i}$-respecting. For $x\in\supp(A_{i})$, observe that the total demand that $x$ sends out in $D_{i}$ is the same as the total demand $x$ receives in $\sum_{S}D_{S}^{\mathrm{forward}}$. But we showed that this is at most $\sum_{e:\text{incident to }x}F_{i-1}(e)\frac{C_{i}(e)h_{C}}{h}\le\beta\deg_{C_{i}}(x)=A_{i}(x)$ by the definition of $A_{i}$. This completes the proof why the invariant holds.

\paragraph{Construct $F'$ and Bound its Quality.}

The flow $F'$ is obtained by combining and concatenating the flows that route $D_{S}^{\mathrm{done}}$ and $D_{S}^{\mathrm{forward}}$ overall $S\in\cluster(\N_{i})$ for all level $i$ in a natural way so that $F'$ routes $D_{F}$. That is, in level $i$, each demand of $D_{i-1}$ is routed either successfully routed in $G'$ via some router in $t_{{\cal R}}$ steps, or forwards to a new demand in $D_{i}$ by routing in $G'$ via some router in $t_{{\cal R}}$ steps as well. As $i\le2t$, the maximum step of $F'$ is $O(t\cdot t_{{\cal R}})$.

Now, we bound the congestion on $G'$. $F'$ simultaneously routes, for all $i$, $D_{S}^{\mathrm{done}}$ and $D_{S}^{\mathrm{forward}}$ for each $S\in\cluster(\N_{i})$ on $R^{S}$. Since $D_{S}^{\mathrm{done}}$ and $D_{S}^{\mathrm{forward}}$ are $B_{i,S}$-respecting, the congestion for routing both $D_{S}^{\mathrm{done}}$ and $D_{S}^{\mathrm{forward}}$ on $R^{S}$ is at most $1$. But each router is edge-disjoint from each other, so the congestion for routing $F'$ on $\cup_{i}E({\cal R}_{i})=\cup_{i}\cup_{S\in\N_{i}}R^{S}$ is at most $1$. Since $E(G')=\cup_{i}E({\cal R}_{i})$ scaled down the capacity by $t\kappa_{\Pi}$, the congestion for routing $F'$ in $G'$ is then $t\kappa_{\Pi}$.

To summarize, we have successfully constructed a flow $F'$ routing $D_{F}$ in $G'$ with $\congest_{F'}\le t\kappa_{\Pi}$ and $\step_{F'}\le O(t\cdot t_{{\cal R}})$ as desired.

\subsection{Reduction to Witnessed Expander Decomposition}

Observe that we can state \Cref{alg:emu} as a reduction from length-constrained low-step emulators to witnessed expander decomposition, since the algorithm simply compute witnessed expander decomposition $O(t)$ times. This can be formalized as follows. 
\begin{cor}
\label{cor:reduc LC emu}Suppose there is an algorithm ${\cal A}$ that, given an arbitrary node-weighting $A_{i-1}$ of graph $G$, computes a witnessed expander decomposition $C_{i}$ and $({\cal N},\mathcal{R},\Pi_{{\cal R}\to G})$ with parameters $(h_{\cov},\phi,\beta,s_{\decomp})$ such that
\begin{itemize}
\item $C_{i}$ is a $h_{C}$-moving cut
\item $({\cal N}_{i},\mathcal{R}_{i},\Pi_{{\cal R}_{i}\to G})$ is a \emph{$(h_{\cov},t_{{\cal R}},h_{\Pi},\kappa_{\Pi})$}-witness of $A_{i-1}+\beta\deg_{C_{i}}$ in $G-C_{i}$.
\item $\mathcal{R}_{i}$ has total number of edges at most $\mathrm{size}_{{\cal R}}$ and the embedding $\Pi_{{\cal R}_{i}\to G}$ has path count at most $\path_{\Pi}$.
\item $h_{\cov}\gets4h,$ $s_{\decomp}\gets t$, $\phi\gets1/2\beta N^{1/t}$, and $\beta$ is chosen such that $\beta=h_{C}/h$.
\end{itemize}
Then, there is an algorithm for computing an $h$-length-constrained $O(t\cdot t_{{\cal R}})$-step emulator $G'$ for $G$ with length slack $O(t\cdot t_{{\cal R}})$ and congestion slack $t\kappa_{\Pi}$. The emulator contains $2t\cdot\mathrm{size}_{{\cal R}}$ edges and the embedding $\Pi_{G'\rightarrow G}$ has path count $2t\cdot\path_{\Pi}$.

The algorithm makes $O(t)$ calls to ${\cal A}$ and spend additional work of $O(|E(G')|)$ and depth $O(\log n)$. 
\end{cor}

\Cref{thm:h leng emu} is then obtained simply by plugging the existential result of witnessed expander decomposition with parameters $(h_{\cov},\phi,\beta,s_{\decomp})$. From \Cref{cor:witness ED exist}, we need to set $\beta=t^{2}$ so that $\beta=h_{C}/h$. Thus, we get 
\begin{align*}
t_{{\cal R}} & =t\\
\kappa_{\Pi} & =\Otil(N^{O(1/t)}/\phi)=N^{O(1/t)}\poly(\log N)\\
\mathrm{size}_{{\cal R}} & =n\cdot N^{O(1/t)}\poly(\log N),
\end{align*}
which implies \Cref{thm:h leng emu} by \Cref{cor:reduc LC emu}. 

From this reduction, we also immediately obtain an algorithmic result (\Cref{thm:h leng emu alg}) by plugging in the algorithmic witnessed expander decomposition from \Cref{thm:witness ED alg} with parameters $(h_{\cov},\phi,\beta,s_{\decomp})$ into \Cref{cor:reduc LC emu}. Define $t=1/\eps$. We need to set $\beta=\exp(\poly(1/\eps))$ so that $\beta=h_{C}/h$. Thus, we get
\begin{align*}
t_{{\cal R}} & =t\\
\kappa_{\Pi} & =N^{\poly\eps}/\phi=N^{\poly\eps}\\
\path_{\Pi},\mathrm{size}_{{\cal R}} & =m\cdot N^{O(1/t)}\poly(\log N),
\end{align*}
which implies \Cref{thm:h leng emu alg}.
\section{Bootstrapping Length-Constrained Low-Step Emulators}
\label{sec:bootstrap-emu}

In this section, we establish efficient algorithms and representations for length-constrained low-step emulators. At a high level, this requires overcoming two barriers, the first technical and the second conceptual. 
\begin{enumerate}
\item Efficiency: The $(h,s)$-length $(\phi,\kappa)$-expander decomposition algorithm  has a polynomial-in-$h$ dependency in the running time. Even in unit-length graphs, $h$ can be as large as $n$, which is prohibitively slow. Nevertheless, this is a purely technical issue, as a fast decomposition algorithm for any $h$ can still exist. In fact, we could obtain such an algorithm using the techniques developed in our paper. 
\item Representation size: Even if there was a fast $(h,s)$-length $(\phi,\kappa)$-expander decomposition algorithm for any value of $h$, the additional embedding $|\Pi_{\mathcal{R}\to G}|$ may have size at least $hn$, since we need to embed at least a linear number of paths as per \Cref{thm:witness ED alg}, and each path of length $h$ can consist of $h$ many edges. Conceptually, this barrier appears unavoidable with explicit embeddings as required by $\Pi_{\mathcal{R}\to G}$. In this section, we bypass this issue by \emph{stacking} emulators on top of each other in a hierarchical fashion. More precisely, each emulator may embed not only into the original graph, but into previously computed emulators. This way, an $h$-length path can be implicitly represented across multiple levels of stacking: an edge can embed into an emulator at a previous level, whose path in this emulator contains edges that embed further into previous emulators, and so on. 
\end{enumerate}
\begin{thm}
[Bootstrapping]\label{thm:stacking} Let $h,d,t$ be given parameters. Suppose there exists $\alpha>1$ and an algorithm $\mathcal{A}$ that, given an $m$-edge $n$-vertex graph $G$ and $h'\le(2t+3)h$, computes an $h'$-length-constrained $t$-step emulator $G'$ of $G$ with length slack $s$, congestion slack $\kappa$, number of edges at most $|E(G')|\le\alpha|E(G)|$, and path count of the embedding $|\path(\Pi_{G'\rightarrow G})|\le\alpha|E(G)|$.

Then, there is an algorithm that, given graph $G$ and $(h,d,t,h_{0})$ as parameters where $h_{0}\le h$, computes graphs $G'_{0},G'_{1},G'_{2},\ldots,G'_{d}$ such that for each index $i$, 
\begin{enumerate}
\item $G'_{i}$ is an $h_{0}h^{i}$-length-constrained $t$-step emulator of $G$ with length slack $s^{i+1}(2t+3)^{i}$ and congestion slack $(2\kappa)^{i+1}$.\label{item:stacking-1} 
\item $G'_{i}$ has at most $(2\alpha)^{i+1}m$ edges.\label{item:stacking-2} 
\item There is an embedding $\Pi_{G'_{i}\rightarrow G\cup G'_{i-1}}$ that embeds $G'_{i}$ into $G\cup G'_{i-1}$ with length slack $1$, congestion $1$, path count $(2\alpha)^{i+1}m$, and maximum step at most $(2st+3s)h$.\footnote{We define $G'_{-1}=\emptyset$.} Moreover, $\Pi_{G'_{i}\rightarrow G\cup G'_{i-1}}$ only routes through edges of $G$ with length in the range $(h_{0}h^{i-1},h_{0}h^{i}]$. \label{item:stacking-3} 
\end{enumerate}
The algorithm calls $\mathcal{A}$ on $d$ many graphs, each with at most $O(d(2\alpha)^{d}m)$ edges, and, outside these calls, runs in $O(d(2\alpha)^{d}m)$ work and $\Otil(d)$ depth.
\end{thm}

Before proving \Cref{thm:stacking}, we explain why \Cref{item:stacking-3} is important; it is crucial in the lemma below. 
\begin{lem}
\label{lem:fast flow map}For any $i\in[d]$, for any flow $F'$ in $G'_{i}$, there is a flow $F$ in $G$ routing the same demand where $\congest_{F}\le\congest_{F'}$ and $\leng_{F}\le\leng_{F'}$. Given the edge representation $\flow_{F'}$ of $F'$, we can compute the edge representation $\flow_{F}$ of $F$ in $O(hst\cdot(2\alpha)^{i+1}m)$ work and $\Otil(i)$ depth. 
\end{lem}

\begin{proof}
By scaling, we can assume that $F'$ has congestion and length $1$ in $G'_{i}$. We will construct the edge representation $\flow_{F}$ of $F$ with congestion and length $1$. Although the existence of $F$ follows immediately because $G'_{i}$ embeds into $G$ with congestion $1$ and length slack $1$, below we will how to construct the flow $F$ inductively level by level for the efficiency on constructing the edge representation $\flow_{F}$ of $F$.

Define $F'_{i}\gets F'$. Let $\flow_{F'_{i}}\gets\flow_{F'}$ be the edge representation of $F'_{i}$. We construct a flow $\Fhat_{i}$ in $G\cup G'_{i-1}$ as follows: for each directed edge $(v,w)\in\overleftrightarrow{E}(G'_{i})$, $\Fhat_{i}$ routes $\flow_{F'_{i}}(v,w)$ units of flow through the $(v,w)$-flow-paths of $\Pi_{G'_{i}\rightarrow G\cup G'_{i-1}}$. By \Cref{thm:stacking}(\ref{item:stacking-3}), we have $|\path(\Pi_{G'_{i}\rightarrow G\cup G'_{i-1}})|\le(2\alpha)^{i+1}m$ and $\step_{\Pi_{G'_{i}\rightarrow G\cup G'_{i-1}}}\le(2st+3s)h$. Thus, we can explicitly compute the path decomposition of $\Fhat_{i}$ as well as the its edge representation of $\flow_{\Fhat_{i}}$ in $O(hst\cdot(2\alpha)^{i+1}m)$ work and $\Otil(1)$ depth by explicitly summing the flow values overall flow paths of $\Pi_{G'_{i}\rightarrow G\cup G'_{i-1}}$.

Given $\Fhat_{i}$, we now define $F_{i}$ and $F'_{i-1}$ as the flow $\Fhat_{i}$ restricted to edges $G$ and $G'_{i-1}$, respectively. More precisely, for each directed edge $(v,w)\in\overleftrightarrow{E}(G)$ where $\Fhat_{i}((v,w))>0$, $F_{i}$ routes $\flow_{\Fhat_{i}}(v,w)$ units of flows through a single edge $(v,w)$. Similarly, for each directed edge $(v,w)\in\overleftrightarrow{E}(G'_{i-1})$ where $\Fhat_{i}((v,w))>0$, $F'_{i-1}$ routes $\flow_{\Fhat_{i}}(v,w)$ units of flows through $(v,w)$. Since each flow path of $F_{i}$ and $F'_{i-1}$ routes through a single (directed) edge, the edge representations $\flow_{F_{i}}$ and $\flow_{F'_{i-1}}$ can be trivially computed in linear work and $\Otil(1)$ depth.

Next, we repeat the same argument on the edge representation $\flow_{F'_{i-1}}$ of $F'_{i-1}$ and until $i=0$. Finally, we obtain the edge representations $\flow_{F_{i}},\flow_{F_{i-1}},\dots,\flow_{F_{0}}$ of $F_{i},F_{i-1},\dots,F_{0}$ in $G$, respectively, such that, after concatenating the flow paths of $F_{i},F_{i-1},\dots,F_{0}$, we can obtain the flow $F$ that routes exactly the same demand as $F'$. We also have that the edge representation of $\flow_{F}=\flow_{F_{i}}+\flow_{F_{i-1}}+\dots+\flow_{F_{0}}$. The total cost for constructing $\flow_{F}$ is then $O(hst\cdot(2\alpha)^{i+1}m)$ work and $\Otil(i)$ depth as desired.

Note that we have $\leng_{F}\le1$ because $\Pi_{G'_{i}\rightarrow G\cup G'_{i-1}}$ has length slack $1$ for all $i$. Moreover, $\congest_{F}\le1$ because $\Pi_{G'_{i}\rightarrow G\cup G'_{i-1}}$ has congestion $1$ and each $F_{i}$ only route through edges of $G$ with length in the range $(h_{0}h^{i-1},h_{0}h^{i}]$ by \Cref{thm:stacking}(\ref{item:stacking-3}). 
\end{proof}
By plugging the algorithm from \Cref{thm:h leng emu alg} into \Cref{thm:stacking} and \Cref{lem:fast flow map}, we obtain immediately the following, which will be used in the next section. 
\begin{cor}
\label{cor:stacking algo}For any $\eps \in (\log^{-c}, 1)$ for some small enough constant $c$, there are parameters $t=\exp(\poly1/\eps)$ and $\gamma=N^{\poly\eps}$ such that there is a parallel algorithm that, given a $m$-edge graph $G$, and $(h,d,\eps,h_{0})$ as parameters where $h_{0}\le h$, computes graphs $G'_{0},G'_{1},G'_{2},\ldots,G'_{d}$ such that for each index $i\le d$, $G'_{i}$ is an $(h_{0}h^{i})$-length-constrained $t$-step emulator of $G$ containing $m\gamma^{i}$ edges with length slack $O(t^{2}){}^{i}$ and congestion slack $\gamma^{i}$. The algorithm has $m\gamma^{d}\poly(h)$ work and $\Otil(d\cdot\poly(h))$ depth.

For any $i$, given an edge representation of $F'$ in $G'_{i}$, one can compute the edge representation $\flow_{F}$ of $F$ in $G$ that routes the same demand where $\congest_{F}\le\congest_{F'}$ and $\leng_{F}\le\leng_{F'}$ in $m\gamma^{d}h$ work and $\Otil(d)$ depth.
\end{cor}

The rest of this section is for proving \Cref{thm:stacking}.

\subsection{Construction and Analysis}

\begin{algorithm}
\begin{enumerate}
\item Let $G'_{0}$ be an $h_{0}$-length-constrained $t$-step emulator obtained by calling algorithm $\mathcal{A}$ on $G_{\le h_{0}}$, i.e., the graph containing only edges in $G$ of length at most $h_{0}$. 
\item For $1\le i\le d$: 
\begin{enumerate}
\item Let $H_{i}$ be the unit-length graph constructed as follows: 
\begin{enumerate}
\item For each edge $e$ in $G'_{i-1}$, add a corresponding unit-length edge $e$ in $H_{i}$. 
\item For each edge $e$ in $G$ with length in the range $(h_{0}h^{i-1},h_{0}h^{i}]$, add an edge in $H_{i}$ of length $\lceil\ell_{G}(e)/(h_{0}h^{i-1})\rceil$. 
\end{enumerate}
\item Let $H'_{i}$ be a $(2t+3)h$-length-constrained $t$-step emulator obtained by calling algorithm $\mathcal{A}$ on $H_{i}$. 
\item Let $G'_{i}$ be the graph $H'_{i}$ with all edges modified to have length $h_{0}h^{i}\cdot s^{i+1}(2t+3)^{i}$. 
\end{enumerate}
\end{enumerate}
\caption{\textsc{EmulatorWithBootstrapping}$(G,h,d,t)$}
\end{algorithm}

In this subsection, we prove three properties of \Cref{thm:stacking} by induction on $i\ge0$.

The base case $i=0$ is straightforward. Let $G_{\le h_{0}}$ denote the graph containing only edges in $G$ of length at most $h_{0}$. By the guarantee of ${\cal A}$, $G'_{0}$ is an $h_{0}$-length-constrained $t$-step emulator of $G_{\le h_{0}}$ with length slack $s$ and congestion slack $\kappa$. By definition, $G'_{0}$ is also an $h_{0}$-length-constrained $t$-step emulator of $G$ with the same ganrantees, because any flow $F$ in $G$ of length at most $h_{0}$ may route through only edges of length at most $h_{0}$. Moreover, $|E(G'_{0})|\le\alpha m$ and $\Pi_{G'_{0}\rightarrow G}$ has length slack $1$, congestion $1$, and path count $\alpha m\le2\alpha^{2}m$. The maximum step $\Pi_{G'_{0}\rightarrow G}$ is $\step_{\Pi_{G'_{0}\rightarrow G}}\le\leng_{\Pi_{G'_{0}\rightarrow G}}$ because the edge length of $G$ is integral. We have $\leng_{\Pi_{G'_{0}\rightarrow G}}\le sh_{0}\le(2st+3s)h$ because the length of edges in $G'_{0}$ is $h_{0}s$ and $\Pi_{G'_{0}\rightarrow G}$ has length slack $1$. So $\step_{\Pi_{G'_{0}\rightarrow G}}\le(2st+3s)h$. Finally, $\Pi_{G'_{0}\rightarrow G}$ only routes through edges of $G$ with length at most $h_{0}$, by definition of $G_{\le h_{0}}$.

For the rest of the proof, assume that the three properties hold for iteration $i-1$. The two lemmas below establish the analogues of Lemmas~\ref{lem:forward map}~and~\ref{lem:backward map} from \Cref{sec:emu}. We defer their proofs to \Cref{sec:algorithm-backward-mapping,sec:algorithm-forward-mapping}. 
\begin{lem}
[Backward Mapping]\label{lem:algorithm-backward-mapping}The graph $G'_{i}$ has at most $(2\alpha)^{i+1}m$ edges and there is an embedding $\Pi_{G'_{i}\to G\cup G'_{i-1}}$ with congestion $1$, length slack $1$, path count at most $(2\alpha)^{i+1}m$, and $\step_{\Pi_{G'_{i}\to G\cup G'_{i-1}}}\le(2st+3s)h$. Also, there is an embedding $\Pi_{G'_{i}\to G}$ with congestion $1$ and length slack $1$.
\end{lem}

\begin{lem}
[Forward Mapping]\label{lem:algorithm-forward-mapping} Let $F$ be a flow in $G$ with $\leng_{F}\le h_{0}h^{i}$. There is a flow $F'$ routing $D_{F}$ in $G'_{i}$ with $\congest_{F'}\le(2\kappa)^{i+1}\cdot\congest_{F}$ and $\step_{F'}\le t$. 
\end{lem}

\Cref{lem:algorithm-backward-mapping} immediately implies properties~(\ref{item:stacking-2})~and~(\ref{item:stacking-3}) of \Cref{thm:stacking} for iteration $i$. To see that property~(\ref{item:stacking-1}) is satisfied, we check each requirement in \Cref{def:h-len emu}: 
\begin{enumerate}
\item By construction, edges in $G'_{i}$ have the same length $h_{0}h^{i}\cdot s^{i+1}(2t+3)^{i}$. 
\item Given a flow $F$ in $G$ where $\leng_{F}\le h_{0}h^{i}$, \Cref{lem:algorithm-forward-mapping} guarantees a flow $F'$ in $G'$ routing the same demand where $\congest_{F'}\le(2\kappa)^{i}\cdot\congest_{F}$ and $\step_{F'}\le t$. 
\item By \Cref{lem:algorithm-backward-mapping}, $G'_{i}$ can be embedded into $G$ with congestion $1$ and length slack $1$. 
\end{enumerate}
Finally, we show that the algorithm calls $\mathcal{A}$ on $d$ many graphs, each with at most $O((2\alpha)^{d}m)$ edges, and runs in $O((2\alpha)^{d}m)$ work and $\tilde{O}(1)$ depth outside these calls. On each iteration $1\le i\le d$, we call $\mathcal{A}$ on the graph $H_{i}$ which consists of edges from $G'_{i-1}$ and $G$ (with their lengths modified), which is at most $O((2\alpha)^{d}m)$ edges in total. Outside of this call, the algorithm clearly runs in time linear in $G'_{i-1}$ and $G$, which is $O((2\alpha)^{d}m)$ time. Over the $d$ iterations, the total work outside calls to $\mathcal{A}$ is $O(d(2\alpha)^{d}m)$ and the total depth is $\tilde{O}(d)$.

\subsection{Proof of \Cref{lem:algorithm-backward-mapping}: Backward Mapping}

\label{sec:algorithm-backward-mapping}

First, we show that the number of edges in $G'_{i}$ is at most $(2\alpha)^{i+1}m$. This is because $|E(G'_{i})|=|E(H'_{i})|\le\alpha|E(H_{i})|$ by the guarantee of ${\cal A}$, and $|E(H_{i})|\le|E(G'_{i-1})|+|E(G)|\le(2\alpha)^{i}m+m$. So $|E(G'_{i})|\le\alpha((2\alpha)^{i}+1)m\le(2\alpha)^{i+1}m$. Our next goal is to construct the embedding $\Pi_{G'_{i}\to G\cup G'_{i-1}}$ and $\Pi_{G'_{i}\to G}$ with desired properties. To do this, we will show the following embedding: $\Pi_{G'_{i}\to H'_{i}}$, $\Pi_{H'_{i}\to H_{i}}$, $\Pi_{H_{i}\to G\cup G'_{i-1}}$, and $\Pi_{H_{i}\to G}$. We will obtain the goal by composing them. Recall that all $G'_{i}$ have length $h_{0}h^{i}\cdot s^{i+1}(2t+3)^{i}$ including the case when $i=0$. 
\begin{enumerate}
\item \textbf{From $G'_{i}$ to $H'_{i}$:} Since $G'_{i}$ is the graph $H'_{i}$ with all edge lengths scaled up by factor $h_{0}h^{i-1}s^{i}(2t+3)^{i-1}$, there is a trivial embedding $\Pi_{G'_{i}\to H'_{i}}$ with congestion $1$ and length slack $1/(h_{0}h^{i-1}s^{i}(2t+3)^{i-1})$. 
\item \textbf{From $H'_{i}$ to $H_{i}$: }By the guarantee of algorithm $\mathcal{A}$, it returns an embedding $\Pi_{H'_{i}\to H_{i}}$ with congestion $1$, length slack $1$. The path count $|\path(\Pi_{H'_{i}\to H_{i}})|$ is at most $\alpha|E(H_{i})|\le(2\alpha)^{i+1}m$. Next, we bound $\step_{\Pi_{H'_{i}\to H_{i}}}$. We have $\step_{\Pi_{H'_{i}\to H_{i}}}\le\leng_{\Pi_{H'_{i}\to H_{i}}}$ because the edge length of $H_{i}$ is integral. Also, flow path of $\Pi_{H'_{i}\to H_{i}}$ has length $\leng_{\Pi_{H'_{i}\to H_{i}}}\le(2t+3)hs$ because each edge in $H'_{i}$ has length $(2t+3)hs$ and the length slack of $\Pi_{H'_{i}\to H_{i}}$ is $1$. So $\step_{\Pi_{H'_{i}\to H_{i}}}\le(2t+3)hs$. 
\item \textbf{From $H_{i}$ to $G\cup G'_{i-1}$: }Since $H_{i}$ is the graph $G'_{i-1}$ scaled down by factor $h_{0}h^{i-1}s^{i}(2t+3)^{i-1}$, together with edges in $G$ with length in the range $(h_{0}h^{i-1},h_{0}h^{i}]$ scaled down by factor at most $h_{0}h^{i-1}$. So there is a trivial embedding $\Pi_{H_{i}\to G\cup G'_{i-1}}$ with congestion $1$ and length slack $\max\{h_{0}h^{i-1}s^{i}(2t+3)^{i-1},h_{0}h^{i-1}\}=h_{0}h^{i-1}s^{i}(2t+3)^{i-1}$. 
\item \textbf{From $H_{i}$ to $G$: }We embed $H_{i}$ further to $G$ as follows. We split the trivial embedding $\Pi_{H_{i}\to G\cup G'_{i-1}}$ into an embedding $\Pi_{1}$ from $H_{i}$ to $G'_{i-1}$ and another embedding $\Pi_{2}$ from $H_{i}$ to $G$ that only congests edges of length more than $h_{0}h^{i-1}$, each with congestion $1$ and length slack at most $h_{0}h^{i-1}s^{i}(2t+3)^{i-1}$. By induction, $G'_{i-1}$ can be embedded into $G$ with congestion $1$ and length slack $1$. Actually, we claim the stronger property that $G'_{i-1}$ can be embedded with congestion $1$ and length slack $1$ into the subgraph of $G$ consisting of all edges of length at most $h_{0}h^{i}$. This is because the edges in $G$ of length greater than $h^{i}$ are ignored in the first $i-1$ levels of the construction. So the same inductive statement must hold on the graph with these edges taken out. Let $\Pi_{G'_{i-1}\to G}$ be this strengthened embedding. We compose $\Pi_{1}$ with $\Pi_{G'_{i-1}\to G}$ and then combine it with $\Pi_{2}$, we obtain an embedding $\Pi_{H_{i}\to G}$ with congestion $1$ and length slack $h_{0}h^{i-1}s^{i}(2t+3)^{i-1}$. 
\end{enumerate}
By composing the embedding $\Pi_{G'_{i}\to H'_{i}}$, $\Pi_{H'_{i}\to H_{i}}$, $\Pi_{H_{i}\to G\cup G'_{i-1}}$, we obtain $\Pi_{G'_{i}\to G\cup G'_{i-1}}$ with congestion $1$ and length slack $1/(h_{0}h^{i-1}s^{i}(2t+3)^{i-1})\times1\times h_{0}h^{i-1}s^{i}(2t+3)^{i-1}=1$. Since both $\Pi_{G'_{i}\to H'_{i}}$ and $\Pi_{H_{i}\to G\cup G'_{i-1}}$ are trivial embedding, we have $|\path(\Pi_{G'_{i}\to G\cup G'_{i-1}})|\le(2\alpha)^{i+1}m$ and $\step_{\Pi_{G'_{i}\to G\cup G'_{i-1}}}\le(2t+3)hs$, inheriting the properties of $\Pi_{H'_{i}\to H_{i}}$. By the definition of $H_{i}$, we have $\Pi_{G'_{i}\rightarrow G\cup G'_{i-1}}$ only routes through edges of $G$ with length in the range $(h_{0}h^{i-1},h_{0}h^{i}]$.

The embedding $\Pi_{G'_{i}\to G}$ with congestion $1$ and length slack $1$ is obtained by composing $\Pi_{G'_{i}\to H'_{i}}$, $\Pi_{H'_{i}\to H_{i}}$, and $\Pi_{H_{i}\to G}$.

\subsection{Proof of \Cref{lem:algorithm-forward-mapping}: Forward Mapping}

\label{sec:algorithm-forward-mapping}

Let $F$ be a flow in $G$ with $\leng_{F}\le h_{0}h^{i}$ from the lemma statement. Our goal is to show a flow $F'$ routing $D_{F}$ in $G'_{i}$ with $\congest_{F'}\le(2\kappa)^{i}\cdot\congest_{F}$ and $\step_{F'}\le t$.

\paragraph{Construct flow $F^{*}$ on $H_{i}$.}

We first construct a flow $F^{*}$ that routes $D_{F}$ in $H_{i}$ with $\leng_{F^{*}}\le(2t+3)h$ and $\congest_{F^{*}}\le((2\kappa)^{i}+1)\cdot\congest_{F}$. We start by decomposing the flow-paths in $F$ as follows.

For each flow-path $P$ in $F$, we first break down the path into at most $2h$ segments $P_{1},P_{2},\ldots$ such that each path $P_{i}$ either has length at most $h_{0}h^{i-1}$ or consists of a single edge. To do so, initialize $P'\gets P$ and $i\gets1$, and while $P'$ is non-empty, let path $P_{i}$ be the longest prefix of $P'$ of length at most $h_{0}h^{i-1}$, or the first edge of $P'$ if its length is already greater than $h_{0}h^{i-1}$; then, remove the edges of $P_{i}$ from $P'$ and increment $i$ by $1$. If $P_{i}$ is the longest prefix of $P'$ of length at most $h_{0}h^{i-1}$, and if there is an edge $e$ after $P_{i}$ in $P'$, then the combined length of $P_{i}$ and $e$ is greater than $h_{0}h^{i-1}$, and furthermore, $e$ must be removed on the next iteration. It follows that every two iterations decreases the length of $P'$ by at least $h_{0}h^{i-1}$, and since $P$ has length at most $h_{0}h^{i}$, there are at most $2h$ many paths.

For each path $P_{i}$ of length at most $h_{0}h^{i-1}$, add it to a new flow $F_{1}$, and for the remaining paths $P_{i}$ (consisting of single edges of length greater than $h_{0}h^{i-1}$), add it to a new flow $F_{2}$. Let $F_{1}$ and $F_{2}$ be the final flows after repeating this procedure for all flow-paths $P$ in $F$. By construction, we have $F=F_{1}+F_{2}$, $\congest_{F_{1}}\le\congest_{F}$, $\congest_{F_{2}}\le\congest_{F}$, $\leng_{F_{1}}\le h_{0}h^{i-1}$, and $\step_{F_{2}}=1$, and moreover, each flow-path $P$ in $F$ decomposes into at most $2h$ flow-paths in $F_{1}$ and $F_{2}$.

By induction, property~(\ref{item:stacking-1}) guarantees that $G'_{i-1}$ is an $h_{0}h^{i-1}$-length-constrained $t$-step emulator for $G$ with congestion slack $(2\kappa)^{i}$. Since $\leng_{F_{1}}\le h_{0}h^{i-1}$, there exists a flow $F'_{1}$ in $G'_{i-1}$ routing demand $D_{F_{1}}$ where $\congest_{F'_{1}}\le(2\kappa)^{i}\cdot\congest_{F_{1}}\le(2\kappa)^{i}\cdot\congest_{F}$ and $\step_{F'_{1}}\le t$. Since edges $H_{i}$ have unit length, there is a corresponding flow $F_{1}^{*}$ in $H_{i}$ where $\congest_{F_{1}^{*}}\le(2\kappa)^{i}\cdot\congest_{F}$ and $\leng_{F_{1}^{*}}\le t$.

By construction, each path in $F_{2}$ is a single edge $e$ of length $\ell_{G}(e)\in(h_{0}h^{i-1},h_{0}h^{i}]$, so there is a corresponding edge in $H_{i}$ of length $\lceil\ell_{G}(e)/h_{0}h^{i-1}\rceil\le\ell_{G}(x)/(h_{0}h^{i-1})+1$. Let $F_{2}^{*}$ be the flow in $H_{i}$ that routes each single edge in $F_{2}$ through its corresponding edge in $H_{i}$. By construction, $\congest_{F_{2}^{*}}=\congest_{F_{2}}\le\congest_{F}$.

Finally, we concatenate flows $F_{1}^{*}$ and $F_{2}^{*}$ in $H_{i}$ as follows. For each flow-path $P$ in $F$, consider the decomposition into at most $2h$ segments $P_{1},P_{2},\ldots$. For each path $P_{i}$ of length at most $h_{0}h^{i-1}$, take a corresponding flow in $F_{1}^{*}$ of length at most $t$, and for each single-edge path $P_{i}$ of length greater than $h_{0}h^{i-1}$, take the flow in $F_{2}^{*}$ along its corresponding edge in $H_{i}$. In both cases, the flow in $H_{i}$ has length at most $\max\{t,\ell_{G}(P_{i})/(h_{0}h^{i-1})+1\}$. Concatenating these flows over all $i$ produces a flow for path $P$ of length at most $\sum_{i}\max\{t,\ell_{G}(P_{i})/(h_{0}h^{i-1})+1\}\le2ht+h_{0}h^{i}/h_{0}h^{i-1}+2h=(2t+3)h$. Over all flow-paths $P$, the final flow $F^{*}$ in $H_{i}$ satisfies $\leng_{F^{*}}\le(2t+3)h$ and $\congest_{F^{*}}\le\congest_{F_{1}^{*}}+\congest_{F_{2}^{*}}\le((2\kappa)^{i}+1)\cdot\congest_{F}$.

\paragraph{Use emulator $H'_{i}$.}

Since $H'_{i}$ is a $(2t+3)h$-length-constrained $t$-step emulator of $H_{i}$ with congestion slack $\kappa$, there is a flow $F^{\dagger}$ in $H'_{i}$ routing $D_{F^{*}}=D_{F}$ with $\congest_{F^{\dagger}}\le\kappa\cdot\congest_{F^{*}}\le\kappa\cdot((2\kappa)^{i}+1)\cdot\congest_{F}\le(2\kappa)^{i+1}\cdot\congest_{F}$ and $\step_{F^{\dagger}}\le t$. Finally, since $G'_{i}$ is simply $H'_{i}$ with edge length increased, the flow $F^{\dagger}$ in $H'_{i}$ translates to a flow $F'$ in $G'_{i}$ with $\congest_{F'}\le(2\kappa)^{i}\cdot\congest_{F}$ and $\step_{F'}\le t$ as promised. 

\section{Low-Step Emulators}
\label{sec:low-step emu}
In this section, we define and construct emulators similar to the length-constrained low-step emulators, but they preserve information about general flows instead of length-constrained flows.

To do this, we need a notion of \emph{path-mapping} between flows. For any two flows $F$ in $G$ and $F'$ in $G'$ routing the same demand, observe that there exist path-decomposition of $F$ and $F'$ denoted by ${\cal P}$ and ${\cal P}'$, respectively, and a bijection $\pi:{\cal P}\rightarrow{\cal P}'$ such that for every $P\in{\cal P}$ and $P'=\pi(P)$, we have $\val(F'(P'))=\val(F(P))$. We call $\pi$ a \emph{path-mapping from $F$ to $F'$}. We say that $\pi$ has length slack $s$ if $\l_{G'}(\pi(P))\le s\cdot\l_{G}(P)$.

Intuitively, if $F$ can be mapped to $F'$ via a path-mapping of length slack $s$, then $F'$ is ``as short as'' $F$ up to a factor of $s$. Now we are ready to define our main object. 
\begin{defn}
[Low-Step Emulators]\label{def:emu}Given a graph $G$ and a node-weighting $A$ in $G$, we say that $G'$ is a \emph{$t$-step emulator} of $A$ with length slack $s$ and congestion slack $\kappa$ if 
\begin{enumerate}
\item Given a flow $F$ in $G$ routing an $A$-respecting demand, there is a flow $F'$ in $G'$ routing the same demand where $\congest_{F'}\le\kappa\cdot\congest_{F}$ and $\step_{F'}\le t$. Moreover, there exists a path-mapping from $F$ to $F'$ with length slack $s$. 
\item $G'$ can be embedded into $G$ with congestion $1$ and length slack $1$. 
\end{enumerate}
\end{defn}

\begin{remark}
Let $G'$ is a $t$-step emulator of $A$ in $G$. Then, $G'$ is simultaneous an $h$-length-constrained $t$-step emulator of $A$ in $G$ for all $h$, except that $G'$ does not satisfies Condition \ref{enu:uniform} on length uniformity.
\end{remark}

The main theorems of this section are the construction of low-step emulators. 
\begin{thm}
[Existential]\label{thm:emu exist}Given any graph $G$, a node weighing $A$ and parameter $t$, there exists a $t$-step emulator $G'$ for $A$ in $G$ with length slack $O(t)$, congestion slack $\poly(t\log N)N^{O(1/\sqrt{t})}$, and $|E(G')|\le n\cdot N^{O(1/\sqrt{t})}\poly(\log N)$. 
\end{thm}

Next, we show an algorithmic version of the above theorem when $A=\deg_{G}$ and the number of emulator edges is close to $m$, instead of $n$.

\begin{thm}
[Algorithmic]\label{thm:emu algo}Given any $m$-edge graph $G$ and $\eps \in (\log^{-c} N, 1)$ for some sufficiently small constant $c$, let $t=\exp(\poly(1/\eps))$. There is a parallel algorithm $\textsc{LowStepEmu}(G,\eps)$ that constructs a $t$-step emulator $G'$ for $A$ in $G$ with length slack $\exp(\poly(1/\eps))$, congestion slack $N^{\poly\eps}$, and $|E(G')|\le mN^{\poly\eps}$.

The algorithms takes $mN^{\poly\eps}$ work and $N^{\poly\eps}$ depth. Given an edge representation $\flow_{F'}$ of flow $F'$ in $G'$, there is an algorithm $\textsc{FlowMap}(G',\flow_{F'})$ that computes an edge representation $\flow_{F}$ of flow $F$ in $G$ routing the same demand where $\congest_{F}\le\congest_{F'}$ and $\leng_{F}\le\leng_{F'}$ using $mN^{\poly\eps}$ work and $\tilde{O}(1/\poly\eps)$ depth. 
\end{thm}

It is a natural question to ask if we can obtain above emulator for general $A$ whose number of edges is close to $|\supp(A)|$. We believe this to be possible, but leave this for future work.

In the rest of this section, we prove \Cref{thm:emu exist} and \Cref{thm:emu algo}. We will use the following basic lemma showing how to construct a $t$-step emulator given $2^{i}$-length-constrained $t$-step emulator for each $i$. 
\begin{lem}
\label{lem:reduc to h leng emu}Let $A$ be a a node-weighting of $G$. For $i\in\{1,\dots,d=\log N\}$, let $G'_{i}$ be a $2^{i}$-length-constrained $t$-step emulator of $A$ with length slack $s$ and congestion slack $\kappa$. Let $G'=\biguplus_{i}G'_{i}$ with capacity scaled down by $d$. Then $G'$ is a $t$-step emulator of $A$ with length slack $2st$ and congestion $\kappa d$. 
\end{lem}

\begin{proof}
We will argue the following: 
\begin{enumerate}
\item Given a flow $F$ in $G$ routing an $A$-respecting demand with congestion $1$, there is a flow $F'$ in $G'$ routing $D_{F}$ where $\congest_{F'}\le\kappa d$ and $\step_{F'}\le t$. Moreover, there exist a path-mapping between $F$ and $F'$ with length slack $2st$. 
\item $G'$ can be embedded into $G$ with congestion $1$ and length slack $1$. 
\end{enumerate}
For the first point, given $F$, we decompose $F=F_{1}+\dots+F_{d}$ so that $F_{i}$ contains all flow paths $P$ of $F$ where $2^{i-1}\le\ell(P)<2^{i}$ for $i\ge1$. For each flow $F_{i}$ in $G$, by \Cref{def:h-len emu}, there is a flow $F'_{i}$ routing $D_{F_{i}}$ in $G'_{i}$ with $\congest_{F'_{i}}\le\kappa$, $\step_{F'_{i}}\le t$, and $\leng_{F'_{i}}\le st\cdot2^{i}$. Define $F'=F'_{1}+\dots+F'_{d}$ as a flow in $G'$. We now bound the congestion, step, and length slack of $F'$, respectively. We have $\congest_{G',F'}\le d\cdot\max_{i}\congest_{G'_{i},F'_{i}}\le\kappa d$ because $G'$ is a disjoint union of $G'_{i}$ after scaling down the capacity by $d$. Also, we have $\step_{F'}\le t$ as $\step_{F'_{i}}\le t$ for all $i$. To bound the length slack, since $F_{i}$ and $F'_{i}$ route the same demand, there exists a path mapping $\pi_{i}$ from $F_{i}$ to $F_{i}'$. Since all flow paths of $F_{i}$ and $F_{i}'$ have length at least $2^{i-1}$ and less than $st\cdot2^{i}$ respectively, the length slack of $\pi_{i}$ is at most $2st$. From $\pi_{1},\dots,\pi_{d}$, there exists a natural path-mapping $\pi$ from $F$ to $F'$ with length slack $2st$. This completes the proof of the first point.

For the second point, for each $i$, let $\Pi_{G'_{i}\rightarrow G}$ be the embedding from $G'_{i}$ into $G$ with congestion $1$ and length slack $1$. Observe that $\Pi_{G'\rightarrow G}=\frac{1}{d}\biguplus_{i}\Pi_{G'_{i}\rightarrow G}$ is an embedding of $G'$ (which is $\biguplus_{i}G'_{i}$ after scaling down the capacity by $d$) into $G$ with congestion $d/d=1$ and length slack $1$. 
\end{proof}

\paragraph{Existential Emulators: Proof of \Cref{thm:emu exist}.}

From \Cref{thm:h leng emu}, for any $i$, there exists a $2^{i}$-length-constrained $t$-step emulator $G'_{i}$ for $A$ in $G$ with with length slack $s=O(t)$ and congestion slack $\kappa=\poly(t\log N)N^{O(1/\sqrt{t})}$ where $|E(G'_{i})|\le nN^{O(1/\sqrt{t})}\poly(\log N).$

By plugging $G'_{1},\dots,G'_{\log N}$ into \Cref{lem:reduc to h leng emu}, we obtain a $t$-step emulator $G'$ for $A$ in $G$ with length slack $2st=O(t^{2})$ and congestion slack $\kappa d=\poly(t\log N)N^{O(1/\sqrt{t})}$. Since $G'=\biguplus_{i}G'_{i}$ with capacity scaled down, the bound of $|E(G'_{i})|$ follow.

\paragraph{Algorithmic Emulators: Proof of \Cref{thm:emu algo}.}

Given $\eps \in (\log^{-c}, 1)$ for some small enough constant $c$, let $t=\exp(\poly(1/\eps))$ and $\gamma=N^{\eps^{c_{0}}}$ be the parameters from \Cref{cor:stacking algo} for some constant $c_{0}>0$. Set $\eps'\gets\eps^{c_{0}/2}$ and $h\gets N^{\eps'}$ as the stacking parameter. We round up $h$ so that it is a power of $2$. For every $h_{0}=2^{j}$ where $h_{0}\le h$, we do the following. Let $d'=\log_{h}N=1/\eps'$. Therefore, $\gamma^{d'}\le N^{(\eps^{c_{0}})/\eps^{c_{0}/2}}=N^{\eps^{c_{0}/2}}=N^{\poly\eps}$. We will exploit this inequality.

Construct $G'_{j,0},G'_{j,1},\ldots,G'_{j,d'}$ via \Cref{cor:stacking algo} such that for each index $k$, $G'_{j,k}$ is an $(h_{0}h^{k})$-length-constrained $t$-step emulator of $G$. Each emulator $G'_{j,k}$ contains $|E(G'_{j,k})|\le m\gamma^{d'}=m\cdot N^{\poly\eps}$ edges and has with length slack at most $O(t^{2}){}^{d'}=\exp(\poly(1/\eps)\cdot(1/\eps'))=\exp(\poly(1/\eps))$ and congestion slack $\kappa\le\gamma^{d'}\le N^{\poly\eps}$.

Therefore, we obtained $2^{i}$-length-constrained $t$-step emulator $G'_{i}$ of $G$ containing $m\cdot N^{\poly\eps}$ edges with length slack $s=\exp(\poly(1/\eps))$ and congestion slack $\kappa=N^{\poly\eps}$ for all $i\le\log N$. By plugging these emulators $G'_{i}$ into \Cref{lem:reduc to h leng emu}, we obtain a $t$-step emulator $G'$ of $G$ with length slack $2st=\exp(\poly(1/\eps))$ and congestion slack $\kappa\log N=N^{\poly\eps}$ where $G'=\biguplus_{i}G'_{i}$ with capacity scaled down by $\log N$. Clearly, we have $|E(G')|\le m\cdot N^{\poly\eps}$.

Let us analyze the construction time of $G'$. We call \Cref{cor:stacking algo} $O(\log h)=O(1/\poly\eps)$ times, each of which takes $m\gamma^{d'}\poly(h)=mN^{\poly\eps}$ work and $\Otil(d\cdot\poly(h))=N^{\poly\eps}$ depth.

Finally, given an edge representation $\flow_{F'}$ of flow $F'$ in $G'$, we show how to compute the edge representation of the corresponding flow $F$ in $G$. Given $\flow_{F'}:\overleftrightarrow{E}(G')\rightarrow\mathbb{R}_{\ge0}$, we define $\flow_{F'_{i}}$ as $\flow_{F'}$ restricted to $\overleftrightarrow{E}(G'_{i})$ that represents a flow $F'_{i}$ in $G'_{i}$ where $\leng_{F'_{i}}\le\leng_{F'}$ and $\congest_{F'_{i}}\le\congest_{F'}/\log N$ (because $G'=\biguplus_{i}G'_{i}$ with capacity scaled down by $\log N$). By \Cref{cor:stacking algo}, we can compute the edge representation $\flow_{F_{i}}$ of flow $F_{i}$ in $G$ where $D_{F_{i}}=D_{F'_{i}}$, $\leng_{F_{i}}\le\leng_{F'}$ and $\congest_{F_{i}}\le\congest_{F'}/\log N$ in $m\gamma^{d'}h=mN^{\poly\eps}$ work and $\Otil(d)=\tilde{O}(1/\poly\eps)$ depth. By concatenating all $F_{i}$, we get the flow $F$ in $G$ where $D_{F}=D_{F'}$, $\leng_{F}\le\leng_{F'}$ and $\congest_{F}\le\sum_{i}\congest_{F_{i}}/\log N\le1$. The edge representation $\flow_{F}$ of $F$ can be defined as $\flow_{F}=\sum_{i}\flow_{F_{i}}$. Thus, we can return $\flow_{F}$ in $mN^{\poly\eps}$ work and $\tilde{O}(1/\poly\eps)$ depth.

\begin{comment}

\subsection{Algorithmic Emulators for General Node-Weighting}
\begin{itemize}
\item Think how to get vertex sparsifiers. 
\begin{itemize}
\item Emulator $H$ for the whole graph first. 
\begin{itemize}
\item Need that $\deg_{H}\le\deg_{G}\cdot N^{\eps}$. 
\end{itemize}
\item Round the length and compute LC-emulator for $A$ on $H$. 
\end{itemize}
\end{itemize}
\end{comment}


\section{Routing on an Expansion Witness}\label{sec:length-constrained-expander-routing}

In this section, we prove the following.

\begin{restatable}[Routing on a Router]{thm}{routingrouter}\label{thm:routing-router}
Given a $t$-step $\gamma$-router $R$ for a node weighting $A$, an $A$-respecting demand $D$ and $\eps \in (\log^{-c} N, 1)$ for some sufficiently small constant $c$, for parameters  $\lambda_{rr}(\eps)=\exp(\poly(1/\eps))$ and $\kappa_{rr}(\eps)=N^{\poly(\eps)}$, one can compute a flow $F$ routing $D$ with
\begin{itemize}
    \item congestion $\gamma \kappa_{rr}(\eps)$ and length $t \lambda_{\mathrm{rr}}(\eps)$, and
    \item support size $|\supp(F)| \le (|E(R)| + |\supp(D)|) N^{\poly \eps}$,
\end{itemize}
with $(|E(R)| + |\supp(D)|) \cdot \poly(t) N^{\poly\eps}$ work and $\poly(t) N^{\poly\eps}$ depth.
\end{restatable}

Since an expansion witness covers neighbourhoods with routers, the above result for routing on a router gives as a corollary the following result for routing on a witness.

\begin{restatable}[Routing on a Witness]{cor}{routingwitness}\label{cor:routing-witness}
Given a \emph{$(h, t_{{\cal R}}, t_{\Pi}, \kappa_{\Pi})$-witness} $({\cal N},\mathcal{R},\Pi_{{\cal R}\to G})$ of $A$ in $G$, an $A$-respecting demand $D$ such that for all $(a, b) \in \supp(D)$ there exists $S \in \mathcal{S} \in \mathcal{N}$ such that $a, b \in S$, and $\eps \in (\log^{-c} N, 1)$ for some sufficiently small constant $c$, one can compute a flow $F$ routing $D$ with
\begin{itemize}
    \item length $h_{\Pi} t_{{\cal R}} \lambda_{rr}(\eps))$ and congestion $\kappa_{\Pi}\kappa_{rr}(\eps)$, and
    \item support size $|\supp(F)| \le (|\path(\Pi_{{\cal R}\to G})| + |\supp(D)|) N^{\poly \eps}$,
\end{itemize}
with $(|\supp(D)| + \left|\path(\Pi_{{\cal R}\to G})\right|) \cdot \poly(t_{\cal R}) N^{\poly(\eps)}$ work and $\poly(t_{\cal R}) N^{\poly(\eps)}$ depth.
\end{restatable}




\subsection{Section Preliminaries}

Our algorithm is based on the \emph{cut-matching game}, especially in the low hop-length regime developed in \cite{haeupler2022cut}.

\textbf{Cut-Matching Game.} For a node weighting $A$, a cut-matching game with $r$ rounds produces a sequence of capacitated unit-length graphs $G^{(0)}, \dots, G^{(r)}$ on $\supp(A)$, where $G^{(0)}$ is the empty graph. In round $i$, the cut-player selects based on $G^{(0)}, G^{(1)}, \dots, G^{(i)}$ a pair of node weightings $(A^{(i)}, B^{(i)})$, where $A^{(i)}$ and $B^{(i)}$ are $A$-respecting, have disjoint support, and have equal size (i.e. $|A^{(i)}| = |B^{(i)}|$). The matching player then produces a capacitated unit-length matching graph $\widetilde G^{(i)}$, with edges between the supports of $A^{(i)}$ and $B^{(i)}$, and capacities such that $\deg_{\widetilde G^{(i)}} = A^{(i)} + B^{(i)}$. The produced graphs are then added to the current graph: $G^{(i+1)} := G^{(i)} + \widetilde G^{(i)}$.

\textbf{Matching Strategy.} For a node weighting $A$, a $(r, t, \eta, \Delta)$-cut strategy for the cut matching game with $r$ rounds produces, regardless of the matching player, a graph $G^{(r)}$ which is a $t$-step $\eta$-congestion router for $A$ with $\deg_{G^{(r)}} \prec \Delta A$.

We use the following cut-matching strategy of \cite{haeupler2022cut}.

\begin{theorem}[Good Cut Strategy, Theorem~4.2 of \cite{haeupler2022cut}]\label{thm:cutmatch-game-exists}
For every node weighting $A$ and every $t \le \log N$, there exists a $(r, t, \eta, \Delta)$-cut strategy with $r, \eta, \Delta \leq N^{\poly(1 / t)}$. Suppose that each graph $G^{(i)}$ produced by the matching player contains at most $m'$ edges. Then, such a cut strategy can be computed in time $\tilde{O}(\poly(t) \cdot (T(m') + m'))$, where $T(m')$ is the time needed for computing an $(h, 2^t)$-length $(\phi, N^{\poly(1/t)})$-expander decomposition on a capacitated unit-length $m'$-edge graph.
\end{theorem}

Using the expander decomposition algorithm of \Cref{thm:witness ED alg} with $\epsilon=1/t$ (and disregarding linkedness), we immediately obtain the following.

\begin{corollary}\label{lem:cutmatch-game-exists}
For every node weighting $A$ and every $\eps \in (\log^{-c} N, 1)$ for some sufficiently small constant $c$, we can compute a $(r, t, \eta, \Delta)$-cut strategy with $t=1/\epsilon$ and $r, \eta, \Delta \leq N^{\textup{poly}\epsilon}$ in depth $\tilde O(1)\cdot\textup{poly}(t)N^{\textup{poly}\epsilon}$ and work $\tilde O(m' + |\supp(A)|)\cdot\textup{poly}(t)N^{\textup{poly}\epsilon}$, where $m'$ is the maximum number of edges in any graph produced by the matching player.
\end{corollary}


 Note that what we use here has two differences to the result stated in \cite{haeupler2022cut}:
\begin{itemize}
    \item The cut matching game is on a node weighting, instead of a vertex set.
    \item The matching player can return arbitrary complete flows between the two vertex sets, instead of being restricted to return a perfect matching.
\end{itemize}
Obtaining the first generalization is simple: first, round down the node weighting $A$ into powers of two, and bucket equal powers of two, forming node weightings $A_1, \dots, A_{b}$ (for $b \leq \log N$) such that $A_{i}(v) \in \{0, 2^i\}$, exactly one of $A_{i}(v)$ is nonzero for any $v$, and $A/2 \leq \sum_i A_i \leq A$. Then, the node weighting $A_i$ of maximum value $|A_i|$ satisfies $|A_i| \geq \frac{1}{b} \sum_{i'} |A_{i'}|$. Next, we run a cut matching game on vertex set $|\supp(A_i)|$, with step bound $t - 2$. Finally, the remaining vertices need to be connected to the router constructed on $|\supp(A_i)|$. This can be done through $b - 1$ cuts, each of which has $A_j$, $j \neq i$ as one side, and a subdemand $A_i' \leq A_i$ of size $|A_i'| = |A_j|$ as the other side. Now, any demand can be routed with congestion $b$ times higher (which can be absorbed into the $N^{\poly \eps}$-factor) through paths of length $(t - 2) + 2 = t$. 

The second generalization is possible with a capacitated length-bounded expander decomposition algorithm. The algorithm or potential analysis of the cut-matching game needs no changes.

\paragraph{Length-Constrained Multicommodity Flow.}
We will use the following lemma for $h$-length $k$-commodity flow, obtained by applying a result of \cite{haeupler2023length} to the specific case of routers. Note that the \emph{additive} $k$ in support size (as opposed to multiplicative $k$ as in the algorithms of \cite{Haeupler2023lenboundflow}) is critical for our application. For derivation of \Cref{lem:lengthbound-lowcom-flow}, see \Cref{sec:derivation-lengthbound-lowcom-flow}.

\begin{restatable}{lemma}{lengthboundlowcomflow}\label{lem:lengthbound-lowcom-flow}
Let $G = (V, E)$ be a $t$-step $\gamma$-router for node weighting $A$, and $\{(A_j, A_j')\}_{j \in [k]}$ node weighting pairs such that $\sum_j A_j + A_j'$ is $A$-respecting and $|A_j| = |A_j'|$ for all $j$. Then, one can compute a flow $F = \sum_j F_j$ where $F_j$ is a complete $A_j$ to $A_j'$ -flow for each $j \in [k]$ with
\begin{itemize}
    \item length $2t$ and congestion $\tilde{O}(\gamma)$, and
    \item support size $|\supp(F)| \leq \tilde{O}(|E(G)| + k + \sum_{j} |\supp(A_j + A_j')|)$,
\end{itemize}
with $\tilde{O}((|E(G)| + \sum_j |\supp(A_j + A_j')|) \cdot k \cdot \poly(t))$ work and $\tilde{O}(k \cdot \poly(t))$ depth.
\end{restatable}

\subsection{Algorithm for Routing on a Router}

In this section, we present \Cref{alg:router-routing}, which routes a flow on a \emph{router}. The correctness of the algorithm is proven in \Cref{sec:analysis-router-routing}. The simple extension to routing on a witness is performed in \Cref{sec:routing-witness}.

In addition to the length-bounded cut-matching game strategy and flow algorithm presented in the section preliminaries, \Cref{alg:router-routing} needs two simple functions $\mathrm{SplitFlow}$ and $\mathrm{ConcatFlow}$ for manipulating flows.

$\mathrm{SplitFlow}$ splits a complete flow $F$ from a node weighting $A = \sum_{i \in k} A_i$ to $A'$ into complete flows $F_i$ from node weightings $A_i$ to $A'_i$ for $\sum_i F_i = F$ and $\sum_i A'_i = A'$:

\begin{lemma}\label{lem:split-flow-lemma}
    Let $G$ be a graph, $A = \sum_{i \in k} A_i$ and $A'$ be node weightings, and $F$ be a complete flow from $A$ to $A'$. Then, there is a deterministic algorithm $\mathrm{SplitFlow}(F, \{A_i\}_{i \in k})$ that returns (flow, node weighting) pairs $\{(F_i, A_i')\}$ such that
    \begin{itemize}
        \item $\sum_i F_i = F$ and $\sum_i A_i' = A'$,
        \item $F_i$ is a complete flow from $A_i$ to $A_i'$, and
        \item $\sum_i |\supp(A_i')| \leq \sum_i |\supp(F_i)| \leq \sum_i |\supp(A_i)| + |\supp(F)|$
    \end{itemize}
    The algorithm has work $\tilde{O}(|\supp(F)|)$ and depth $\tilde{O}(1)$.
\end{lemma}

\begin{proof}
    For every vertex $v \in \supp(A)$, let $I := \{i : A_i(v) > 0\}$ be the set of node weightings with nonzero weight on $v$ and $a = |I|$, and let $P_1, \dots, P_{b} \in \supp(F)$ be the flow paths from $v$. Fix an arbitrary order of the $a$ node weightings and $b$ paths. Then, form $a + b - 1$ pairs $(i, j)$, where a pair is formed if the prefix sums from node weighting $i - 1$ to $i$ and paths $j - 1$ to $j$ overlap. For every pair, add flow path $P_j$ to flow $F_i$ with value equal to the overlap. As sorting can be done with depth $\tilde{O}(1)$, this can be done with depth $\tilde{O}(1)$.
\end{proof}

$\mathrm{ConcatFlow}$ \textit{concatenates} a complete flow $F$ from node weighting $A$ to node weighting $A'$ with a complete flow $F'$ from node weighting $A'$ to node weighting $A''$, forming a complete flow $F^{\mathrm{res}}$ from $A$ to $A''$.

\begin{lemma}\label{lem:concat-flow-lemma}
    Let $G$ be a graph, $A, A'$ and $A''$ be node weightings with $|A| = |A'| = |A''|$, $F$ a complete $A$-to-$A'$-flow, and $F'$ a complete $A'$-to-$A''$-flow. Then, there is a deterministic algorithm $\mathrm{ConcatFlow}(F, F')$ that returns a complete $A$-to-$A''$-flow $F^{\mathrm{res}}$ such that
    \begin{itemize}
        \item Each flow path in $F^{\mathrm{res}}$ is a concatenation of a flow path in $F$ with a flow path in $F'$, and the total flow value of flow paths using a path $P \in F$ or $P' \in F'$ is $F(P)$ or $F(P')$ respectively.
        \item $|\supp(F^{\mathrm{res}})| \leq |\supp(F)| + |\supp(F')|$.
    \end{itemize}
    The algorithm has work $\tilde{O}(|\supp(F)| + |\supp(F')|)$ and depth $\tilde{O}(1)$.
\end{lemma}

\begin{proof}
    For every vertex $v \in \supp(A')$, let $P_1, \dots, P_a \in \supp(F)$ be the flow paths of $F$ to $v$, and $P'_1, \dots, P'_b$ the flow paths of $F'$ from $v$. We have $\sum_{i} F(P_i) = \sum_j F'(P_j)$ as $F$ is a complete flow to $A'$ and $F'$ is a complete flow from $A'$.

    Now, as with \Cref{lem:split-flow-lemma}, pair paths $P_i$ with paths $P_j$ if the prefix sums of flow value overlap, to form $a + b - 1$ pairs $(i, j)$. For every such pair, add a path $P$ that consists of $P_i$ concatenated with $P_j$ of value equal to the overlap.
\end{proof}

\textbf{Brief overview of \Cref{alg:router-routing}.} The algorithm is recursive, with the goal of splitting the instance into $k$ smaller instances of roughly the same total size (where the "size" of an instance equals the number of edges in the router plus the number of pairs in the demand). Suppose the graph contains more than one vertex (otherwise, we are in a trivial base case). Then, the algorithm performs the following steps:
\begin{enumerate}
    \item Split the vertex set of the router into $k$ equal size vertex subsets.
    \item Embed a router in each of the vertex subsets, by running a cut-matching game in each subset in parallel.
    \item Compute flows sending demand from each vertex subset to the correct vertex subset, so that it only remains to route flow within each vertex subset.
    \item Recursively route the resulting demand within each vertex subset using the embedded router.
\end{enumerate}
Each recursive instance will be on a vertex set that is a factor $k$ smaller than the current one, for a recursion depth of $\log_k |V(R)|$. Thus, if the total size of the recursive instances is at most $N^{\eps}$ times the size of the current instance, the total size of the instances at the bottommost level of the recursion is $N^{\eps \cdot \log_k |V(R)|}$ times the initial size -- for $k = |V(R)|^{\eps}$, this is $N^{\poly \eps}$. Notably, we must ensure that the total size of the recursive instances is not $k$ times the size of the current instance.

In step 2, the routers we recurse on are constructed through parallel cut-matching games: we run an independent cut-matching game in each of the $k$ vertex subsets. Each round $i$, we receive from each game $j \in [k]$ a partition $(A^{(i)}_j, B^{(i)}_j)$ of the node weighting of its vertex subset. To construct the matching, we call \Cref{lem:lengthbound-lowcom-flow} \textit{once}, with $k$ commodities, one for each pair $(A^{(i)}_j, B^{(i)}_j)$. Thus, the number of flow paths produced is $\tilde{O}(|E(G)| + k + \sum_{j} |\supp(A^{(i)}_j + B^{(i)}_j)|) \leq \tilde{O}(|E(G)| + k)$ instead of $\tilde{O}(k \cdot |E(G)|)$; as each flow path will correspond to an edge in the constructed routers, this avoids multiplying the total size of the recursive instances by $k$. The cut-matching games need to run for $N^{\poly \eps}$ rounds, but this is an acceptable blowup in instance size.

For step 3, we want to route the part of the demand that starts and ends at different vertex subsets (say, $A_j$ and $A_{j'}$) from $A_j$ to $A_{j'}$, so that it only remains to route flow within each subset. To do this, we create a $k^2$-commodity flow instance, where each commodity corresponds to a pair of source-vertex-subset, destination-vertex-subset, with the node weightings equal to the amount of flow each vertex of the subsets needs to send or receive from the other subset. We call \Cref{lem:lengthbound-lowcom-flow} once to solve this flow instance; again, the additive $k^2$ in support size is critical, as this goes to the total support size of the demands in the recursive instance.

Step 4 is then simple. It remains to route flow within each vertex subset, which we can do using the routers embedded in step 2. After computing these flows, we map them back to the original graph using the embeddings of the routers, and concatenate the flows of step 3 with those produced recursively.

For notational simplicity, the algorithm represents the demands with node weightings: the demand is a set of node weightings $\{S_w\}_{w \in V}$, with the correspondence $S_w(v) = D(v, w)$.

\begin{algorithm}
\caption{Router Routing Algorithm $\mathrm{RouteRouter}(G, t, \gamma, \epsilon, A, \{S_w\}_{w \in V}, k)$}\label{alg:router-routing}
\textbf{Input:} $t$-step $\gamma$-router $G = (V, E)$ for a node weighting $A$ with edge capacities $u$, demand node weightings $S_w$ ($D(v, w) = S_w(v)$) such that $\sum_{w \in V} S_w$ and $|S_w|$ are $A$-respecting, and a recursion parameter $k \ge 1$.\\
\textbf{Output:} Multicommodity flow $F = \sum_{w} F_w$ satisfying the demand (i.e. $\val((F_w)_{(v, w)}) = S_w(v)$). 
\begin{enumerate}
\setcounter{enumi}{-1}
\item Base case: If $|V| = 1$, output the empty flow as there is no demand to satisfy.
\item Partition vertices into $k$ parts: let $V_1, \ldots, V_k$ be a partition of $V$ such that $\big||V_{j}|-|V_{j'}|\big|\le1$ for all $j, j' \in [k]$, and let $A_j, S_{w, j}$ be the node weightings $A$, $S_{w}$ restricted to $V_{j}$.
\item Construct routers $G_j$ for each $A_j$ with embeddings $\Pi_j$: run $k$ cut-matching games in parallel, one on each $A_j$, using a $(r, t', \eta, \Delta)$-cut strategy of \Cref{lem:cutmatch-game-exists} with parameters $t' = 1 / \epsilon$ and $r, \eta, \Delta \leq N^{\poly \epsilon}$:
\begin{enumerate}
  \item For each round $i = 1, 2, \ldots, r$ sequentially:
    \begin{enumerate}
    \item For each $j \in [k]$, let $A^{(i)}_j, B^{(i)}_j \leq A_j$ be the $i$th node weightings produced by the cut strategy on $A_j$.
    \item Let $F^{(i)} = \sum_j F_j^{(i)}$ be a $k$-commodity flow for node weighting pairs $\{(A^{(i)}_j, B^{(i)}_j)\}_{j \in [k]}$ computed using \Cref{lem:lengthbound-lowcom-flow}. \label{line:Fqi}
    \item For each flow-path $P\in F_j^{(i)}$ with endpoints $v, w$, add an edge $e = (v, w)$ of capacity $F_j^{(i)}(P)$ to the matching graph $\tilde G_j^{(i)}$, with embedding $\Pi_j(e) = P$. \label{line:tildeGi}
    \end{enumerate}
  \item For each $j \in [k]$, let $G_j = G^{(r)}_j$ be the constructed router and $\Pi_j$ its embedding. \label{line:let-Gi}
\end{enumerate}
\item Compute flows between pairs $V_j, V_{j'}$:
\begin{enumerate}
    \item For each $j, j' \in [k]$, let $A_{j, j'}(v) := \sum_{w \in V_{j'}} S_{w, j}(v)$ and $B_{j, j'}(w) := |S_{w, j}| I[w \in V_{j'}]$.
    \item Let $F^{\mathrm{match}} = \sum_{j, j'} F_{j, j'}^{\mathrm{match}}$ be a $k^2$-commodity flow for node weighting pairs $\{(A_{j, j'}, B_{j, j'})\}_{j, j' \in [k]}$ computed using \Cref{lem:lengthbound-lowcom-flow}. \label{line:let-F}
    \item For all $j, j' \in [k]$, let $\{(F_{w, j}^{\mathrm{match}}, S'_{w, j})\}_{w \in V_{j'}} := \mathrm{SplitFlow}(F_{j, j'}^{\mathrm{match}}, \{S_{w, j}\}_{w \in V_{j'}})$.
    \item Let $S'_{w} := \sum_{j \in [k]} S'_{w, j}$.
\end{enumerate}
\item Recurse to route $\{S'_{w}\}_{w \in V_{j'}}$ inside each $G_{j'}$:
\begin{enumerate}
    \item For all $j' \in [k]$,
    \begin{enumerate}
        \item Let $F^{\mathrm{rec}}_{j'} := \sum_{w \in V_{j'}} F^{\mathrm{rec}}_{w} := \mathrm{RouteRouter}(G_{j'}, t', \eta, \epsilon, A_{j'}, \{S'_w\}_{w \in V_{j'}}, k)$.
        \item Let $F^{\mathrm{tail}}_{w} := \Pi_{j'}(F^{\mathrm{rec}}_{w})$ for $w \in V_{j'}$.
        \item Let $\{(F^{\mathrm{tail}}_{w, j}, \cdot)\}_{j \in [k]} := \mathrm{SplitFlow}(F^{\mathrm{tail}}_{w}, \{S'_{w, j}\}_{j \in [k]}))$ for $w \in V_{j'}$.
    \end{enumerate}
    \item Return $F = \sum_{w} F_{w} = \sum_{j, w} \mathrm{ConcatFlow}(F^{\mathrm{match}}_{w, j}, F^{\mathrm{tail}}_{w, j})$
\end{enumerate}

\end{enumerate}
\end{algorithm}

\subsection{Analysis of the Algorithm} \label{sec:analysis-router-routing}

Proof of correctness of \Cref{alg:router-routing} is split into multiple simple lemmas. \Cref{lem:routing-router-recsize} certifies the validity of the recursive instance and bounds its size, \Cref{lem:routing-router-timeused} bounds the work and depth of a single call, excluding recursive calls, and \Cref{lem:routing-router-quality} bounds the congestion, length and support size of the produced flow. After proving the lemmas, we combine them for the proof of \Cref{thm:routing-router}.

\begin{lemma}\label{lem:routing-router-recsize}
For each recursive call made by \Cref{alg:router-routing},
\begin{itemize}
    \item $G_{j'}$ is a $t'$-step $\eta$-router for $A_{j'}$, and
    \item $\sum_{w \in V_{j'}} S'_w$ and $|S'_w|$ (restricted to $V_{j'}$) are $A_{j'}$-respecting,
\end{itemize}
and the total size of recursive instances satisfies
\begin{equation*}
\sum_{j'} \left(|E(G_{j'})| + \sum_{w \in V_{j'}} |\supp(S'_w)|\right) \leq \tilde{O}\left(k^2 + |E(G)| + \sum_{w \in V} |\supp(S_w)|\right) N^{\poly \eps}.
\end{equation*}
\end{lemma}

\begin{proof}
    For the first claim, the graphs $G_{j'}$ are produced through a $(r, t', \eta, \Delta)$-cut strategy for the node weightings $A_{j'}$, and are thus $t'$-step $\eta$-routers for $A_{j'}$.

    For the second claim, for all $j, j'$, the sum $\sum_{w \in V_{j'}} S'_{w, j}$ is $B_{j, j'}$-respecting, as $F^{\mathrm{match}}_{j, j'}$ is a complete $A_{j, j'}$, $B_{j, j'}$-flow, and \Cref{lem:split-flow-lemma} partitions the destination node weighting. We have $B_{j, j'}(w) := |S_{w, j}|$ for $w \in V_{j'}$, thus $\sum_j B_{j, j'}(w) = |S_w|$ for $w \in V_{j'}$. Since $|S_w|$ is $A$-respecting and $A_{j'}$ is $A$ restricted to $V_{j'}$, $|S_w|$ is $A_{j'}$-respecting for $w \in V_{j'}$. Finally, by \Cref{lem:split-flow-lemma}, $|S'_w| = |S_w|$.

    It remains to bound the total size of the recursive instances. We show the following two bounds, which combine to the desired bound:
    \begin{itemize}
        \item $\sum_{w} |\supp(S'_w)| \leq \tilde{O}(k^2 + E(G) + \sum_{w} |\supp(S_w)|)$.
        \item $\sum_{j'} |E(G_{j'})| \leq \tilde{O}(|E(G)| + k) N^{\poly \eps}$.
    \end{itemize}
    For the first bound,
    \begin{itemize}
        \item By \Cref{lem:split-flow-lemma}, $\sum_w |\supp(S'_w)| \leq \sum_{w, j} |\supp(S'_{w, j})| \leq \sum_{w} |\supp(S_{w})| + |\path(F^{\mathrm{match}})|$.
        \item By \Cref{lem:lengthbound-lowcom-flow}, $|\supp(F^{\mathrm{match}})| \leq \tilde{O}(k^2 + E(G) + \sum_{j, j'} |\supp(A_{j, j'} + B_{j, j'})|) = \tilde{O}(E(G) + k^2 + \sum_{w} |\supp(S_w)|)$.
        \item Thus, $\sum_w |\supp(S'_w)| \leq \tilde{O}(k^2 + E(G) + \sum_{w} |\supp(S_w)|)$.
    \end{itemize}
    For the second bound,
    \begin{itemize}
        \item By \Cref{lem:lengthbound-lowcom-flow}, $\sum_j |\supp(F^{(i)}_j)| = |\supp(F^{(i)})| \leq \tilde{O}(E(G) + k + \sum_{j} |\supp(A_j^{(i)} + B_j^{(i)})|) = \tilde{O}(E(G) + k)$.
        \item The cut strategy (\Cref{lem:cutmatch-game-exists}) used has $r \leq N^{\poly \eps}$ rounds, thus $\sum_j |E(G_j)| \leq \tilde{O}(|E(G)| + k) N^{\poly \eps}$.
    \end{itemize}
\end{proof}

\begin{lemma}\label{lem:routing-router-timeused}
In a call to \Cref{alg:router-routing}, excluding the recursive calls, the work done is $\tilde{O}((|E(G)| + \sum_w |\supp(S_w)|) \cdot k^2 \cdot N^{\poly \eps} \cdot \poly(t))$ and the depth is $\tilde{O}(k^2 \cdot N^{\poly \eps} \cdot \poly(t))$.
\end{lemma}

\begin{proof}
    Excluding the recursive calls, work done outside calls to \Cref{lem:lengthbound-lowcom-flow} or \Cref{lem:cutmatch-game-exists} is negligble. For those two,
    \begin{itemize}
        \item the cut-strategy of \Cref{lem:cutmatch-game-exists} can be computed with work $\tilde{O}(|E(G)| + k) \cdot \poly(t) N^{\poly \eps}$ and depth $\tilde{O}(1) \cdot \poly(t) N^{\poly \eps}$, as the maximum number of edges $m'$ in a matching graph produced is $m' = \tilde{O}(|E(G)| + k)$. The total number of cut-strategies computed is $k$.
        \item computing the flows for the cut-matching games takes $\tilde{O}(|E(G)| \cdot k \cdot \poly(t))$ work and $\tilde{O}(k \cdot \poly(t))$ depth, and is done $r = N^{\poly \eps}$ times. Computing the flows for matching vertex sets $V_j, V_j'$ takes $\tilde{O}((|E(G)| + \sum_w |\supp(S_w)|) \cdot k^2 \cdot \poly(t))$ work and $\tilde{O}(k^2 \cdot \poly(t))$ depth, and is done once.
    \end{itemize}
    Thus, the total work excluding recursive calls is $\tilde{O}((|E(G)| + \sum_w |\supp(S_w)|)\cdot k^2 \cdot N^{\poly \eps} \cdot \poly(t))$, and the depth is $\tilde{O}(k^2 \cdot N^{\poly \eps} \cdot \poly(t))$.
\end{proof}

\begin{lemma}\label{lem:routing-router-quality}
Suppose that each flow $F^{\mathrm{rec}}_{j'} := \sum_{w \in V_{j'}} F_{w}^{\mathrm{rec}}$ returned by the recursive calls from a call to \Cref{alg:router-routing} satisfies the demand $\{S'_w\}_{w \in V_{j'}}$ (i.e. $\val((F^{\mathrm{rec}}_{w})_{v, w}) = S'_{w}(v)$), with
\begin{itemize}
    \item length at most $t^{\mathrm{rec}}$,
    \item congestion at most $\kappa^{\mathrm{rec}}$, and
    \item support size at most $s^{\mathrm{rec}} \cdot (|E(G_{j'})| + \sum_{w \in V_{j'}} |\supp(S'_{w})|)$.
\end{itemize}
Then, the flow $F = \sum_{w} F_w$ returned satisfies the demand $\{S_{w}\}_{w \in V}$ with
\begin{itemize}
    \item length at most $t^{\mathrm{rec}} \cdot O(t)$,
    \item congestion at most $\kappa^{\mathrm{rec}} \cdot \tilde{O}(\gamma N^{\poly \eps})$, and
    \item support size at most $\left(s^{\mathrm{rec}} \cdot \tilde{O}(N^{\poly \eps})\right) \cdot (|E(G)| + \sum_{w} |\supp(S_w)|) + \tilde{O}(k^2)$.
\end{itemize}
\end{lemma}

\begin{proof}
    First, consider the flows $F^{\mathrm{tail}}_{w} := \Pi_{j'}(F^{\mathrm{rec}}_{w})$. The length of any path in this flow is at most $2t \cdot t^{\mathrm{rec}}$, as the length of each flow path computed by \Cref{lem:lengthbound-lowcom-flow} is at most $2t$, and each edge in $G_{j'}$ maps to one such flow path. For this reason, we also have $|\supp(F^{\mathrm{tail}}_{w})| = |\supp(F^{\mathrm{rec}}_{w})|$.
    
    For congestion, consider some edge $e$ of $E(G)$. The total capacity of edges in the graphs $G_{j'}$ that map to a flow path containing $e$ is at most the number of rounds $r$ of the cut strategy times the maximum congestion of a flow produced by \Cref{lem:lengthbound-lowcom-flow}, which is $\tilde{O}(\gamma)$, and the maximum congestion of an edge $e' \in E(G_{j'})$ is the edge's capacity times $\kappa^{\mathrm{rec}}$, thus the congestion of $\sum_{w} F^{\mathrm{tail}}_{w}$ is at most $r \cdot \tilde{O}(\gamma) \cdot \kappa^{\mathrm{rec}} = \tilde{O}(\gamma \cdot \kappa^{\mathrm{rec}}) \cdot N^{\poly \eps}$.

    Next, consider the flows $F^{\mathrm{tail}}_{w, j}$. The total congestion and dilation of these flows follows from $\sum_{w, j} F^{\mathrm{tail}}_{w, j} = \sum_{w} F^{\mathrm{tail}}_{w}$. For support size, by \Cref{lem:split-flow-lemma}, we have 
    \begin{equation*}
        \sum_{w, j} |\supp(F^{\mathrm{tail}}_{w, j})| \leq \sum_{w} (\sum_{j} |\supp(S'_{w, j})| + |\supp(F^{\mathrm{tail}}_{w})|) \leq \sum_{w, j} |\supp(S'_{w, j})| + \sum_{w} |\supp(F^{\mathrm{rec}}_{w})|.
    \end{equation*}
    and, using the bound $\sum_{w, j} |\supp(S'_{w, j})| \leq \tilde{O}(k^2 + E(G) + \sum_w |\supp(S_w)|)$ from the proof of \Cref{lem:routing-router-recsize} and the assumption from the lemma,
    \begin{equation*}
        \sum_{w, j} |\supp(F^{\mathrm{tail}}_{w, j})| \leq \left(s^{\mathrm{rec}} \cdot \tilde{O}(N^{\poly \eps})\right) \cdot (|E(G)| + \sum_{w} |\supp(S_w)|) + \tilde{O}(k^2).
    \end{equation*}

    Now, consider the flows $F^{\mathrm{match}}_{j, j'}$. As they are computed by a single call to \Cref{lem:lengthbound-lowcom-flow}, the congestion of $\sum_{j, j'} F^{\mathrm{match}}_{j, j'}$ is at most $\tilde{O}(\gamma)$, the length of each flow path is at most $2t$, and
    \begin{equation*}
        \sum_{j, j'} |\supp(F^{\mathrm{match}}_{j, j'})| \leq \tilde{O}(k^2 + E(G) + \sum_{j, j'} |\supp(A_{j, j'} + B_{j, j'})|) = \tilde{O}(k^2 + E(G) + \sum_{w} |\supp(S_w)|) 
    \end{equation*}
    as argued in the proof of \Cref{lem:routing-router-recsize}. As before, after splitting, the flow $\sum_{w, j} F^{\mathrm{match}}_{w, j}$ retains the same congestion bound $\tilde{O}(\gamma)$, the length bound $2t$ and has the same support size: $\tilde{O}(k^2 + E(G) + \sum_{w} |\supp(S_w)|) + \sum_{w, j} |\supp(S_{w, j})| = \tilde{O}(k^2 + E(G) + \sum_{w} |\supp(S_w)|)$.

    The returned flow is the concatenation of the flows $F^{\mathrm{match}}_{w, j}$ with the flows $F^{\mathrm{tail}}_{w, j}$. Thus, by \Cref{lem:concat-flow-lemma}, the congestion is at most the sum of the two congestions, the length the sum of the two lengths, and the support size the sum of the two support sizes. In each case, the quantity of $F^{\mathrm{tail}}_{w, j}$ is larger, giving the desired bounds.

    It remains to show the flow satisfies the demand $\{S_w\}_{w \in V}$. This is simple: as the flow $F^{\mathrm{match}}_{w, j}$ is a complete $S_{w, j}$-to-$S'_{w, j}$-flow, and $F^{\mathrm{tail}}_{w, j}$ is a complete $S'_{w, j}$-to-$w$-flow, their concatenation is a complete $S_{w, j}$-to-$w$-flow. The sum of such flows over $j$ is a complete $S_{w}$-to-$w$-flow, as desired.
\end{proof}

We are ready to prove \Cref{thm:routing-router}.

\routingrouter*

\begin{proof}
    We produce a flow with the properties claimed by the theorem by calling \Cref{alg:router-routing} with the graph $R$ and node weighting $A$ (which is a $t$-step $\gamma$-router), $\{S_w\}_{w \in V}$ defined by $S_w(v) = D(v, w)$, $k = |V(R)|^{\eps}$ (which results in recursion depth $d := \lceil \log_k |V(R)| \rceil = \lceil 1 / \eps \rceil$, and error parameter $\epsilon' = \epsilon^{1 / c'}$ for a sufficiently large constant $c'$, such that
    \begin{itemize}
        \item $\tilde{O}(N^{\poly \eps'})^{d} \leq N^{\poly \eps}$, and
        \item $O(1 / \eps')^{d} \leq \exp(\poly(1 / \eps))$.
    \end{itemize}
    Note that this is possible as $\eps \geq \log^{-c} N$, thus $\tilde{O}(N^{\poly \eps'}) = N^{\poly \eps'}$ for a different polynomial.

    Now, having set the parameters, consider the $j$th level of recursion, for $0 \leq j \leq d$. As the recursion always splits into $k$ instances, there are $k^{j}$ recursive instances at this level. Let $\mathrm{siz}_{j}$ denote the total size (sum of $|E(G')| + \sum_w |\supp(S'_w)|$) of instances at the $j$th level of recursion (with $\mathrm{siz}_0 = |E(R)| + \sum_w |\supp(S_w)|$. We have
    \begin{equation*}
        \mathrm{siz}_{j + 1} \leq \left(\mathrm{siz}_j + k^{j} \cdot k^2\right) \tilde{O}(N^{\poly \eps'}),
    \end{equation*}
    thus $\mathrm{siz}_{j} \leq (\mathrm{siz}_0 + k^{j + 1}) \cdot \tilde{O}(N^{\poly \eps'})^{j}$.

    At every recursion level except the topmost, the graph $G'$ is a $t'$-step $\eta$-router for $t' = 1 / \eps'$ and $\eta \leq N^{\poly \eps'}$. Let $t^{\mathrm{rec}}_j$ be the maximum length of a flow returned from the $j$th recursion level and $\kappa^{\mathrm{rec}}_j$ the maximum congestion, as in \Cref{lem:routing-router-quality}. Then, we have
    \begin{itemize}
        \item $t^{\mathrm{rec}}_{j - 1} \leq t^{\mathrm{rec}}_j \cdot O(1 / \eps')$, and
        \item $\kappa^{\mathrm{rec}}_{j - 1} \leq \kappa^{\mathrm{rec}}_j \cdot \tilde{O}(N^{\poly \eps'})$.
    \end{itemize}
    For support size, the total support size of the flows returned from the second-bottommost level is at most $\mathrm{siz}_{d - 1} \cdot \tilde{O}(N^{\poly \eps'}) + k^{d - 1} \cdot \tilde{O}(k^2)$. Thus, the final returned flow has
    \begin{itemize}
        \item length at most $t \cdot O(1 / \eps')^{d} \leq t \cdot \exp(\poly(1 / \eps))$,
        \item congestion at most $\gamma \cdot \tilde{O}(N^{\poly \eps'})^{d} \leq \gamma \cdot N^{\poly \eps}$, and
        \item support size at most $\left(\mathrm{siz}_{d - 1} + k^{d + 1}\right) \cdot \tilde{O}(N^{\poly \eps'})^d \leq (|E(R)| + |\supp(D)|) N^{\poly \eps}$.
    \end{itemize}
    
    Finally, we analyze the work and depth. By \Cref{lem:routing-router-timeused}, the total work of the algorithm is
    \begin{equation*}
        \mathrm{siz}_{d} \cdot \poly(t) = (\mathrm{siz}_0 + k^{d + 1}) \cdot \tilde{O}(N^{\poly \eps'})^{d} \cdot \poly(t) \leq (|E(R)| + |\supp(D)|) \cdot N^{\poly \eps} \cdot \poly(t)
    \end{equation*}
    and its depth is $\tilde{O}(k^{3} \cdot N^{\poly \eps'} \cdot \poly(t)) \leq N^{\poly \eps} \cdot \poly(t)$ (as the size of the recursive instance does not affect its depth, and the recursion depth is $k$).
\end{proof}

\subsection{Routing on a Witness} \label{sec:routing-witness}

Recall the definition of an expansion witness:

\expansionwitness*

Our algorithm satisfying \Cref{cor:routing-witness} is simple: for every demand pair, choose an arbitrary neighbourhood cover cluster containing both vertices of the pair. Then, call routing in a router within the embedded router of each cluster to route that cluster's demand.

\routingwitness*

\begin{proof}
    First, to simplify the algorithm, for every edge $e$ in a router $R^{S}$ of a cluster $S \in \cS \in \mathcal{R}$, we split the edge into multiple edges, each of which $\Pi_{\mathcal{R} \rightarrow G}$ maps to a path, not a flow. After this, the total size of the routers is $|\path(\Pi_{\mathcal{R} \rightarrow G})|$.
    
    Now, let $D_{S}, S \in \cS \in \cN$ be demands such that $D_{S}$ is restricted to cluster $S$, $D_{S}(a, b) \in D(a, b)$, and $\sum_{S} D_{S} = D$. For each cluster $S$, let $F^{\mathrm{router}}_{S}$ be a flow routing $D_{S}$ on $R^{S}$ computed by \Cref{thm:routing-router} with error parameter $\epsilon$. Then, $F^{\mathrm{router}}_{S}$ has length $t_{\mathcal{R}} \cdot \lambda_{\mathrm{rr}}(\eps)$, congestion $1 \cdot \kappa_{\mathrm{rr}}(\eps)$ and support size $(|E(R^{S})| + |\supp(D_S)|) \cdot N^{\poly \eps}$. This takes $(|E(R^{S})| + |\supp(D_S)|) \cdot N^{\poly \eps} \cdot \poly(t_{\mathcal{R}})$ work and has depth $N^{\poly \eps} \cdot \poly(t_{\mathcal{R}})$; over all the clusters $S$, the total support size is $(|\supp(D)| + |\path(\Pi_{\mathcal{R} \rightarrow G})|) \cdot N^{\poly \eps}$ and the work is $(|\supp(D)| + |\path(\Pi_{\mathcal{R} \rightarrow G})|) \cdot N^{\poly \eps} \cdot \poly(t_{\mathcal{R}})$.

    Next, we map the flows back to the original graph using $\Pi_{\mathcal{R} \rightarrow G}$. Let $F_{S} := \Pi_{\mathcal{R} \rightarrow G}(F^{\mathrm{router}}_{S})$ be the flow on the router mapped into $G$ by the embedding of the expansion witness, and $F = \sum_{S} F_{S}$. Then, $F$ has length $h_{\Pi} \cdot t_{\mathcal{R}} \cdot \lambda_{\mathrm{rr}}(\eps)$, congestion $\kappa_{\mathrm{rr}}(\eps) \cdot \kappa_{\Pi}$ and support size $|\supp(F)| = \sum_{S \in \cS \in \mathcal{R}} |F^{\mathrm{router}}_{S}|$, as desired. The work and depth of this stage of the algorithm are negligible.
\end{proof}

\section{Low-Step Multi-commodity Flow}
\label{sec:low-step-flow}

In this section, we give $O(t)$-step $k$-commodity flow algorithms whose the running time are $|E|\cdot\poly(t)N^{\poly\eps}$. We emphasize that the running time is independent from $k$, in contrast the algorithms from \cite{Haeupler2023lenboundflow}. Later in \Cref{sec:general flow}, we will further remove the $\poly(t)$ dependency. %
\begin{comment}
Compared with length-constrained multicommodity flow algorithms of \cite{Haeupler2023lenboundflow}, our algorithms only work in undirected graphs and has worse approximation instead of

require no length approximation, are $(1+\epsilon)$-approximate in congestion and function in directed graphs, they have a multiplicative work dependency on the number of commodities $k$. This makes them infeasible to directly use in a close-to-linear time multicommodity flow algorithm.

In this section, we present length-bounded multicommodity flow approximation algorithms without this multiplicative work dependency on $k$. At a very high level, these results are achieved through combining length-bounded expander decomposition with the length-constrained expander routing algorithm from \Cref{sec:length-constrained-expander-routing}, which lacks work dependency on $k$.

While this approach removes the dependency on $k$, it restricts the algorithms to undirected graphs and requires approximation in both length and congestion with spanner-type tradeoffs. However, as flow sparsifiers are always undirected, share the length-congestion approximation tradeoff, and are required to later eliminate the $\poly(h)$-factors in work and depth shared by both the algorithms of this section and those of \cite{Haeupler2023lenboundflow}, this drawback is inconsequential. 
\end{comment}

Below, a flow $F$ is said to \emph{partially route} a demand $D$ if $D_{F}\leq D$ (pointwise), where $D_{F}$ is the demand routed by the flow $F$. The definition of the \emph{total length} and \emph{average length} of a flow are defined as follows: 
\begin{align*}
\mathrm{totlen}(F) & :=\sum_{P}F(P)\ell(P)=\sum_{e}F(e)\ell(e),\\
\avglen(F) & :=\frac{\mathrm{totlen}(F)}{\mathrm{val}(F)}.
\end{align*}
The following is the main result of this section. We can control both the maximum step $t$ and total length bound $T$ of the flow. The bound of the total length is very useful, as we will see in \Cref{sec:flow boost}, as it us to ``boost'' the congestion slack $\kappa$ to match the length slack $s\ll\kappa$. 
\begin{thm}
[Low-step Non-concurrent Flow]\label{lem:polyh-nonconcflow} Let $G$ be a graph with integral edge lengths $\ell \geq 1$ and capacities $u \geq 1$, $D$ an integral demand, $t \geq 1$ a step bound, $T \geq 1$ a total length bound, and $\eps \in (\log^{-c} N, 1)$ for some sufficiently small constant $c$ a tradeoff parameter. Then, $\textsc{LowStepNonConcFlow}(G, D, t, T, \eps)$ (\Cref{alg:avglen-lowstep-nonconcflow}) returns a multicommodity flow $F$ partially routing $D$, such that 
\begin{enumerate}
\item $F$ has maximum step length $ts$, total length $T s$ and congestion $\kappa$ for step slack $s = \exp(\poly1/\eps)$ and congestion slack $\kappa = N^{\poly(\eps)}$. The support size of $F$ is $\supportsize$. 
\item Let $F^{*}$ be the maximum-value multicommodity flow partially routing $D$ of step length $t$, total length $T$ and congestion $1$. Then, $\val(F) \geq \val(F^{*})$. 
\end{enumerate}
The algorithm has depth $\poly(t) N^{\poly(\eps)}$ and work $(|E| + \mathrm{supp}(D)) \cdot \poly(t) N^{\poly(\eps)}$.
\end{thm}

The organization of this section is as follows. In \Cref{subsec:weak cutmatch}, we first give a \textit{weak} cutmatch algorithm, which we need later. In \Cref{subsec:max len ncflow}, we give an algorithm for computing multi-commodity flows with bounded maximum length. We will use this key subroutine to prove \Cref{lem:polyh-nonconcflow} in \Cref{subsec:avg max len ncflow}. 

\subsection{Weak Cutmatch for Many Commodities}

\label{subsec:weak cutmatch}

For the flow algorithms, we will need a \textit{weak} cutmatch algorithm $\textsc{WeakCutmatch}$. It has a similar guarantee as the cutmatch algorithm from \cite{Haeupler2023lenboundflow}, but, in contrast to \cite{Haeupler2023lenboundflow}, our running time is independent of the number of commodities. This comes at the cost of slack in both length and congestion and a weaker bound on the size $|C|$ of the cut: the cut has size at most a $\phi$-fraction of the size of the total demand, instead of just the un-routed part of the demand.

\begin{lem}
\label{lem:highcomcutmatch} Let $G$ be a graph with integral edge lengths $\ell \geq 1$ and capacities $u \geq 1$, $D$ an integral demand, $h \geq 1$ a maximum length bound, $\phi$ a sparsity parameter and $\eps \in (1 / \log N, 1)$ a tradeoff parameter. Then, $\textsc{WeakCutmatch}(G,D,h,\phi,\eps)$ returns a multicommodity flow, $h$-length moving cut pair $(F,C)$ such that 
\begin{enumerate}
\item $F$ partially routes $D$. Moreover, $\val(F_{(a,b)})\in\{0,D(a,b)\}$. That is, for every vertex pair, either none of the demand or all of the demand is routed by $F$. 
\item $F$ has length $hs$ and congestion $\frac{\kappa}{\phi}$ for length slack $s=\exp(\poly(1/\eps))$ and congestion slack $\kappa=N^{\poly(\eps)}$. The support size of $F$ is $\supportsize$. 
\item For any $(a,b)$, if $\dist_{G - C}(a, b) \leq h$, then $\val(F_{(a,b)}) = D(a,b)$. That is, $F$ routes all the demand between any vertex pair not $h$-separated by the cut $C$. \label{lem:highcomcutmatch:item:close} 
\item $C$ has size at most $\phi \cdot |D|$. 
\end{enumerate}
The algorithm has depth $\poly(h) N^{\poly(\eps)}$ and work $(|E|+|\mathrm{supp}(D)|)\cdot\poly(h)N^{\poly(\eps)}$. 
\end{lem}

The algorithm for \Cref{lem:highcomcutmatch} is simple. We first compute a witnessed length-constrained $\phi$-expander decomposition $C$ for $\load(D)$ of size $|C| \le \phi|D|$. Since $G - C$ has an expansion witness, all demand pairs $(a,b)$ that appear together in some cluster $S \in \mathcal{S} \in \mathcal{N}$ that are still close in $G-C$ can be fully routed using routing in a witnessed graph (\Cref{cor:routing-witness}) with a length-constrained multicommodity flow of congestion $\approx 1 / \phi$.

\begin{algorithm}
\caption{Weak Cutmatch: $\textsc{WeakCutmatch}(G,D,h,\phi,\epsilon)$}
\label{alg:highcom-cutmatch} \textbf{Input:} Graph $G$ with integral edge lengths $\ell \geq 1$ and capacities $u \geq 1$, integral demand $D$, maximum length bound $h \geq 1$, sparsity parameter $\phi$, tradeoff parameter $\eps$.\\
 \textbf{Output:} A pair $(F, C)$ of a multicommodity flow partially routing $D$ and a $h$-length moving cut. It is guaranteed that $F$ has length $hs$ and congestion $\frac{\kappa}{\phi}$ for $s=\exp(\poly(1/\eps))$ and $\kappa=N^{\poly(\eps)}$, that $|C|\leq\phi\cdot|D|$, and that the demand between any vertices $h$-close in $G-C$ is routed by $F$. 
\begin{enumerate}
\item Let $s' := \exp(\poly(1 / \eps))$ and $\kappa' := N^{\poly(\eps)}$, with appropriate asymptotics.
\item Let $C$ and $(\mathcal{N}, \mathcal{R}, \Pi_{\mathcal{R} \rightarrow \mathcal{G}})$ be a pair of a $h$-length moving cut of size $\phi |D|$ and a $(h, \lceil 1 / \eps \rceil, hs', 2\kappa' / \phi)$ expansion witness for $G - C$, computed by \Cref{thm:witness ED alg} on $G$ and node weighting $\load(D)$ with parameters $(h, \phi / 2, 0, \lceil 1 / \eps \rceil)$.
\item Let $D'(a, b) = D(a, b) \mathbb{I}[\exists S \in \mathcal{S} \in \mathcal{N} : a, b \in S]$.
\item Let $F$ be an flow routing $D'$ of length $h s'^3$, congestion $\kappa'^2$ and support size $(|E| + |D'|) N^{\poly \epsilon}$ computed by \Cref{cor:routing-witness} on graph $G - C$, node weighting $\load(D)$ and witness $(\mathcal{N}, \mathcal{R}, \Pi_{\mathcal{R} \rightarrow G})$. 
\item Return $(F, C)$, with length slack $s = s'^3$ and congestion slack $\kappa = 2\kappa'^2$.
\end{enumerate}
\end{algorithm}

\begin{proof}[Proof of \Cref{lem:highcomcutmatch}.]
We show each of the claims. 
 
\begin{itemize}
\item \textbf{Property 1.} \Cref{cor:routing-witness} guarantees the returned flow routes exactly $D'$. As $D'(a, b) \in \{0, D(a, b)\}$, for every vertex pair, either none of the demand or all of the demand is routed.
\item \textbf{Property 2.} The flow has length $hs = hs'^3$ and congestion $\kappa / \phi = 2\kappa'^2 / \phi$ with appropriately chosen $s'$ and $\kappa'$ as \Cref{cor:routing-witness} produces a flow of length $h \cdot t_{\mathcal{R}} t_{\Pi} \lambda_{rr}(\epsilon)$ and congestion $\kappa_{\Pi} \kappa_{rr}(\epsilon)$. For support size, \Cref{thm:witness ED alg} guarantees the embedding $\Pi_{\mathcal{R} \rightarrow G}$ has path count $|E| N^{\poly(\eps)}$, thus the call to \Cref{cor:routing-witness} produces a flow of support size $(|E| N^{\poly(\eps)} + |\supp(D)|) N^{\poly(\eps)} = (|E| + |\supp(D)|) N^{\poly(\eps)}$.
\item \textbf{Property 4.} If $\dist_{G - C}(a, b) \leq h$, then $a, b \in S$ for some $S \in \mathcal{S} \in \mathcal{N}$ as $\mathcal{N}$ is a neighbourhood cover of $G - C$ with covering radius $h$.
\item \textbf{Property 3.} \Cref{thm:witness ED alg} guarantees that $|C| \leq (\phi / 2) |\load(D)| = \phi |D|$.
\item \textbf{Work and depth.} The algorithm's work consists of a call to \Cref{thm:witness ED alg} and a call to \Cref{cor:routing-witness}, both of which have depth at most $\poly(h) N^{\poly(\eps)}$ and work $(|E|+\supp(D)) \cdot \poly(h) N^{\poly(\eps)}$, thus the algorithm has this work and depth. 
\end{itemize}
\end{proof}

\subsection{Maximum Length-Constrained Non-Concurrent Flow}

\label{subsec:max len ncflow}

In this section, we give an algorithm for computing multi-commodity flows with bounded maximum length. When we use later it as a subroutine, we the input demand will be $\frac{1}{n^{2}}$-fractional. This is why the statement handles $\frac{1}{n^{2}}$-fractional demands instead of just integral demands.
\begin{lem} \label{lem:polyh-maxlen-nonconcflow}
Let $G$ be a graph with integral edge lengths $\ell \geq 1$ and capacities $u \geq 1$, $D$ a $\frac{1}{n^{2}}$-fractional demand, $h \geq 1$ a maximum length bound, and $\eps \in (\log^{-c} N, 1)$ for some sufficiently small constant $c$ a tradeoff parameter. Then, $\textsc{MaxLenNonConcFlow}(G, D, h, \eps)$ (\Cref{alg:polyh-maxlen-nonconcflow}) returns a $\frac{1}{n^2}$-fractional multicommodity flow $F$ routing a subdemand $D'$ of $D$ such that 
\begin{enumerate}
\item $F$ has maximum length $hs$ and congestion $\kappa$ for $s = \exp(\poly 1 / \eps)$ and $\kappa = N^{\poly(\eps)}$. The support size of $F$ is $\supportsize$. 
\item Let $F^{*}$ be the maximum-value multicommodity flow partially routing $D$ of maximum length $h$ and congestion $1$. Then, $\val(F) \geq \val(F^{*})$. 
\end{enumerate}
The algorithm has depth $\poly(h) N^{\poly(\eps)}$ and work $(|E| + \mathrm{supp}(D)) \cdot \poly(h) N^{\poly(\eps)}$.
\end{lem}

Our high-level strategy is as follows. We try to repeatedly apply $\textsc{WeakCutmatch}$ to send flow as much as possible. We will set its sparsity parameter to be small enough so that the cuts from $\textsc{WeakCutmatch}$ have small total size. Since $\textsc{WeakCutmatch}$ guarantees that the demands pairs that are not cut, i.e., those still close after applying the cut, must have been fully routed, this mean that we have send a lot of flow. The complete algorithm description of \Cref{lem:polyh-maxlen-nonconcflow} to carry out this strategy is shown in \Cref{alg:polyh-maxlen-nonconcflow}. 

\begin{algorithm}[h]
\caption{Length-Bounded NC Multicommodity Flow: $\textsc{MaxLenNonConcFlow}(G,D,h,\protect\eps)$}
\label{alg:polyh-maxlen-nonconcflow}

\textbf{Input:} Graph $G$ with integral edge lengths $\ell \geq 1$ and capacities $u \geq 1$, $\frac{1}{n^{2}}$-fractional demand $D$, maximum length bound $h \geq 1$, tradeoff parameter $\epsilon$.\\
 \textbf{Output:} A multicommodity flow $F$ of length $hs$ and congestion $\kappa$ routing a $\frac{1}{n^{2}}$-fractional subdemand of $D$ for length slack $s=\exp(\poly1/\eps)$ and congestion slack $\kappa=N^{\poly(\eps)}$. For every multicommodity flow $F^{*}$ partially routing $D$ of maximum length $h$ and congestion $1$, it is guaranteed that $\val(F) \geq \val(F^{*})$. 
\begin{enumerate}
\item Let $\gamma := n^{2}$ and $\kappa' := 16 \gamma \cdot \lceil\log\gamma|D|\rceil$. 
\item Let $F_{\mathrm{res}}\leftarrow0$, $D'\leftarrow\gamma D$, and $G'$ be $G$ with capacities $\gamma u$. 
\item For $i\in\{1,2,\dots,\lceil\log|D|\rceil\}$: \label{line:outer loop} 
\begin{enumerate}
\item Let $C\leftarrow0$, $F\leftarrow0$ and $S_{\mathrm{close}}\leftarrow\mathrm{supp}(D')$. 
\item For $p\in\{1,2,\dots,\lceil\log\gamma|D|\rceil\}$:\label{line:inner loop} 
\begin{enumerate}
\item Let $D'_{\mathrm{cap}}(a, b) := \min(D'(a, b), 2^p) \cdot \mathbb{I}[(a, b) \in S_{\mathrm{close}}]$. \label{line:cap-line} 
\item Let $(C',F')=\textsc{WeakCutmatch}(G'-C,D'_{\mathrm{cap}}, 2h, 1/\kappa',\eps)$. \label{line:cutmatch-line} 
\item Set $C\leftarrow C+C'$ and $F \leftarrow F + F'$. 
\item Set $D'(a,b) \leftarrow D'(a,b) - \mathrm{val}(F'_{(a,b)})$ for each $(a,b)\in\supp(D')$. 
\item Set $S_{\mathrm{close}} \leftarrow S_{\mathrm{close}} \cap \{(a, b) : \val(F'_{(a, b)}) > 0\}$ 
\end{enumerate}
\item $F_{\mathrm{res}}\leftarrow F_{\mathrm{res}}+F$. 
\end{enumerate}
\item Let $F'_{\mathrm{final}}=\textsc{WeakCutmatch}(G'-C,D',2h,1/N^2,\eps)$. 
\item Let $F_{\mathrm{final}}$ be a subflow of $F'_{\mathrm{final}}$ with integral $D_{F_{\mathrm{final}}}$ of value $\val(F_{\mathrm{final}}) = \min(\val(F'_{\mathrm{final}}), 2n^{2})$ routing a subdemand of $D'$. 
\item Let $F_{\mathrm{res}} \leftarrow F_{\mathrm{res}} + F_{\mathrm{final}}$.\label{line:final-flow} 
\item Return $\frac{1}{\gamma} F_{\mathrm{res}}$. 
\end{enumerate}
\end{algorithm}

\paragraph{Algorithm Explanation.}

We motivate the algorithm in more details here. First, we scale up the demand and capacity of the input graph by $\gamma=n^{2}$ from the beginning to allow us to solely work on integral demands. Our returned flow $F_{\mathrm{res}}$ directly satisfies Property 1 as the sum of logarithmically many flows returned from $\textsc{WeakCutmatch}$ plus a tiny, short flow. However, showing the second property, i.e. that $\val(F_{\mathrm{res}})\ge\val(F^{*})$, is non-trivial. 

The algorithm has two main loops. For the outer for-loop, we will make progress as follows. After the iteration $i$ of the outer loop, let $F_{i}^{*}$ denote the max-value flow with small maximum length and congestion partially routing remaining $D'$. We will route a flow $F$ (and add it $F_{\mathrm{res}}$) to with value $\val(F)\ge\val(F_{i}^{*})/2$ as long as $\val(F_{i}^{*})\ge2n^{2}$. This implies that, either either $\val(F^{*})-\val(F_{\mathrm{res}})$ halves, or we have $\val(F^{*})-\val(F_{\mathrm{res}})\leq2n^{2}$ for every iteration. If the latter happens, then we will add $F_{\mathrm{final}}$ to $F_{\mathrm{res}}$ at the end of value $\min(\val(F^{*})-\val(F_{\mathrm{res}}),2n^{2})$. Note that $F_{\mathrm{final}}$ can be computed trivially as we do not need to worry about the congestion as we have scaled up the capacity of the graph by $n^{2}$ from the beginning. 

The inner for-loop tries to construct $F$ where $\val(F)\ge\val(F_{i}^{*})/2$ assuming $\val(F_{i}^{*})\ge2n^{2}$. Our strategy is to maintain a subdemand $D'_{cap}$ of $D'$ such that $|D'_{cap}|\le2\val(F_{i}^{*})$. Our definition of $D'_{cap}$ satisfies this when $p=1$ since $\val(F_{i}^{*})\ge2n^{2}$. Also, $|D'_{cap}|$ can grow by only a factor of $2$ per iteration of the inner loop. 

We will use $\textsc{WeakCutmatch}$ to cut or route $D'_{cap}$. The interesting case is when $\val(F_{i}^{*})\le|D'_{cap}|\le2\val(F_{i}^{*})$. If $\textsc{WeakCutmatch}$ routes more than half of $D'_{cap}$, then we have routed at least $\val(F_{i}^{*})/2$ and we have achieved the goal. Otherwise, $\textsc{WeakCutmatch}$ cuts/separates more than half of $D'_{cap}$ which reduces $|D'_{cap}|$ in the next iteration and maintain the invariant $|D'_{cap}|\le2\val(F_{i}^{*})$. 

The bound of $|D'_{cap}|$ is useful because it means that each cut from $\textsc{WeakCutmatch}$ has size bounded by $\phi|D'_{cap}|=O(\phi\val(F_{i}^{*}))$ where $\phi$ is a parameter we can choose. So the total size size after $\tilde{O}(1)$ iterations is small compared to $\val(F_{i}^{*})$. So all cuts found did not ``separates'' too many demand pairs routed by $F_{i}^{*}$. Since $\textsc{WeakCutmatch}$ guarantees that the demand pairs  that are not separated must be are completely routed. This implies that the total flow we have routed during the inner for loop is at least $\val(F_{i}^{*})/2$ again.

For the proof of the theorem, we will use the following lemmas: 
\begin{lem}
\label{lem:sepcutlowerbound} Let $G$ be a graph with edge capacities $u$ and lengths $\ell$. Let $F$ be a $h$-length multicommodity flow and $C$ a $2h$-length moving cut. Let $F'$ be the subflow of $F$ that is $2h$-separated in $G-C$, i.e., for all paths $P$, $F'(P):=F(P)\mathbb{I}[l_{G-C}(P)>2h]$. Then, 
\[
\mathrm{val}(F')\leq2|C|\cdot\congest_{F}.
\]
\end{lem}

\begin{proof}
Let $L$ be the total length of the flow $F$ in $G$. The total length of $F$ in $G-C$ is at least $L+h\cdot\mathrm{val}(F')$, as the length of the subflow $F'$ has increased from at most $h$ to at least $2h$. On the other hand, the total length of $F$ in $G-C$ is at most $L+\sum_{e}F(e)\cdot2h\cdot C(e)\le L+\congest_{F}\cdot2h\sum_{e}u(e)C(e)=L+\congest_{F}\cdot2h|C|$. This implies that gives 
\[
L+h\cdot\mathrm{val}(F')\leq\sum_{P}F(P)\ell_{G-C}(P)\leq L+2h\cdot|C|\cdot\mathrm{cong}_{F}.
\]
Cancelling $L$ and dividing by $h$ gives $\mathrm{val}(F')\leq2|C|\cdot\mathrm{cong}_{F}$, as desired. 
\end{proof}
\begin{lem}
\label{lem:integral-good-solution-exists} Let $G'$ be a $n$-vertex graph with edge capacities $u\geq n^{2}$ and lengths $\ell$, $h$ a maximum length bound, $\gamma\geq1$ a congestion bound and $D$ an integral demand. Then, there exists a flow $F^{*}$ of maximum length $h$ and congestion $2\gamma$ that routes an \textbf{integral} subdemand of $D$, such that any flow $F^{**}$ of maximum length $h$ and congestion $\gamma$ that routes a subdemand of $D$ satisfies $\mathrm{val}(F^{*})\geq\mathrm{val}(F^{**})$. 
\end{lem}

\begin{proof}
Let $F^{**}$ be the maximum-value flow of length $h$ and congestion $\gamma$ routing a subdemand of $D$. Let $F^{*}$ be the flow created by rounding up the flow value between every vertex pair up to the next integer. This increases the total flow by at most $n^{2}$, and thus the total flow over any edge by at most $|\mathrm{supp}(D)|\leq n^{2}$. Thus, as the capacity of every edge is at least $n^{2}$, the congestion goes up by at most $1\leq\gamma$, while the flow value does not decrease. The routed demand still remains a subdemand of $D$, as $D$ is integral. 
\end{proof}

\begin{lem}
\label{lemma:separated-or-cut} At the end of each iteration of the for-loop on line \ref{line:outer loop}, for every $(a,b)$, either $D'(a,b)=0$ (the flow fully satisfies the $(a, b)$-demand) or $\mathrm{dist}_{G'-C}(a,b) > 2h$ (the pair is $2h$-separated). 
\end{lem}

\begin{proof}
Consider the iteration of the for-loop on line \ref{line:inner loop} where $p = \lceil\log\gamma|D|\rceil$. Then, $2^{p}\geq\gamma|D|$, and in particular $D'(a, b) \leq 2^{p}$ holds for every vertex pair $(a, b)$. There are four possible situations before the updates to $C, F, D'$ and $S_{\mathrm{close}}$ of the iteration:
\begin{itemize}
    \item $D'(a, b) = 0$,
    \item $(a, b) \in S_{\mathrm{close}}$ and $\val(F'_{(a, b)}) = D'_{\mathrm{cap}}(a, b)$,
    \item $(a, b) \in S_{\mathrm{close}}$ and $\val(F'_{(a, b)}) = 0$, and
    \item $(a, b) \not \in S_{\mathrm{close}}$.
\end{itemize}
In the first case we are trivially done. In the second case, we are done as $D'(a, b) = D'_{\mathrm{cap}}(a, b)$. In the third case, since the cutmatch algorithm guarantees that for $(a, b)$ with $\dist_{G' - C}(a, b) \leq 2h$ we have $\val(F'_{(a, b)}) = D'_{\mathrm{cap}}(a, b)$, the pair must be $2h$-separated. Similarly, for $(a, b)$ to have left $S_{\mathrm{close}}$ in a previous iteration, their distance must have been greater than $2h$. Thus, as distances cannot decrease, their distance remains greater than $2h$.
\end{proof}

We are ready to prove \Cref{lem:polyh-maxlen-nonconcflow}. 
\begin{proof}
The work and depth bounds are clear, as the algorithm's work and depth are dominated by $\tilde{O}(1)$ calls to \Cref{alg:highcom-cutmatch}.

\paragraph{Property 1. }

The cutmatch algorithm only produces integral flows, thus the returned flow is $\frac{1}{n^{2}}$-fractional. The flow $F_{\mathrm{res}}$ before adding $F_{\mathrm{final}}$ is the sum of $\tilde{O}(1)$ flows, each produced by \Cref{alg:highcom-cutmatch} with maximum length bound $2h$ and sparsity parameter $\frac{1}{\kappa'}=O(1 / \gamma\log N)$. \Cref{alg:highcom-cutmatch} thus guarantees that the flow has length $2h \cdot s'' = h \cdot \exp(\poly1/\eps)$ and congestion $\kappa'\kappa''=\gamma N^{\poly(\eps)}$, for $s''=\exp(\poly1/\eps)$ and $\kappa''=N^{\poly(\eps)}$. The flow $F_{\mathrm{final}}$ has length $2h \cdot s''$ and congestion at most its value of $2 \gamma$, thus adding it to $F_{\mathrm{res}}$ does not change the asymptotic length and congestion, and the returned flow has congestion $N^{\poly(\eps)}$, as desired. Finally, the support size is $\supportsize$ as $F$ is the sum of $\tilde{O}(1)$ flows of support size $\supportsize$.

\paragraph{Property 2. }

The algorithm starts by scaling up all capacities and demands by $\gamma=n^{2}$. Then, $u\geq n^{2}$, thus by \Cref{lem:integral-good-solution-exists}, the maximum-value length-$h$ congestion-$2\gamma$ flow $F^{*}$ routing an \emph{integral} subdemand of $\gamma D$ has value at least the value of any maximum-length-$h$ congestion-$\gamma$ flow routing a possibly fractional subdemand of $\gamma D$. It thus suffices to show that $\val(F_{\mathrm{res}})\ge\val(F^{*})$ at the end of the algorithm.

Fix a iteration $i$ of the outer for-loop on line \ref{line:outer loop}. Let $D'_{i}$ be the remaining demand $D'$ at the start of iteration $i$. Let $F_{i}^{*}$ be the maximum-value length-$h$ and congestion-$2\gamma$ flow routing an integral subdemand of $D'_{i}$. Consider the demand $D_{i}^{\Delta}=\min\{D'_{i},D_{F^{*}}\}$ where $D_{F^{*}}$ is routed by $F^{*}$. Since $D_{i}^{\Delta}$ is a subdemand of $D_{F^{*}}$, it is routable by a length-$h$ congestion-$2\gamma$ flow. So $\val(F_{i}^{*})\ge|D_{i}^{\Delta}|\ge\val(F^{*})-\val(F_{res})$ for $F_{res}$ at the beginning of iteration $i$. 

If $\val(F_{i}^{*})\geq2n^{2}$, we will show that the flow $F$ constructed at the end of iteration $i$ satisfies $\val(F)\ge\val(F_{i}^{*})/2$. Therefore, at the end of iteration $i$ when we set $F_{\mathrm{res}}\leftarrow F_{\mathrm{res}}+F$, either $\val(F^{*})-\val(F_{\mathrm{res}})$ halves, or we have $\val(F^{*})-\val(F_{\mathrm{res}})\leq2n^{2}$. As the difference is initially at most $\gamma|D|=n^{2}|D|$ and at the end of the algorithm we add a flow $F_{\mathrm{final}}$ of value $2n^{2}$ to the returned flow, we have that after $\lceil\log|D|\rceil$ iterations and after line \ref{line:final-flow}, we have $\val(F_{\mathrm{res}})\ge\val(F^{*})$.

Now, assume that $\val(F_{i}^{*})\geq2n^{2}$. We show that the flow $F$ constructed satisfies $\val(F)\ge\val(F_{i}^{*})/2$. To show this, by \Cref{lemma:separated-or-cut}, we have that the produced cut $C$ $2h$-separates all of the demand not routed by $F$. Thus, it is sufficient to show that at most half of $F_{i}^{*}$ is $2h$-separated by $C$.

By \Cref{lem:sepcutlowerbound}, any length-$2h$ cut $C$ $2h$-separates at most $4\gamma|C|$ of $F_{i}^{*}$, as $F_{i}^{*}$ has congestion at most $2\gamma$. Thus, as long as $|C|\leq\mathrm{val}(F_{i}^{*})/8\gamma$, $C$ separates at most half of $F_{i}^{*}$. The size of the length-$2h$ cut $C'$ produced on line \ref{line:cutmatch-line} of the algorithm is at most 
\[
\frac{1}{\kappa'}|D'_{\mathrm{cap}}|=\frac{1}{8\gamma}\frac{1}{\lceil\log\gamma|D|\rceil}\cdot\left(\frac{1}{2}|D'_{\mathrm{cap}}|\right),
\]
by the guarantee of \textsc{WeakCutmatch} from \Cref{lem:highcomcutmatch}. Thus, assuming that $|D'_{\mathrm{cap}}| \leq 2\mathrm{val}(F_{i}^{*})$ at every point during iteration $i$, we have that the size of the final length-$2h$ cut $C$ produced is at most $\frac{\mathrm{val}(F_{i}^{*})}{8\gamma}$. So, at most half of $F_{i}^{*}$ can be $2h$-separated by $C$, as desired. 

Next, we prove the following result, which gives either $|D'_{\mathrm{cap}}|\leq2\mathrm{val}(F_{i}^{*})$ at every point during iteration $i$ or directly that $\val(F)\geq\frac{1}{2}\val(F_{i}^{*})$, by induction: for every $p\in\{0,1,\dots,\lceil\log\gamma|D|\rceil\}$, either $\val(F)\geq\frac{1}{2}\val(F_{i}^{*})$, or $|D'_{\mathrm{cap},p+1}|\leq2\val(F_{i}^{*})$, where $D'_{\mathrm{cap},p}$ is the demand defined at line \ref{line:cap-line} in the iteration of the loop on line \ref{line:inner loop} for a particular $p$.

The base case follows from $|D_{\mathrm{cap},p}^{\prime}|\leq2^{p}|\mathrm{supp}(D')|\leq2n^{2}\le2\val(F_{i}^{*})$ when $p=1$. Assume the claim holds for $p'<p$. If for the previous iteration $\val(F)\geq\frac{1}{2}\val(F_{i}^{*})$, we are done, as the value of $F$ cannot decrease and the value of $F_{i}^{*}$ is invariant. Thus, we may assume that the other condition holds.

Observe that $|D'_{\mathrm{cap},p+1}|\leq2|D'_{\mathrm{cap},p}|$. Thus, if $|D'_{\mathrm{cap},p}|\leq\val(F_{i}^{*})$, we are done. Assume the contrary, and consider the pair $(C',F')$ returned by \textsc{WeakCutmatch}. By \Cref{lem:highcomcutmatch}, either the moving cut $C'$ drops pairs from $S_{\mathrm{close}}$ contributing at least a $\frac{1}{2}$-fraction of $|D'_{\mathrm{cap}, p}|$ or the flow $F'$ has value $\val(F')\ge\frac{1}{2}|D'_{\mathrm{cap},p}|$. In the latter case, we are done, as now $\val(F)\geq\val(F')\geq\frac{1}{2}|D'_{\mathrm{cap},p}|\geq\frac{1}{2}\val(F_{i}^{*})$. In the former case, we are done, as the demand pairs dropped from $S_{\mathrm{close}}$ will not contribute to $|D'_{\mathrm{cap}, p}|$ the next iteration and all iterations after, and the demand value for the other pairs is at most doubled, and thus $|D'_{\mathrm{cap},p+1}| \leq 2(\frac{1}{2}|D'_{\mathrm{cap},p}|)\leq2\mathrm{val}(F_{i}^{*})$. 
\end{proof}

\subsection{Low-Step Total-Length-Constrained Non-Concurrent Flow}

\label{subsec:avg max len ncflow}

\begin{comment}

\begin{lem} \label{lem:avglen-highcom-nonconcflow} Let $G$ be a graph with edge lengths $\ell$ and capacities $u\geq1$, $h\geq1$ a maximum length bound, $b\geq1$ an average length bound, $D$ an integral demand and $\epsilon\in(\frac{1}{\log N},1)$ a tradeoff parameter. Then, $\textsc{AvgLenNCFlow}(G,D,h,b,\eps)$ returns a multicommodity flow $F$ partially routing $D$ such that 
\begin{enumerate}
\item $F$ has length $hs$, average length $bs$ and congestion $\kappa$ for $s=\exp(\poly1/\eps)$ and $\kappa=N^{\poly(\eps)}$, and support size $\supportsize$, 
\item Let $F^{*}$ be the maximum-value multicommodity flow partially routing $D$ of length $h$, average length $b$ and congestion $1$. Then, $\mathrm{val}(F)\geq\mathrm{val}(F^{*})$. 
\end{enumerate}
The algorithm has depth $\poly(h)N^{\poly(\eps)}$ and work $(|E|+|V|+\mathrm{supp}(D))\cdot\poly(h)N^{\poly(\eps)}$. \end{lem} 
\end{comment}

In this section, we prove \Cref{lem:polyh-nonconcflow}. The algorithm needs to round edge lengths to go from a step bound and a length bound to just a length bound. The following fact shows the correctness of the approach:
\begin{fact}
\label{lem:roundinglemma} Let $G$ be a graph with positive edge lengths $\ell$, $h$ be a length bound, and $t$ be a step bound. Define the length function $\ell'(e)=\left\lceil \frac{t\cdot\ell(e)}{h}\right\rceil $. Then, for any path $p$, we have
\[
\max\left(|p|,\frac{t\cdot\ell(p)}{h}\right)\le\ell'(p)\leq|p|+\frac{t\cdot\ell(p)}{h}.
\]
In particular, 
\begin{itemize}
\item if $\ell(p)\leq h$ and $|p|\leq t$, then $\ell'(p)\leq2t$, and 
\item if $\ell'(p)\leq2t$, then $\ell(p)\leq2h$ and $|p|\leq2t$.
\end{itemize}
\end{fact}

\begin{proof}
We have $\ell'(p)=\sum_{e\in p}\left\lceil \frac{t\cdot\ell(e)}{h}\right\rceil $. Clearly, $\max\left(|p|,\frac{t\cdot\ell(p)}{h}\right)\le\sum_{e\in p}\left\lceil \frac{t\cdot\ell(e)}{h}\right\rceil \leq|p|+\frac{t\cdot\ell(p)}{h}.$
\end{proof}
\begin{comment}
For a path $p$ of length $\ell(p)\leq h$ and step-length $|p|\leq t$, we have 
\[
\ell'(p)=\sum_{e\in p}\left\lceil \frac{t\cdot\ell(e)}{h}\right\rceil \leq|p|+\frac{t\cdot\ell(p)}{h}\leq2t.
\]
For a path $p$ of length $\ell'(p)\leq2t$, we have 
\[
\ell'(p)=\sum_{e\in p}\left\lceil \frac{t\cdot\ell(e)}{h}\right\rceil \geq\max\left(|p|,\frac{t\cdot\ell(p)}{h}\right),
\]
thus $|p|\leq2t$ and $\ell(p)\leq2h$. 
\end{comment}

Now, we describe the algorithm of \Cref{lem:polyh-nonconcflow} in \Cref{alg:avglen-lowstep-nonconcflow}. 


\begin{algorithm}[H]
\caption{Low-Step Non-Concurrent Flow: $\textsc{LowStepNonConcFlow}(G,D,t,T,\protect\eps)$}
\label{alg:avglen-lowstep-nonconcflow}

\textbf{Input:} Graph $G$ with integral edge lengths $\ell \geq 1$ and capacities $u \geq 1$, integral demand $D$, step bound $t \geq 1$, total length bound $T \geq 1$, tradeoff parameter $\epsilon$.\\
 \textbf{Output:} A multicommodity flow $F$ partially routing $D$ of step length $ts$, total length $Ts$ and congestion $\kappa$ for $s = \exp(\poly 1 / \eps)$ and $\kappa = N^{\poly(\eps)}$. For every multicommodity flow $F^{*}$ partially routing $D$ of step length $t$, total length $T$ and congestion $1$, it is guaranteed that $\mathrm{val}(F) \geq \mathrm{val}(F^{*})$. 
\begin{enumerate}
\item Let $s'=\exp(\poly(1/\eps))$ be the length slack of \Cref{alg:polyh-maxlen-nonconcflow} on tradeoff parameter $\epsilon$. 
\item Let $T'=5s'T$. 
\item Let $F\leftarrow0$ and $D'\leftarrow D$. 
\item For $p\in\{0,1,\dots,\lceil\log nN\rceil\}$: 
\begin{enumerate}
\item Let $h=2^{p}$. 
\item Let $\ell'(e)=\left\lceil \frac{t\cdot\ell(e)}{h}\right\rceil $ and $G'$ be $G$ with edge lengths $\ell'$. 
\item Let $F^{\aug}=\textsc{MaxLenNonConcFlow}(G',D',2t,\eps)$.\label{line:compute length-bounded nc flow} 
\item Let $\lambda$ be the maximum value in $[0,1]$ such that $\totlen(F+\lambda F^{\aug})\leq T'$. \label{line:enforce avg len} 
\item Set $F\leftarrow F+\lambda F^{\aug}$. 
\item If $\lambda<1$, return $F$.\label{line: ncflow early return} 
\item For all $(a,b)\in\supp(D')$, set $D'(a,b)\leftarrow D'(a,b)-\mathrm{val}(F_{(a,b)}^{\aug})$. 
\end{enumerate}
\item Return $F$. \label{line: ncflow last return} 
\end{enumerate}
\end{algorithm}
\begin{proof}
The work and depth are dominated by the calls to $\textsc{MaxLenNonConcFlow}$, of which there are $\tilde{O}(1)$-many. Note that since flow returned by $\textsc{MaxLenNonConcFlow}$ routes a $\frac{1}{n^{2}}$-fractional subdemand, $D'$ is always $\frac{1}{n^{2}}$-fractional and is a valid input for $\textsc{MaxLenNonConcFlow}$. Now, we prove the two properties of the returned flow $F$.

\paragraph{Property 1.}

Let $s=\exp(\poly1/\eps)$ and $\kappa=N^{\poly(\eps)}$ be the length slack and congestion slack respectively of $\textsc{\textsc{MaxLenNonConcFlow}}$ from \Cref{lem:polyh-maxlen-nonconcflow}. By \Cref{lem:roundinglemma}, the step bound of $F$ is at most $2ts'$ because $F$ has maximum $\ell'$-length at most $2ts'$. The total length bound of $F$ is at most $T'=3s'b$ as explicitly enforced by line \ref{line:enforce avg len}. The flow $F$ has congestion $\kappa\lceil\log nN\rceil$ as there are at most $1+\lceil\log nN\rceil$ iterations. The support size bound $\supp(F)=\supportsizet$ follows directly from \Cref{lem:polyh-maxlen-nonconcflow}.

\paragraph{Property 2.}

Let $F^{*}$ be the maximum-value multicommodity flow partially routing $D$ using step $t$, total length $T$ and congestion $1$. Our goal is to show that $\val(F)\ge\val(F^{*})$.

For $p\in\{0,1,\dots,\lceil\log nN\rceil\}$, let $F_{p}^{*}$ be the sub-flow of $F^{*}$ with path lengths in $\ell$ at most $2^{p}$. Note that $F^{*}=F_{\lceil\log nN\rceil}^{*}$ because simple paths have length at most $nN$ as $\ell(e)\leq N$. Let $D_{p}^{*}$ be the demand routed by $F_{p}^{*}$. Let $F_{p}^{*\aug}$ be the flow that augment $F_{p-1}^{*}$ to $F_{p}^{*}$, i.e., $F_{p}^{*}=F_{p-1}^{*}+F_{p}^{*\aug}$. Let $F_{p}^{\mathrm{aug}}$ be the flow produced by \Cref{alg:polyh-maxlen-nonconcflow} from \Cref{lem:polyh-maxlen-nonconcflow} on line \ref{line:compute length-bounded nc flow} when $h=2^{p}$. Let $F_{p}=F_{p-1}+F_{p}^{\mathrm{aug}}$ where $F_{-1}=0$. That is, $F_{p}^{\mathrm{aug}}$ augments $F_{p-1}$ to $F_{p}$.

First, we show that $\val(F_{p})\geq\val(F_{p}^{*})$ for all $p$. Consider $D'_{p}=D'-D_{F_{p-1}}$. That is, $D'_{p}$ is the remaining demand $D'$ at the start of the iteration of the loop when $h=2^{p}$. Consider the demand $D_{p}^{\Delta}=\min(D'_{p},D_{p}^{*})$. By the definition of $F_{p}^{*}$, there is a step-$t$ $\ell$-length-$h$ congestion-$1$ flow routing $D_{p}^{\Delta}$. So, $D_{p}^{\Delta}$ can be routed by a flow with $\ell'$-length $2t$ and congestion $1$ by \Cref{lem:roundinglemma}. Therefore, by the guarantee of $\textsc{MaxLenNonConcFlow}$ from \Cref{lem:polyh-maxlen-nonconcflow}, the flow $F_{p}^{\aug}$ produced on line \ref{line:compute length-bounded nc flow} has value at least $|D_{p}^{\Delta}|\geq\val(F_{p}^{*})-\val(F_{p-1})$. Hence, we have $\val(F_{p})=\val(F_{p}^{\aug})+\val(F_{p-1})\geq\val(F_{p}^{*})$.

Suppose $F$ is returned on line \ref{line: ncflow last return}. Then, $F=F_{\lceil\log nN\rceil}$ and so we have $\val(F)=\val(F_{\lceil\log nN\rceil})\ge\val(F_{\lceil\log nN\rceil}^{*})=\mathrm{val}(F^{*})$. 

From now, we assume that $F$ is returned on line \ref{line: ncflow early return} at the iteration $p$. Observe that $F=F_{p-1}+\lambda F_{p}^{\aug}$ and $\totlen(F)=T'$. The key claim is the following, which will be proved at the end. 
\begin{claim}
\label{claim:Fhat in nonconc flow}For all $p$, there is a subflow $\Fhat_{p}$ of $F_{p}$ such that $\mathrm{val}(\Fhat_{p})=\mathrm{val}(F_{p}^{*})$ and $\totlen(\Fhat_{p})\le4s'\cdot\totlen(F_{p}^{*})$. 
\end{claim}

We will use the above claim only for $F_{p-1}$. Now, suppose for contradiction that $\val(F)<\val(F^{*})$. We will analyze $\totlen(F^{*})=\totlen(F^{*}-F_{p-1}^{*})+\totlen(F_{p-1}^{*})$. Let us analyze the two term as follows. First, we have 
\begin{align*}
\totlen(F^{*}-F_{p-1}^{*}) & =\avglen(F^{*}-F_{p-1}^{*})(\val(F^{*})-\val(F_{p-1}^{*}))\\
 & \ge\frac{1}{4s'}\avglen(F-\Fhat_{p-1})(\val(F^{*})-\val(F)+\val(F)-\val(\Fhat_{p-1}))\\
 & =\frac{1}{4s'}\left(\avglen(F-\Fhat_{p})(\val(F^{*})-\val(F))+\totlen(F-\Fhat_{p-1})\right)\\
 & >\frac{1}{4s'}\totlen(F-\Fhat_{p-1})
\end{align*}
where the first inequality follows from (1) the minimum $\ell$-length of $F^{*}-F_{p}^{*}$ is at least $2^{p-1}$, (2) the maximum $\ell$-length of $F-\Fhat_{p-1}$ is at most $2^{p+1}s'$ because the maximum $\ell'$-length of $F$ is $2ts'$ and by \Cref{lem:roundinglemma}, and (3) $\val(\Fhat_{p-1})=\val(F_{p-1}^{*})$ by \Cref{claim:Fhat in nonconc flow}. The last inequality is by our assumption that $\val(F^{*})-\val(F)>0$. Second, by \Cref{claim:Fhat in nonconc flow}, we directly have 
\[
\totlen(F_{p-1}^{*})\ge\frac{1}{4s'}\cdot\totlen(\Fhat_{p-1}).
\]
Combining the two inequalities, we get a contradiction because 
\begin{align*}
T & \ge\totlen(F^{*})\\
 & =\totlen(F^{*}-F_{p}^{*})+\totlen(F_{p}^{*})\\
 & >\totlen(F)/4s'\\
 & >T
\end{align*}
where the last equality is because $\totlen(F)=T'=5s'T$. This concludes the proof that $\val(F)\ge\val(F^{*})$ when returned on line \ref{line: ncflow early return}. 
\end{proof}
Finally, we prove \Cref{claim:Fhat in nonconc flow}.
\begin{proof}
[Proof of \Cref{claim:Fhat in nonconc flow}]We prove by induction. The base case when $p=-1$ is trivial because $\mathrm{val}(F_{-1}^{*})=0$ as every edge has length at least $1$ and $F_{-1}=0$. Now, assume the claim holds for all $p'<p$. 

Since $\val(F_{p})\geq\val(F_{p}^{*})$ and $\val(\Fhat_{p-1})=\val(F_{p-1}^{*})$ by induction, we have that $\val(F_{p}-\Fhat_{p-1})\ge\val(F_{p}^{*\aug})$.\footnote{Note that $F_{p}-\Fhat_{p-1}$ is well-defined because $\Fhat_{p-1}$ is a subflow of $F_{p-1}$ which is a subflow of $F_{p}$.} So, there exists $\gamma\in[0,1]$ such that $\gamma\val(F_{p}-\Fhat_{p-1})=\val(F_{p}^{*\aug})$. We define $\Fhat_{p}^{\aug}=\gamma(F_{p}-\Fhat_{p-1})$ and set $\Fhat_{p}=\Fhat_{p-1}+\Fhat_{p}^{\aug}$. 

Let us verify that $\Fhat_{p}$ satisfies the two properties. For the first property, we have 
\[
\val(\Fhat_{p})=\val(\Fhat_{p-1})+\val(\Fhat_{p}^{\aug})=\mathrm{val}(F_{p-1}^{*})+\mathrm{val}(F_{p}^{*\aug})=\val(F_{p}^{*})
\]
 by induction and by the definition of $\Fhat_{p}^{\aug}$. For the second property, observe that maximum $\ell'$-length of $\Fhat_{p}^{\aug}$, which is a subflow of $F_{p}$, is at most $2ts'$, and so the maximum $\ell$-length of $\Fhat_{p}^{\aug}$ is at most $(2ts')\cdot h/t=2^{p+1}s'$ by \Cref{lem:roundinglemma}. But the minimum $\ell$-length of $F_{p}^{*\aug}$ is at least $2^{p-1}$. Hence, $\avglen(\Fhat_{p}^{\aug})\le4s'\avglen(F_{p}^{*\aug})$ and so $\totlen(\Fhat_{p}^{\aug})\le4s'\totlen(F_{p}^{*\aug})$ because the value of the two flows are same. Thus, by induction, we get 
\[
\totlen(\Fhat_{p})=\totlen(\Fhat_{p-1})+\totlen(\Fhat_{p}^{\aug})\le4s'\cdot\totlen(F_{p-1}^{*})+4s'\cdot\totlen(F_{p}^{*\aug})=4s'\cdot\totlen(F_{p}^{*}).
\]
This completes the inductive step of the claim.
\end{proof}

\section{Flow Boosting}
\label{sec:flow boost}


In this section, we show how to \emph{boost} a flow algorithm that achieves length slack $s$ and congestion slack $\kappa$ to an algorithm that achieves length slack $s$ and congestion slack $(1+\epsilon)s$ with an additional running time overhead of $\poly(\kappa/\epsilon)$. Since we are primarily interested in the regime when $s=\exp(\poly(1/\epsilon))$ and $\kappa=n^{\poly\epsilon}$, boosting effectively reduces the congestion slack down to the length slack.

In order to bypass the $O(mk)$ flow-path decomposition barrier, our algorithms must output an \emph{implicit} flow, which we formalize as a flow \emph{oracle}.

\begin{definition}[Flow Oracle]
    A flow oracle $\mathcal{O}_{F}$ for a multi-commodity flow $F$ on a graph $G$ is a data structure supporting the following query:
    \begin{itemize}
        \item Given a subset $S$ of pairs of vertices of $G$, return the edge representation $\flow_{F_{S}}$ of $F_{S}$, where $F_{S} := \sum_{(a, b) \in S} F_{(a, b)}$ is the subflow of $F$ between the vertex pairs $(a, b) \in S$.
    \end{itemize}
    The oracle has query work $Q_w$ and query depth $Q_d$ if every query $S$ takes at most $Q_w$ work and has depth at most $Q_d$.
\end{definition}

\subsection{Flow Boosting Template}

We begin with a generic flow boosting \emph{template} that does not depend on the specifics of the flow problem, and instead works for any convex set $\mathcal F$ of satisfying flows. For illustration, the reader can imagine that $\mathcal F$ is the set of concurrent or non-concurrent flows for a given demand.

\begin{theorem}[Flow Boosting Template]\label{thm:boosting-template}
Let $G=(V,E,u,b)$ be a graph with capacity function $u$ and cost function $b$, and let $B\ge0$ be the cost budget. Let $\mathcal F$ be any convex set of flows in $G$ containing at least one capacity-respecting flow, and let $\epsilon,s,\kappa\ge0$ be parameters. Suppose an algorithm is given an oracle $\mathcal O$ that, given any integral edge length function $\ell:E\to\{1,2,\ldots,O(m^{1/\epsilon}N/\epsilon)\}$, computes the edge representation of a (not necessarily capacity-respecting) flow $F\in\mathcal F$ such that
 \begin{itemize}
 \item \textup{Length slack $s$:} $\sum_{e\in E}F(e)\cdot\ell(e)\le s\cdot\sum_{e\in E}F^*(e)\cdot\ell(e)$ for any flow $F^*\in\mathcal F$ that is also capacity-respecting, and
 \item \textup{Congestion slack $\kappa$:} $F(e)\le\kappa u(e)$ for all $e\in E$, i.e., $F(e)/\kappa$ is capacity-respecting.
 \end{itemize}
%
Then, there is a deterministic algorithm that makes $O(\kappa\epsilon^{-2}\log^2n)$ calls to oracle $\mathcal O$ and outputs the edge representation of a flow $\bar F\in\mathcal F$ and scalar $\lambda\ge0$ such that
 \begin{enumerate}
 \item \textup{Feasibility:} The flow $\lambda\bar F$ is capacity-respecting with cost at most $B$, and\label{item:feasibility}
 \item \textup{Approximation factor $\approx s$:} Let $\lambda^*$ be the maximum value such that there exists flow $F^*\in\mathcal F$ where $\lambda^*F^*$ is capacity-respecting with cost at most $B$. Then, $\lambda\ge\frac{1-O(\epsilon)}s\lambda^*$.\label{item:approx-factor}
 \end{enumerate}
Moreover, the flow $\bar F$ is a convex combination of the flows returned by oracle $\mathcal O$, and this convex combination can be output as well.

Furthermore, if the oracle $\mathcal O$ also outputs a flow oracle with query work $Q_w$ and query depth $Q_d$, then the algorithm can also output a flow oracle for $\bar{F}$ with query work $\tilde O(\kappa\epsilon^{-2}Q_w)$ and query depth $\tilde O(Q_d)$.

The algorithm takes $\tilde O(\kappa\epsilon^{-2}m)$ work and $\tilde O(\kappa\epsilon^{-2})$ time outside of the oracle calls.
\end{theorem}

For the rest of this subsection, we prove \Cref{thm:boosting-template}. The proof closely follows Sections~5~and~6 of \cite{garg2007faster}, so we claim no novelty here. We first impose the assumption that $\lambda^*\ge1$ for $\lambda^*$ as defined in Condition~\ref{item:approx-factor} of \Cref{thm:boosting-template}.

Let $\mathcal K$ be the set of capacity-respecting flows in $G$, and consider the following flow LP of the graph $G$. We have a variable $x(F)\ge0$ for each $F\in\mathcal F\cap\mathcal K$ indicating that we send flow $F$ scaled by $x(F)$. To avoid clutter, we also define $b(F)=\sum_{e\in E}F(e)\cdot b(e)$ as the cost of the flow $F$.
\begin{align*}
\max \qquad & \sum_{F\in\mathcal F\cap\mathcal K} x(F)
\\\text{s.t.} \qquad & \sum_{F\in\mathcal F\cap\mathcal K} F(e) \cdot x(F) \le u(e) & \forall e\in E
\\ & \sum_{F\in\mathcal F\cap\mathcal K}b(F)\cdot x(F)\le B
\\ & x\ge0
\end{align*}
Let $\beta$ be the optimal value of this LP. Note that since $\mathcal F\cap\mathcal K$ is convex, there is an optimal solution with $x(F^*)=\beta$ for some $F^*\in\mathcal F$ and $x(F)=0$ elsewhere. It follows that $\beta=\lambda^*$.

The dual LP has a length $\ell(e)\ge0$ for each edge $e\in E$ as well as a length $\phi\ge0$ of the cost constraint.
\begin{align*}
\min \qquad & \sum_{e\in E} u(e)\cdot\ell(e) + B\cdot\phi & =:D(\ell,\phi)
\\\text{s.t.} \qquad & \sum_{e\in E} F(e) \cdot(\ell(e)+b(e)\phi)\ge 1 & \forall F\in\mathcal F\cap\mathcal K
\\ & \ell\ge0,\,\phi\ge0
\end{align*}
Let $D(\ell,\phi)=\sum_{e\in E} u(e)\ell(e)+B\cdot\phi$ be the objective value of the dual LP. Define $\alpha(\ell,\phi)$ as the minimum length of a flow $F\in\mathcal F\cap\mathcal K$ under length function $\ell+\phi b$:
\[ \alpha(\ell,\phi)=\min_{F\in\mathcal F\cap\mathcal K}\sum_{e\in E}F(e)\cdot(\ell(e)+b(e)\phi) .\]
Then by scaling, we can restate the dual LP as finding a length function $\ell$ minimizing $D(\ell,\phi)/\alpha(\ell,\phi)$. By LP duality, the minimum is $\beta$, the optimal value of the primal LP, which we recall also equals $\lambda^*\ge1$.

The algorithm initializes length functions $\ell^{(0)}(e)=\delta/u(e)$ and $\phi^{(0)}=\delta/B$ for parameter $\delta=m^{-1/\epsilon}$, and proceeds for a number of iterations. For each iteration $i$, the algorithm wishes to call the oracle $\mathcal O$ on length function $\ell^{(i-1)}+\phi^{(i-1)}b$, but the length function $\ell^{(i-1)}+\phi^{(i-1)}b$ is not integral. However, we will ensure that they are always in the range $[\delta/N,O(1)]$. So the algorithm first multiplies each length by $O(N/(\delta\epsilon))=O(m^{1/\epsilon}N/\epsilon)$ and then rounds the weights to integers so that each length is scaled by roughly the same factor up to $(1+\epsilon)$. The algorithm calls the oracle on these scaled, integral weights to obtain a flow $F^{(i)}$. On the original, unscaled graph, the flow satisfies the following two properties:
 \begin{enumerate}
 \item \textup{Length slack $(1+\epsilon)s$:} $\sum_{e\in E}F^{(i)}(e)\cdot(\ell^{(i-1)}(e)+b(e)\phi^{(i-1)})\le(1+\epsilon)s\cdot\alpha(\ell^{(i-1)},\phi^{(i-1)})$, and
 \item \textup{Congestion slack $\kappa$:} $F(e)\le\kappa u(e)$ for all $e\in E$, i.e., $F/\kappa\in\mathcal K$.
 \end{enumerate}
Define $z^{(i)}=\min\{1,\,B/b(F^{(i)})\}$ so that $b(z^{(i)}F^{(i)})\le B$, i.e., the cost of the scaled flow $z^{(i)}F^{(i)}$ is within the budget $B$.
The lengths are then modified as
\[ \ell^{(i)}(e)=\ell^{(i-1)}(e)\bigg(1+\frac\epsilon\kappa\cdot\frac{z^{(i)}F^{(i)}(e)}{u(e)}\bigg) \quad\text{and}\quad \phi^{(i)}=\phi^{(i-1)}\bigg(1+\frac\epsilon\kappa\cdot\frac{b(z^{(i)}F^{(i)})}B\bigg) . \]
This concludes the description of a single iteration. The algorithm terminates upon reaching the first iteration $t$ for which $D(t)\ge1$ and outputs
\[ \bar F=\frac{z^{(1)}F^{(1)}+z^{(2)}F^{(2)}+\cdots+z^{(t-1)}F^{(t-1)}}{z^{(1)}+z^{(2)}+\cdots+z^{(t-1)}}\quad\text{and}\quad\lambda=\frac{z^{(1)}+z^{(2)}+\cdots+z^{(t-1)}}{\kappa\log_{1+\epsilon}1/\delta} .\]
\paragraph{Analysis.}
We will analyze the values of $D(\ell^{(i)},\phi^{(i)})$ and $\alpha(\ell^{(i)},\phi^{(i)})$ only for the lengths $\ell^{(i)},\phi^{(i)}$. To avoid clutter, we denote $D(i)=D(\ell^{(i)},\phi^{(i)})$ and $\alpha(i)=\alpha(\ell^{(i)},\phi^{(i)})$.  For each iteration $i$, we have
\begin{align*}
D(i)&=\sum_{e\in E} u(e)\cdot\ell^{(i)}(e)+B\cdot\phi^{(i)}
\\&=\sum_{e\in E}u(e)\cdot\ell^{(i-1)}(e)\bigg(1+\frac\epsilon\kappa\cdot\frac{z^{(i)}F^{(i)}(e)}{u(e)}\bigg)+B\cdot\phi^{(i-1)}\bigg(1+\frac\epsilon\kappa\cdot\frac{b(z^{(i)}F^{(i)})}B\bigg)
\\&=D(i-1)+\frac\epsilon\kappa\cdot\sum_{e\in E}z^{(i)}F^{(i)}(e)\cdot\ell^{(i-1)}(e)+\frac\epsilon\kappa\cdot z^{(i)}\cdot b(F^{(i)})\cdot\phi^{(i-1)}
\\&=D(i-1)+\frac\epsilon\kappa\cdot z^{(i)}\bigg(\sum_{e\in E}F^{(i)}(e)\cdot(\ell^{(i-1)}(e)+b(e)\phi^{(i-1)})\bigg)
\\&\le D(i-1)+\frac\epsilon\kappa\cdot z^{(i)}\cdot(1+\epsilon)s\cdot \alpha(i-1).
\end{align*}
Since $D(i-1)/\alpha(i-1)\ge\beta$ by definition of $\beta$, we have
\[ D(i) \le D(i-1)+\frac\epsilon\kappa\cdot z^{(i)}\cdot(1+\epsilon)s\cdot \frac{D(i-1)}\beta .\]
Define $\epsilon'=\epsilon(1+\epsilon)s/\kappa$ so that
\[ D(i)\le\bigg(1+\frac{\epsilon(1+\epsilon)sz^{(i)}}{\kappa\beta}\bigg)D(i-1) = \bigg(1+\frac{\epsilon'z^{(i)}}\beta\bigg)D(i-1) .\]
Since $D(0)=m\delta$ we have for $i\ge1$
\begin{align*}
D(i)\le\bigg(\prod_{j\le i}(1+\epsilon'z^{(j)}/\beta)\bigg)m\delta&=\bigg(1+\frac{\epsilon'z^{(i)}}\beta\bigg)m\delta\prod_{j\le i-1}\bigg(1+\frac{\epsilon'z^{(j)}}\beta\bigg)
\\&\le(1+\epsilon')m\delta\exp\bigg(\frac{\epsilon'\sum_{j\le i-1}z^{(j)}}\beta\bigg) ,
\end{align*}
where the last inequality uses our assumption that $\beta\ge1$ and the fact that $z^{(j)}\le1$ by definition. To avoid clutter, define $z^{(\le i)}=\sum_{j\le i}z^{(j)}$ for all $i$.

The procedure stops at the first iteration $t$ for which $D(t)\ge1$. Therefore,
\[ 1\le D(t)\le(1+\epsilon')m\delta\exp\bigg(\frac{\epsilon'z^{(\le t-1)}}\beta\bigg) ,\]
which implies
\begin{gather}
\frac\beta{z^{(\le t-1)}}\le\frac{\epsilon'}{\ln\frac1{(1+\epsilon')m\delta}} .\label{eq:to-establish-condition-2}
\end{gather}
\begin{claim}
The scaled down flow $\frac1{\kappa\log_{1+\epsilon}1/\delta}(z^{(1)}F^{(1)}+z^{(2)}F^{(2)}+\cdots+z^{(t-1)}F^{(t-1)})$ is capacity-respecting with cost at most $B$.
\end{claim}
\begin{proof}
To show it is capacity-respecting, consider an edge $e$. On each iteration, we route $z^{(i)}F^{(i)}(e)\le\kappa u(e)$ units of flow through $e$ and increase its length by a factor $(1+\frac\epsilon\kappa\cdot\frac{z^{(i)}F^{(i)}(e)}{u(e)})\le1+\epsilon$. So for every $\kappa u(e)$ units of flow routed through $e$ over the iterations, we increase its length by at least a factor $1+\epsilon$. Initially, its length is $\delta/u(e)$, and after $t-1$ iterations, since $D(t-1)<1$, the length of $e$ satisfies $\ell^{(t-1)}(e)\le D(t-1)/u(e)<1/u(e)$. Therefore the total amount of flow through $e$ in the first $t-1$ phases is strictly less than $\kappa\log_{1+\epsilon}\frac{1/u(e)}{\delta/u(e)}=\kappa\log_{1+\epsilon}1/\delta$ times its capacity. Scaling the flow down by $\kappa\log_{1+\epsilon}1/\delta$, we obtain a capacity-respecting flow.

To show that the scaled down flow has cost at most $B$, similarly observe that on each iteration, we route a flow of cost $b(z^{(i)}F^{(i)})\le B$ and increase the length $\phi^{(i)}$ by a factor $(1+\frac\epsilon\kappa\cdot\frac{b(z^{(i)}F^{(i)})}B)\le1+\epsilon/\kappa$ over the previous length $\phi^{(i-1)}$. So for every $B$ cost of flow routed, we increase the length $\phi^{(i)}$ by at least a factor $1+\epsilon/\kappa$. And for every $\kappa B$ cost of flow routed, the length increases by at least a factor $(1+\epsilon/\kappa)^\kappa\ge1+\epsilon$. Initially, $\phi^{(0)}=\delta/B$, and after $t-1$ iterations, since $D(t-1)<1$, the length satisfies $\phi^{(t-1)}\le D(t-1)/B<1/B$. Therefore the total cost of flow in the first $t-1$ phases is strictly less than $\kappa B\log_{1+\epsilon}\frac{1/B}{\delta/B}=\kappa B\log_{1+\epsilon}1/\delta$. Scaling the flow down by $\kappa\log_{1+\epsilon}1/\delta$, we obtain a flow with cost at most $B$.
\end{proof}
Recall that $\bar F=\frac{z^{(1)}F^{(1)}+z^{(2)}F^{(2)}+\cdots+z^{(t-1)}F^{(t-1)}}{z^{(\le t-1)}}$ and $\lambda=\frac{z^{(\le t-1)}}{\kappa\log_{1+\epsilon}1/\delta}$, so $\lambda\bar F$ is capacity-respecting with cost at most $B$, fulfilling Condition~\ref{item:feasibility}. To establish Condition~\ref{item:approx-factor}, we use Equation~(\ref{eq:to-establish-condition-2}) and the fact that $\delta=m^{1/\epsilon}$ to obtain
\[ \frac\lambda{\lambda^*}=\frac{z^{(\le t-1)}}{\kappa\log_{1+\epsilon}1/\delta}\cdot\frac1\beta\ge \frac{\ln\frac1{(1+\epsilon')m\delta}}{\epsilon'\cdot\kappa\log_{1+\epsilon}1/\delta} = \frac{\ln\frac1{(1+\epsilon')m\delta}}{\epsilon s\log_{1+\epsilon}1/\delta} \ge \frac{1-O(\epsilon)}s .\]

\paragraph{Running time.}
Recall from above that $\lambda\bar F$ is capacity-respecting with cost at most $B$, so $\lambda\le\beta$. Since $\lambda=\frac{z^{(\le t-1)}}{\kappa\log_{1+\epsilon}1/\delta}$, we obtain $z^{(\le t-1)}\le\beta\kappa\log_{1+\epsilon}1/\delta$. On each iteration $i\le t-1$, either $z^{(i)}=1$ or $z^{(i)}<1$, and the latter case implies that $b(z^{(i)}F^{(i)})=B$, which means $\phi^{(i)}=\phi^{(i-1)}(1+\epsilon/\kappa)$. Initially, $\phi^{(0)}=\delta/B$, and after $t-1$ iterations, since $D(t-1)<1$, we have $\phi^{(t-1)}\le D(t-1)/B<1/B$. So the event $\phi^{(i)}=\phi^{(i-1)}(1+\epsilon/\kappa)$ can happen at most $\log_{1+\epsilon/\kappa}1/\delta$ times. It follows that $z^{(i)}<1$ for at most $\log_{1+\epsilon/\kappa}1/\delta$ values of $i\le t-1$. Since $z^{(\le t-1)}\le\beta\kappa\log_{1+\epsilon}1/\delta$, we have $z^{(i)}=1$ for at most $\beta\kappa\log_{1+\epsilon}1/\delta$ many values of $i\le t-1$. Therefore, the number of iterations $t$ is at most
\[ \log_{1+\epsilon/\kappa}1/\delta+\beta\kappa\log_{1+\epsilon}1/\delta+1=O(\beta\kappa\epsilon^{-2}\log m) .\]
By scaling all edge capacities and costs by various powers of two, we can ensure that $\beta\in[1,2]$ on at least one guess, so the number of iterations is $O(\kappa\log_{1+\epsilon}1/\delta)=O(\kappa\epsilon^{-2}\log m)$. Doing so also ensures that $\lambda^*=\beta\ge1$ as we had previously assumed. For incorrect guesses, we terminate the algorithm above after $O(\kappa\epsilon^{-2}\log m)$ iterations to not waste further computation. Among all guesses, we take the one with maximum $\lambda$ that satisfies feasibility (Condition~\ref{item:feasibility}). Since there are $O(\log n)$ relevant powers of two, the running time picks up an overhead of $O(\log n)$.

\paragraph{Flow oracle.}
Finally, if oracle $\mathcal O$ outputs a flow oracle, then the algorithm can return the following flow oracle for the output flow $F$: on input subset $S$ of pairs of vertices of $G$, query the flow oracles of the flows $F^{(1)},F^{(2)},\ldots,F^{(t-1)}$ to obtain flows $F^{(1)}_S,F^{(2)}_S,\ldots,F^{(t-1)}_S$, respectively, and output
\[ F_S=\frac{z^{(1)}F^{(1)}_S+z^{(2)}F^{(2)}_S+\cdots+z^{(t-1)}F^{(t-1)}_S}{z^{(1)}+z^{(2)}+\cdots+z^{(t-1)}} .\]
Querying each flow oracle $F^{(i)}_S$ takes work $Q_w$ and depth $Q_d$ and can be done in parallel. Since there are $t-1\le O(\kappa\eps^{-2}\log m)$ many flow oracles, the total work is $\tilde O(\kappa\eps^{-2}Q_w)$ and the total depth is $\tilde O(Q_d)$.

\subsection{Flow Boosting: Thresholded Instantiation}

By choosing an appropriate convex set $\mathcal F$ of flows, we directly obtain boosting theorems for computing a flow of a given target value $\tau$ that partially routes a given flow demand $D$. Note that this formulation includes both min-cost concurrent and non-concurrent multicommodity flow: for concurrent we set $\tau=|D|$ forcing the flow to route the entire demand $D$, and for non-concurrent we set $\tau=1$ which is less than or equal to the value of any nonzero demand pair (when the demand is integral).
\begin{theorem}[Thresholded Flow Boosting]\label{thm:boosting}

Let $G=(V,E,u,b)$ be a graph with capacity function $u$ and cost function $b$. Let $B\ge0$ be the cost budget. 
Let $D:V\times V\to\mathbb R_{\ge0}$ be a flow demand, and let $\epsilon,s,\kappa\ge0$ be parameters. Let $\tau$ be the target flow value parameter such that there exists a capacity-respecting flow partially routing $D$ of value at least $\tau$.

Suppose an algorithm is given an oracle $\mathcal O$ that, given any integral edge length function $\ell:E\to\{1,2,\ldots,O(m^{1/\epsilon}N/\epsilon)\}$, computes the edge representation of a flow $F$ of value $\val(F) \geq \tau$ partially routing $D$ such that
 \begin{itemize}
 \item \textup{Length slack $s$:} $\sum_{e\in E}F(e)\cdot\ell(e)\le s\cdot\sum_{e\in E}F^*(e)\cdot\ell(e)$ for any capacity-respecting flow $F^*$ of value $\val(F^{*}) \geq \tau$ partially routing $D$, and
 \item \textup{Congestion slack $\kappa$:} $F(e) \le \kappa u(e)$ for all $e\in E$, i.e., $F(e) / \kappa$ is capacity-respecting.
 \end{itemize}
Then, there is a deterministic algorithm that makes $O(\kappa\epsilon^{-2}\log^2n)$ calls to oracle $\mathcal{O}$ and outputs the edge representation of a flow $\bar F$ of value $\val(\bar{F}) \geq \tau$ partially routing $D$ and scalar $\lambda \ge 0$ such that
 \begin{enumerate}
 \item \textup{Feasibility:} The flow $\lambda\bar F$ is capacity-respecting with cost at most $B$, and
 \item \textup{Approximation factor $\approx s$:} Let $\lambda^*$ be the maximum value such that there exists a flow $F^*$ of value $\val(F^{*}) \geq \tau$ partially routing $D$ where $\lambda^*F^*$ is capacity-respecting with cost at most $B$. Then, $\lambda\ge\frac{1-O(\epsilon)}s\lambda^*$.
 \end{enumerate}
Moreover, the flow $\bar F$ is a convex combination of the flows returned by oracle $\mathcal O$, and this convex combination can be output as well.

Furthermore, if the oracle $\mathcal O$ also outputs a flow oracle with query work $Q_w$ and query depth $Q_d$, then the algorithm can also output a flow oracle for $\bar{F}$ with query work $\tilde O(\kappa\epsilon^{-2}Q_w)$ work and $\tilde O(Q_d)$ depth.

The algorithm takes $\tilde O(\kappa\epsilon^{-2}m)$ work and $\tilde O(\kappa\epsilon^{-2})$ time outside of the oracle calls.
\end{theorem}
\begin{proof}
Let $\mathcal{F}$ be the set of flows of value at least $\tau$ partially routing $D$. Apply \Cref{thm:boosting-template}.
\end{proof}

\section{Constant-Approximate Multi-commodity Flow}
\label{sec:general flow}

This section concludes with our constant-approximate algorithms for $k$-commodity flow in $(m+k)^{1+\eps}$ time. The main results are stated in \Cref{thm:constapprox-conc-nonconc-flow}. 
Below, we first implement the oracle for flow boosting in \Cref{sec:oracle-for-boosting} and then apply the oracle to the flow boosting framework in \Cref{sec:apply-boosting} to obtain our results.

\subsection{Oracle for Flow Boosting}
\label{sec:oracle-for-boosting}
In this subsection we construct the oracle required for \Cref{thm:boosting}.



\begin{thm}[Oracle for Flow-Boosting]\label{thm:fast-unboosted-flow}
Let $G$ be a graph with integral edge lengths $\ell \geq 1$ and capacities $u \geq 1$, $D$ an integral $\deg_{G}$-respecting demand, $\tau$ an integral required flow amount such that there exists a value-$\tau$ capacity-respecting flow partially routing $D$ and $\epsilon \in (\log^{-c} N, 1)$ for a small enough constant $c$ a tradeoff parameter. Then, $\textsc{UnboostedFlow}(G, D, \tau, \eps)$ (\Cref{alg:fast-unboosted-flow}) returns a flow oracle $\mathcal{O}_F$ for a flow $F$ of value $\mathrm{val}(F) \geq \tau$ partially routing $D$, with
\begin{enumerate}
\item Length slack $s$: $\sum_{e \in E} F(e) \cdot \ell(e) \leq s \cdot \sum_{e \in E} F^{*}(e) \cdot \ell(e)$ for any capacity-respecting flow $F^{*}$ of value $\mathrm{val}(F^{*}) \geq \tau$ partially routing $D$, and
\item Congestion Slack $\kappa$: $F(e) \leq \kappa u(e)$ for all $e \in E$, i.e., $F(e) / \kappa$ is capacity-respecting,
\end{enumerate}
for length slack $s = \exp(\poly(1 / \eps))$ and congestion slack $\kappa = N^{\poly(\eps)}$.

The algorithm has depth and the flow oracle has query depth $N^{\poly(\eps)}$, and the algorithm has work and the flow oracle query work $(|E| + \mathrm{supp}(D)) \cdot N^{\poly(\eps)}$. The produced flow has support size $(|E| + \mathrm{supp}(D)) N^{\poly(\eps)}$.
\end{thm}

\begin{algorithm}[h]
\caption{Fast Unboosted Flow: $\textsc{UnboostedFlow}(G, D, \tau, \eps)$ \label{alg:fast-unboosted-flow}}

\textbf{Input:} Graph $G$ with integral edge lengths $\ell \geq 1$ and capacities $u \geq 1$, integral $\deg_{G}$-respecting demand $D$, integral required flow amount $\tau$ such that there exists a value-$\tau$ capacity-respecting flow partially routing $D$, tradeoff parameter $\epsilon$.\\
\textbf{Output:} A flow oracle $\mathcal{O}_F$ for a flow $F$ of value $\mathrm{val}(F) \geq \tau$ partially routing $D$ with length slack $s = \exp(\poly(1 / \eps))$, congestion slack $\kappa = N^{\poly(\eps)}$ and support size $(|E| + \mathrm{supp}(D)) N^{\poly(\eps)}$.
\begin{enumerate}
\item Let $t, G' = \textsc{LowStepEmu}(G, \eps)$.
\item Let $\kappa'$ be the congestion slack of $G'$, and let $G'' = \kappa' G'$ be $G'$ with capacities scaled up by $\kappa'$. 
\item Let $T_{\mathrm{low}} = 0$ and $T_{\mathrm{high}} = N$.
\item Let $F_{\mathrm{res}} = 0$.
\item While $T_{\mathrm{low}} \leq T_{\mathrm{high}}$:
\begin{enumerate}
    \item Let $T = \left\lfloor \frac{T_{\mathrm{low}} + T_{\mathrm{high}}}{2} \right\rfloor$.
    \item Let $F' = \textsc{LowStepNonConcFlow}(G'', D, t, T, \eps)$.
    \item If $\mathrm{val}(F') \geq \tau$, \label{line:value-check}
    \begin{itemize}
        \item Set $T_{\mathrm{high}} = T - 1$.
        \item Set $F_{\mathrm{res}} = F'$.
    \end{itemize}
    \item Else, set $T_{\mathrm{low}} = T + 1$.
\end{enumerate}
\item Return $\mathcal{O}_F = (G', F_{\mathrm{res}})$ with query function $\mathcal{O}_{F}(S) = \textsc{FlowMap}(G', F'_{S})$.
\end{enumerate}
\end{algorithm}

\begin{proof}[Proof of \Cref{thm:fast-unboosted-flow}]
Let $F^{*}$ be the capacity-respecting flow of value at least $\tau$ routing a subdemand of $D$ on $G$ of minimum total length $T^{*}$. \Cref{thm:emu algo} guarantees that $G'$ is a $t$-step emulator for $\deg_{G}$ with length slack $s' = \exp(\poly(1 / \eps))$ and congestion slack $\kappa' = N^{\poly(\eps)}$. Since $D$ is $\deg_{G}$-respecting, there is a flow $F^{*\prime}$ on $G'$ routing the same demand as $F^{*}$ of step-length $t$, congestion at most $\kappa'$ and total length $T^{*} s'$. As $G''$ is simply $G'$ with capacities scaled up by $\kappa'$, the flow $F^{*\prime}$ on $G''$ is a capacity-respecting flow partially routing $D$ of step length $t$, total length $T^{*} s'$, and value at least $\tau$.

\Cref{alg:avglen-lowstep-nonconcflow} guarantees that for input $(G'', D, t, T, \eps)$, for every capacity-respecting flow $F^{*}$ partially routing $D$ of step length $t$ and total length $T$, the returned flow $F$ has value $\val(F) \geq \val(F^{*})$. Thus, notably when $T \geq T^{*} s'$, the returned flow has value at least $\tau$. Let $T'$ be the minimum value of $T$ for which the returned flow $F_{\mathrm{res}}$ had value at least $\tau$ on line \ref{line:value-check}. By the above argument, $T' \leq \lceil T^{*} s' \rceil$.

\Cref{alg:avglen-lowstep-nonconcflow} guarantees that $F_{\mathrm{res}}$ on $G''$ is a congestion-$\kappa''$ flow partially routing $D$ of step length $ts''$, total length $T' s''$ and value at least $\tau$ for the congestion slack $\kappa'' = N^{\poly(\eps)}$ and length slack $s'' = \exp(\poly(1 / \eps))$ of the flow algorithm. The flow has the same step length, total length and value, but congestion $\kappa' \kappa'' = \kappa = N^{\poly(\eps)}$ on $G'$. Since $G'$ can be embedded into $G$ with congestion $1$ and length slack $1$, the flow $F$ that is $F_{\mathrm{res}}$ mapped from $G'$ to $G$ by the embedding is a congestion-$\kappa = \kappa' \kappa''$ flow partially routing $D$ of value at least $\tau$ of total length $T^{*} s = \lceil T^{*} s' \rceil s'' = T' s''$ for length slack $s = \exp(\poly(1 / \eps))$ and congestion slack $\kappa = N^{\poly(\eps)}$\footnote{Note that since $\tau$ is a integer and $\ell \geq 1$, $T^{*} \geq 1$ and taking ceil does not affect the asymptotic congestion slack}, as desired. As $|E(G')| = |E| N^{\poly(\eps)}$ and the embedding from $G'$ to $G$ maps edges to paths, thus not increasing the flow support size, the flow support size is $(|E| + \supp(D)) \cdot N^{\poly(\eps)}$ as desired.

The algorithm consists of one call to \textsc{LowStepEmu} and $\tilde{O}(1)$ calls to $\textsc{LowStepNonConcFlow}$ on a $|E| N^{\poly(\eps)}$-edge graph, thus its work and depth are $(|E| + \supp(D)) \cdot N^{\poly(\eps)}$ and $N^{\poly(\eps)}$ respectively.
\end{proof}

\subsection{Constant-Approximate Min Cost Multi-Commodity Flow}
\label{sec:apply-boosting}

The following theorem is a generalisation of concurrent and non-concurrent flow, which gives concurrent flow when $\tau = |D|$ and non-concurrent flow when $\tau = 1$.


\begin{theorem}[Constant-Approximate Multi-Commodity Flow]\label{thm:const-approx-mcmcflow}
    Let $G = (V, E, u, b)$ be a connected graph with edge capacities $u \geq 1$ and costs $b \geq 0$. Let $B \geq 0$ be the cost budget. Let $D : V \times V \rightarrow \mathbb{N}$ be a integral flow demand, $\tau \in [|D|]$ a integral required flow amount and $\eps \in (\log^{-c} N, 1)$ for some sufficiently small constant $c$ be a tradeoff parameter. Then, $\textsc{MCMCFlow}(G, B, D, \tau, \eps)$ returns a flow oracle $\mathcal{O}_F$ for a flow $F$ of value $\val(F) \geq \tau$ partially routing $D$ and a value $\lambda \geq 0$, such that
    \begin{itemize}
        \item Feasibility: the flow $\lambda F$ is capacity-respecting with cost at most $B$, and
        \item Approximation factor $1 / \exp(\poly(1 / \eps))$: let $\lambda^{*}$ be the maximum value such that there exists a flow $F^{*}$ of value $\val(F^{*}) \geq \tau$ partially routing $D$ where $\lambda^{*} F^{*}$ is capacity-respecting with cost at most $B$. Then, $\lambda \geq \frac{1}{\exp(\poly(1 / \eps))} \lambda^{*}$.
    \end{itemize}
    The algorithm has work and the flow oracle query work $(|E| + \supp(D)) N^{\poly(\eps)}$, and the algorithm has depth and the flow oracle query depth $N^{\poly(\eps)}$. The produced flow $F$ has support size $(|E| + \supp(D)) N^{\poly(\eps)}$.
\end{theorem}

\begin{algorithm}[H]
\caption{Constant-Approximate MCMC Flow: $\textsc{MCMCFlow}(G, B, D, \tau, \eps)$ \label{alg:mcmcflow}}

\textbf{Input:} Connected graph $G$ with edge capacities $u \geq 1$ and costs $b \geq 0$, integral demand $D$, integral required flow amount $\tau$, tradeoff parameter $\epsilon$.\\
\textbf{Output:} A flow oracle $\mathcal{O}_F$ for a flow $F$ of value $\mathrm{val}(F) \geq \tau$ partially routing $D$ and a value $\lambda \geq 0$, such that $\lambda F$ is capacity-respecting with cost at most $B$, and $\lambda \geq (1 / \exp(\poly(1 / \eps))) \lambda^{*}$ for the optimal $\lambda^{*}, F^{*}$ pair.
\begin{enumerate}
\item Let $\gamma = |D|$.
\item Let $\eps' = O(1)$ be small enough that $1 - O(\eps') \geq \frac{1}{2}$ in $\Cref{thm:boosting}$.
\item Let $\mathcal{O}_{F}$, $\lambda$ be the flow-scalar pair returned by \Cref{thm:boosting} on graph $G$ with capacities $\gamma u$, costs $b$, total cost budget $\gamma B$, demand $D$, required flow amount $\tau$, parameters $(\eps', \exp(\poly(1 / \eps)), N^{\poly(\eps)})$ and flow oracle $\textsc{UnboostedFlow}(G, D, \tau, \eps)$.
\item Return $\mathcal{O}_{F}$, $\lambda / \gamma$.
\end{enumerate}
\end{algorithm}

\begin{proof}[Proof of \Cref{thm:const-approx-mcmcflow}]
For a flow $F$, parameter $\lambda$ and value $\gamma \geq 0$, the following are equivalent:
\begin{itemize}
    \item $\lambda F$ is capacity-respecting with capacities $u$ and has cost at most $B$,
    \item $(\lambda \gamma) F$ is capacity-respecting with capacities $\gamma u$ and has cost at most $\gamma B$,
\end{itemize}
but scaling all capacities and the cost bound by $\gamma = |D|$ guarantees that there exists a capacity-respecting flow partially routing $D$ of value at least $\tau$, which is required by \Cref{thm:boosting}.

$\textsc{UnboostedFlow}(G, D, \tau, \eps)$ is a oracle function of the kind required for \Cref{thm:boosting} for length slack $s = \exp(\poly(1 / \eps))$ and congestion slack $\kappa = N^{\poly(\eps)}$. Note that $D$ is $\deg_{G'}$-respecting for graph $G'$ that is $G$ with capacities $\gamma u$ as $u \geq 1$.

The flow $F$ returned as a flow oracle $\mathcal{O}_F$ and value $\lambda$ returned by \Cref{thm:boosting} are guaranteed to satisfy the required properties: $F$ partially routes $D$ and has value at least $\tau$, $\lambda F$ is capacity-respecting with cost at most $B$, and for the maximum value $\lambda^*$ such that there exists a flow $F^{*}$ of value $\val(F^{*}) \geq \tau$ partially routing $D$ where $\lambda^{*} F^{*}$ is capacity-respecting with cost at most $B$, $\lambda \geq \frac{1 - O(\eps')}{s} \lambda^* \geq \frac{1}{2s} \lambda^{*} = \frac{1}{\exp(\poly(1 / \eps))} \lambda^{*}$.

The flow support size is at most the product of the support size $(|E| + |\supp(D)|) N^{\poly(\eps)}$ of flows returned by \Cref{thm:fast-unboosted-flow} and the number of oracle calls $O(\kappa \eps^{-2} \log^2 n) = N^{\poly(\eps)}$, giving the desired bound of $(|E| + |\supp(D)|) N^{\poly(\eps)}$.

For work and depth, \Cref{thm:boosting} makes $O(\kappa \eps'^{-2} \log^2 n) = N^{\poly(\eps)}$ calls to the oracle, which has work $(|E| + |\supp(D)|) N^{\poly(\eps)}$ and depth $N^{\poly(\eps)}$, and takes $\tilde{O}(\kappa \eps'^{-2} m) = |E| \cdot N^{\poly(\eps)}$ work and has depth $\tilde{O}(\kappa \eps'^{-2}) = N^{\poly(\eps)}$ outside the oracle calls. This gives total work $(|E| + |\supp(D)|) N^{\poly(\eps)}$ and depth $N^{\poly(\eps)}$ as desired. 
\end{proof}

\subsection{Constant-Approximate Concurrent/Non-Concurrent Flow}

Finally, we prove a formal version of \Cref{thm:informal-constapprox-cnc-flow}, first giving formal definitions for concurrent and non-concurrent flow.

\begin{definition}[Concurrent Flow Problem]\label{def:concurrent-flow}
    Let $G$ be a connected graph with edge capacities $u \geq 1$ and $D : V \times V \mapsto \mathbb{N}$ an integral demand. The \textit{concurrent flow problem} asks to find a capacity-respecting flow $F$ routing $\lambda D$ for maximum $\lambda$. An algorithm is a $C$-approximation for the concurrent flow problem if it always produces a capacity-respecting flow $F$ routing $\lambda D$, where $\lambda \geq C \lambda^{*}$ and $\lambda^{*}$ is the maximum value for which there exists a capacity-respecting $F^{*}$ routing $\lambda^{*} D$.

    In the concurrent flow problem with costs, each edge has a cost $b \geq 0$ and there is a total cost budget $B$: the produced flow $F$ must additionally have total cost $\sum_{P} F(P) \sum_{e \in P} b(e) \leq B$. An algorithm is a $C$-approximation for the concurrent flow problem with costs if it always produces a capacity-respecting flow $F$ of total cost at most $B$ routing $\lambda D$, where $\lambda \geq C \lambda^{*}$ and $\lambda^{*}$ is the maximum value for which there exists a capacity-respecting $F^{*}$ of total cost at most $B$ routing $\lambda^{*} D$.
\end{definition}

\begin{definition}[Non-Concurrent Flow Problem]\label{def:nonconcurrent-flow}
    Let $G$ be a connected graph with edge capacities $u \geq 1$ and $S$ a set of vertex pairs. The \textit{non-concurrent flow problem} asks to find a capacity-respecting flow $F$ routing flow between vertex pairs in $S$, i.e. $\supp(D_{F}) \subseteq S$ of maximum value. An algorithm is a $C$-approximation for the non-concurrent flow problem if it always produces a capacity-respecting flow $F$ routing flow between vertex pairs in $S$ of value $\val(F) \geq C \val(F^{*})$, where $F^{*}$ is the maximum-value capacity-respecting flow routing flow between vertex pairs in $S$.

    In the non-concurrent flow problem with costs, each edge has a cost $b \geq 0$ and there is a total cost budget $B$: the produced flow $F$ must additionally have total cost $\sum_{P} F(P) \sum_{e \in P} b(e) \leq B$. An algorithm is a $C$-approximation for the concurrent flow problem with costs if it always produces a capacity-respecting flow $F$ of total cost at most $B$ routing flow between vertex pairs in $S$ of value $\val(F) \geq C \val(F^{*})$, where $F^{*}$ is the maximum-value capacity-respecting flow of total cost at most $B$ routing flow between vertex pairs in $S$.
\end{definition}

\begin{theorem}[Constant Approximate Concurrent/Non-Concurrent Flow]\label{thm:constapprox-conc-nonconc-flow}
    For every tradeoff parameter $\eps \in (\log^{-c} N, 1)$ for some sufficiently small constant $c$, for both concurrent and non-concurrent multi-commodity flow with costs, there exists a $(m + k)^{1 + \poly(\eps)}$-work $(m + k)^{\poly(\eps)}$-depth $O(2^{-1 / \eps})$-approximate algorithm that returns a flow oracle $\mathcal{O}_F$ for the flow $F$.
\end{theorem}

\begin{proof} We can select $\epsilon' = \poly(\epsilon)$ such that $\exp(\poly(1 / \eps'))^{-1} = O(2^{-1 / \eps})$. Both the concurrent flow and non-concurrent flow algorithm are direct consequences of applying \Cref{thm:const-approx-mcmcflow} with this $\epsilon'$ and differing $D, \tau$ and returning $\lambda \mathcal{O}_F$:
    \begin{itemize}
        \item Concurrent Flow: apply \Cref{thm:const-approx-mcmcflow} with $\tau = |D|$, return $\lambda \mathcal{O}_F$.
        \item Non-Concurrent Flow: apply \Cref{thm:const-approx-mcmcflow} with $D(a, b) = \mathbb{I}[(a, b) \in S]$, $\tau = 1$, return $\lambda \mathcal{O}_F$.
    \end{itemize}
\end{proof}

\bibliographystyle{alpha}
\bibliography{arxivver/ref,arxivver/import-global-lc-expander-bib}

\appendix

\section{Derivation of \Cref{lem:lengthbound-lowcom-flow}}\label{sec:derivation-lengthbound-lowcom-flow}

In this appendix section, we give the derivation of \Cref{lem:lengthbound-lowcom-flow} from \Cref{thm:lengthbound-lowcom-cutmatch} (of \cite{haeupler2023length}), a low-support-size version of the flow algorithms of \cite{Haeupler2023lenboundflow}. The theorem gives an algorithm for efficiently computing a cutmatch, defined below.

\textbf{Cutmatch.} Given a graph $G = (V, E)$ with length $\ell$, an $h$-length $\phi$-sparse cutmatch of congestion $\gamma$ between disjoint node weighting pairs $\{(A_j, A_j)\}_{j \in [k]}$ consists of, for each $i$, a partition of the node-weightings $A_j = M_j + U_j$ and $A_j' = M_j' + U_j'$ where $M_j, M_j'$ and $U_j, U_j'$ are the "matched" and "unmatched" parts respectively and $|M_j| = |M_j'|$ and
\begin{itemize}
    \item A $h$-length flow $F = \sum_{j} F_j$ in $G$ with lengths $\ell$ of congestion $\gamma$, such that, for each $j \in [k]$, $F_j$ is a complete $M_j$ to $M_j'$ -flow.
    \item A $h$-length moving cut $C$ in $G$, such that for all $j \in [k]$, $\supp(U_j)$ and $\supp(U_j')$ are at least $h$-far in $G - C$, and $C$ has size at most
    \begin{equation*}
        |C| \leq \phi \cdot \left(\left(\sum_{j} |A_j|\right) - \val(F)\right)
    \end{equation*}
\end{itemize}

\begin{theorem}[\cite{haeupler2023length}]\label{thm:lengthbound-lowcom-cutmatch}
Let $G = (V, E)$ be a graph on $m$ edges with edge lengths $\ell \geq 1$ and capacities $u \geq 1$. Then, for any $h \geq 1$, $\phi \leq 1$, there is an algorithm that, given node-weighting pairs $\{(A_j, A_j')\}_{j \in [k]}$, outputs a multi-commodity $h$-length $\phi$-sparse cutmatch $(F, C)$ of congestion $\gamma$ where $\gamma = \tilde{O}\left(\frac{1}{\phi}\right)$. This algorithm has depth $\tilde{O}(k \cdot \poly(h))$ and work $\tilde{O}(|(E(G)| + \sum_j |\supp(A_j + A_j')|) \cdot k \cdot \poly(h))$. Moreover, $|\supp(F)| \leq \tilde{O}(|E(G)| + k + \sum_{j} |\supp(A_j + A_j')|)$.
\end{theorem}

To obtain \Cref{lem:lengthbound-lowcom-flow}, we use the fact that routers have no sparse cuts; selecting $\phi$ and $h$ as twice the router's parameters, the multicommodity flow produced by the cutmatch must completely satisfy the demand.

\lengthboundlowcomflow*

\begin{proof}
We may assume without loss of generality that the supports of $A_j$ and $A_j'$ are disjoint for all $j$; if they are not, we can add a flow path of length $0$ from that vertex to itself of value equal to the minimum of the two. Now, let $(F, C)$ be a $2t$-length $\left(\frac{1}{2\gamma}\right)$-sparse cutmatch $(F, C)$ between $\{(A_j, A_j')\}_{j \in [k]}$. Then, for each $j \in [k]$, $F_j$ is a complete $A_j$ to $A_j'$ -flow. Thus, calling \Cref{thm:lengthbound-lowcom-cutmatch} with length $2t$ and sparsity $\left(\frac{1}{2 \gamma}\right)$ suffices.

We prove the claim. Assume the contrary; then, $\sum_j |U_j| = \sum_j |U_j'| > 0$. For each $j$, let $D_j$ be an arbitrary demand such that $\mathrm{load}(D_j) = U_j + U_j'$, and let $D = \sum_j D_j$. Then, $D$ is $A$-respecting, and there exists a $t$-step $\gamma$-congestion flow on $G$ satisfying $D$. Let $F^{*}$ be that flow. We have
\begin{equation*}
    |C| = \sum_{e \in G} C(e) u(e) \leq \left(\frac{1}{2 \gamma}\right) |D|,
\end{equation*}
but
\begin{itemize}
    \item for every path $P \in F^{*}$, $\sum_{e \in P} C(e) > \frac{1}{2}$, as $C$ is a length-$2t$ cut and every flow-path had its length increased by more than $t$, and
    \item $\frac{1}{\gamma} \sum_{P \in F^{*} : e \in P} F^{*}(P) \leq u(e)$, as $F^{*}$ has congestion $\gamma$,
\end{itemize}
thus
\begin{equation*}
    \left(\frac{1}{2 \gamma}\right) |D| \geq \sum_{e \in G} C(e) u(e) \geq \frac{1}{\gamma} \sum_{P \in F^{*}} F^{*}(P) \sum_{e \in P} C(e) > \frac{1}{2 \gamma} |F^{*}|,
\end{equation*}
a contradiction. Thus, the unmatched part of the cutmatch and the cut are empty.
\end{proof}

\end{document}